\documentclass{article}
\usepackage{kr}

\usepackage{times}
\usepackage{soul}
\usepackage{url}
\usepackage[hidelinks]{hyperref}
\usepackage[utf8]{inputenc}
\usepackage[small]{caption}
\usepackage{graphicx}
\usepackage{amsmath}
\usepackage{amsthm}
\usepackage{booktabs}
\usepackage{thm-restate}

\usepackage{macro}

\title{Cubing for Tuning}

\author{
Haoze Wu$^1$ \and
Clark Barrett$^2$ \and 
Nina Narodytska$^3$
\affiliations
$^1$ Amherst College
$^2$ Stanford University\\
$^3$ VMware by Broadcom\\
\emails
hwu@amherst.edu,
barrett@cs.stanford.edu,
n.narodytska@gmail.com
}

\begin{document}

\maketitle

%

\begin{abstract}

We are exploring the problem of building an automated reasoning procedure that adaptively tunes the high-level solving strategy for a given problem. There are two main distinctive characteristics of our approach: tuning is performed solely online, unlike the common use of tuning as an offline process; and tuning data comes exclusively from the given instance, so we do not rely on the availability of similar benchmarks and can work with unique challenging instances. Our approach builds on top of the divide-and-conquer paradigm that naturally serves partitioned sub-problems for an automated tuning algorithm to obtain a good solving strategy. We demonstrate performance improvement on two classes of important problems\,--\,SAT-solving and neural network verification\,--\,and show that our method can learn unconventional solving strategies in some cases.



\end{abstract}

\section{Introduction}\label{sec:intro}

Algorithmic tuning is a common practice among automated reasoning practitioners seeking to optimize a solver’s performance on a specific problem domain. Traditionally, tuning is regarded as an offline process, where a set of benchmarks is selected and evaluated against a set of candidate solving strategies, with the goal of learning a good solving strategy to be used for solving new problems. However, offline tuning may not always be feasible\,--\,for example, if the solver is used on a single, unique problem instance, or if the practitioner lacks the expertise to navigate the space of possible candidate strategies. Motivated by this observation, we ask the following question: 
\emph{can a solver find a good solving strategy for a given problem using the problem itself as a source of learning?}

Our key high-level idea is to view the given problem as a generator of representative sub-problems and use them to perform strategy tuning. To perform sub-problem generation, we take advantage of the powerful divide-and-conquer paradigm, which is the state-of-the-art approach for solving challenging problems in many domains. For example, in SAT-solving, a technique called cube-and-conquer (\cnc)~\cite{heule2011cube} operates by partitioning a challenging SAT problem into thousands of sub-problems with look-ahead techniques and solving the sub-problems in a parallel manner. We exploit the synergy between algorithmic strategy tuning and the cube-and-conquer approach to answer our main question. First, the formula is partitioned into a set of sub-formulas. Next, small subsets of suitable sub-formulas are selected for tuning and validating the solving strategies. Finally, the learned strategy is used for solving the remaining sub-problems. 

We present our approach, termed \sys (Tuning Algorithm Configuration Online), as a general solving procedure. We instantiate \sys for the solving of SAT formulas and neural network verification queries, and show that our approach consistently boosts the performance of a \cnc-based solving procedure on a range of challenging benchmarks. Additionally, on those benchmarks, our prototype compares favorably against state-of-the-art solvers~\cite{saoudi2024pl,wu2024marabou} in respective domains, contributing several unique solutions. A preview of our results on a single benchmark is shown in Figure~\ref{fig:cherry}. Each point corresponds to a single sub-problem (i.e., cube). We see that \sys discovered a parameter setting that reduces the number of conflicts on a small set of selected cubes (left figure); the conflict-reducing effect of \sys generalizes to the remaining unsolved sub-problems (middle figure); and this reduction in conflicts translates to a reduction in runtime (right figure). 

To summarize, the contribution of this paper includes:
\begin{itemize}
\item the proposal to learn solving strategies online solely from the given problem; 
\item a general methodology, \sys, that leverages formula partitioning algorithm to efficiently customize solver strategy for a given query on the fly;
\item instantiation of the methods in two automated reasoning settings\,--\,SAT-solving and neural network verification;
\item experiments that show the practical benefit of \sys; and
\item discussions of unexpected instance-specific strategies discovered by \sys.
\end{itemize}

\begin{figure*}[t]
    \centering
    \includegraphics[width=0.9\linewidth]{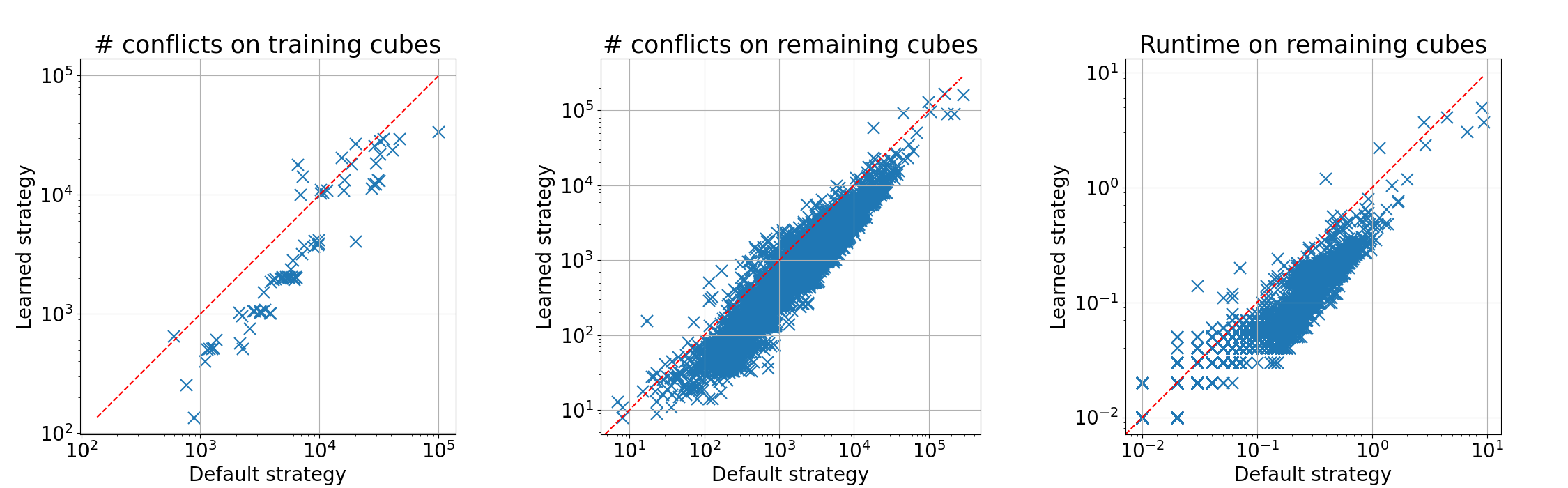}
    \vspace{-0.2cm}
    \caption{\sys learned a new strategy on a challenging SAT benchmark (\benchmark{eq.atree.braun.13}). 
    \cnc with \sys solved the benchmark in 139 seconds, while it took plain \cnc, \painless, and sequential \kissat 242, 1088, 5279 seconds, respectively. The benchmark deals with hardware multiplier equivalence checking and was found especially suitable to be tackled by \cnc~ \protect\cite{heule2011cube}.
    \label{fig:cherry}}
    \vspace{-0.3cm}
\end{figure*}

\section{Methodology}\label{sec:methodology}
This section describes our self-driven strategy learning approach.
Before we start, we briefly recall the \emph{cube-and-conquer} (\cnc) technique, as \sys is based on it. \cnc partitions the input formula \formula into a set of sub-formulas of the the form $\formula \land \cube$, where $\cube$, referred to as a \emph{cube}, is a conjunction of \emph{literals}. The goal of \cnc is to create a number of easier sub-problems that can be potentially processed in parallel.  For example, if \formula is a SAT formula and $p_1$, $p_2$ are propositional variables, then a cube \cube is $(p_1 \land p_2)$ and the corresponding sub-formula is $\set{\formula \land (p_1 \land p_2)}$. The set of obtained sub-formulas have the important property that their disjunction is equi-satisfiable with the original formula, i.e., \formula is satisfiable iff $\left(\bigvee_{\cube\in \cubes}\formula\land\cube\right)$ is satisfiable. In the example above, a valid partition over $p_1$ and $p_2$ would contain four cubes and the corresponding sub-formulas are: $\set{\formula \land (p_1 \land p_2), \formula \land (p_1 \land \neg p_2), \formula \land (\neg p_1 \land p_2), \formula \land (\neg p_1 \land \neg p_2)}$.
Partitioning strategies for different logical theories have been studied in the past in the context of divide-and-conquer-based SAT/SMT-solving~\cite{heule2011cube,hyvarinen2015search}. 
We use \cubes to denote the set of cubes produced by the partitioning strategy. 

\subsection{Overview}

Figure~\ref{fig:overview} provides a high-level overview of the proposed workflow, which consists of three main stages: cubing, strategy learning, and solving. The cubing and solving stages are standard,
while the strategy learning stage is new.

\paragraph{\textbf{Cubing.}} In the cubing stage, the input formula \formula is partitioned into a set of sub-formulas. For SAT-solving, a popular cubing technique is look-ahead~\cite{heule2009look}, which is intended for creating sub-problems that are both easier and \emph{balanced} in terms of difficulty.
Upon generating sub-formulas, a \cnc-based solver would typically proceed to the solving stage, where pre-determined solving strategies are used to solve each of the sub-formulas. In contrast, we observe that cubing opens up the possibilities for tuning the solver on the fly. Based on this observation, we introduce an online strategy learning phase. 

\paragraph{\textbf{Strategy learning.}}  The strategy learning stage consists of two phases: tuning and validation.  The tuning phase optimizes the solving strategy, while the validation phase assesses the performance of a proposed new strategy.

We start by describing the tuning phase (the \textbf{tuning} box in Figure~\ref{fig:overview}). 
To perform strategy learning, first, we need a set of cubes to evaluate different candidate solving strategies. Importantly, the cubes need to be \emph{sufficiently challenging}\,--\,to differentiate solving strategies, but also \emph{sufficiently easy}\,--\,to incur only small overhead when evaluated repeatedly. Cube difficulty is measured by a cost metric (e.g., time or number of conflicts it takes to solve the cube). We propose a procedure \algCollect to obtain such a set of cubes. We describe \algCollect in Sec.~\ref{subsec:collection}. At this stage, the \algCollect procedure outputs a set of cubes \cubesTune and passes them to the tuning procedure \funcTune. The goal of \funcTune is to optimize the solving strategy for \cubesTune by searching over a strategy space \strategyspace with respect to a cost metric. We focus on cases where \strategyspace is formed using a set of configurable parameters for the solving procedure. For example, Tab.~\ref{tab:strategyspace-kissat} shows the parameters that are used to form the strategy space for a SAT solver. As a result of this phase, we obtain an optimized strategy \strategy. We describe \funcTune in Sec.~\ref{subsec:strategy}.

Next, we discuss the validation phase (the \textbf{validation} box in Figure~\ref{fig:overview}). In this phase, we collect a different subset of cubes \cubesVal with \algCollect and use \cubesVal to validate the learned strategy \strategy. Depending on the validation results, \funcValidate either accepts the strategy \strategy found in the tuning phase as the final solving strategy \strategyFinal, or rejects \strategy and sets \strategyFinal to a fallback strategy. We chose to fallback to a sequential portfolio strategy that includes \strategy and the default strategy, but other plausible choices exist, as discussed in Sec.~\ref{subsec:taco}.

\paragraph{\textbf{Solving.}}  In the final solving stage, currently untackled cubes \cubes are solved using the learned strategy \strategyFinal. If all the sub-formulas are unsatisfiable, then the original formula is unsatisfiable; and if one of the sub-formulas is satisfiable, then the original formula is satisfiable.

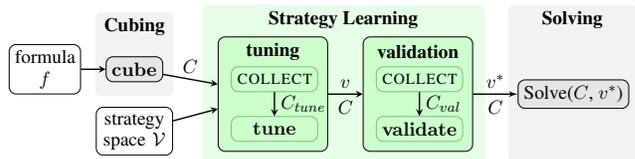
\begin{figure}[t]
\centering
\resizebox{\linewidth}{!}{%
\begin{tikzpicture}[node distance=0.5cm]

\tikzstyle{inputf} = [rectangle, rounded corners, 
text centered, 
draw=black, 
fill=white!10]

\tikzstyle{stepthrough} = [rectangle, rounded corners, 
minimum width=1.1cm, 
text centered, 
draw=black, 
fill=gray!20]

\tikzstyle{stepthroughG} = [rectangle, rounded corners, 
minimum width=2.2cm, 
minimum height=2.3cm,
text centered, 
draw=black,
fill=green!20]

\tikzstyle{stepthroughGSmall} = [rectangle, rounded corners, 
minimum width=1.3cm, 
minimum height=0.5cm,
text centered, 
draw=green!50!black!50, 
fill=green!20]

\tikzstyle{stepthroughBN} = [rectangle, rounded corners, 
minimum width=1cm, 
minimum height=0.7cm,
text centered, 
draw=black, 
fill=blue!20]

\tikzstyle{blockA} = [rectangle,
 rounded corners, 
minimum width=2.6cm, 
minimum height=3.2cm, 
text centered, 
draw=white]

\tikzstyle{blockC} = [rectangle,
 rounded corners, 
minimum width=1.6cm, 
minimum height=1.7cm, 
text centered, 
draw=white]

\tikzstyle{blockB} = [rectangle,
 rounded corners, 
minimum width=5.7cm, 
minimum height=3.2cm, 
text centered, 
text width=4.7cm,
draw=white]

\node (formula) [inputf, font=\large, text width=1.2cm, yshift=0.5cm] {formula \formula};
\node (cubing) [stepthrough, font=\large, right=of formula, xshift=0cm, yshift=0cm] {$\funcCube$};
\node (learn) [stepthroughG, font=\large, right=of cubing, xshift=0.6cm, yshift=-0.5cm, text width=1.9cm] {};
\node (learnCollect) [stepthroughGSmall, font=\large, right=of learn, xshift=-2.45cm, yshift=0.3cm, text width=1.5cm] {\algCollect};
\node (learnTune) [stepthroughGSmall,font=\large, right=of learn, xshift=-2.45cm, yshift=-0.7cm, text width=1.5cm] {\funcTune};
\node (validate) [stepthroughG, font=\large, right=of learn, xshift=0.2cm, yshift=0cm, text width=1.6cm] {};
\node (valCollect) [stepthroughGSmall, font=\large, right=of validate, xshift=-2.45cm, yshift=0.3cm, text width=1.5cm] {\algCollect};
\node (valTune) [stepthroughGSmall,font=\large, right=of validate, xshift=-2.45cm, yshift=-0.7cm, text width=1.5cm] {\funcValidate};

\node (solve) [stepthrough, font=\large,right=of validate, xshift=0.4cm, yshift=0cm] {{Solve(\cubes, \strategyFinal)}};
\node (spaceV) [inputf, font=\large, yshift=0cm,   xshift=0cm, below=of cubing, text width=1.3cm] {strategy space \strategyspace};

\begin{scope}[on background layer]
\node (cubingBG) [blockC,  yshift= -0.2cm, fill=gray!10, above of=cubing]{}; 
\end{scope}
\begin{scope}[on background layer]
\node (learningBG) [blockB, font=\large,  right of=cubingBG,  xshift= 3.7cm, yshift=-0.6cm, fill=green!10]{};
\end{scope}
\begin{scope}[on background layer]
\node (solvingBG) [blockA, font=\large,  right of=cubingBG,  xshift= 8.3cm, yshift=-0.6cm, fill=gray!10]{};
\end{scope}

\draw [arrow] (formula) -- (cubing);
\draw [arrow] (cubing) --  node [above] {\large{\cubes}}  (learn) ;
\draw [arrow] (learn) -- node [above] {\large{\strategy}}  node [below] {\large{\cubes}} (validate);
\draw [arrow] (spaceV) --  (learn);

\draw [arrow] (learnCollect) --  node [right] {\large{\cubesTune}} (learnTune);
\draw [arrow] (valCollect) --  node [right] {\large{\cubesVal}} (valTune);

\tikzstyle{line}     = [draw, -latex']

\path [line, thick] 
            ([yshift=0cm]validate.east) --  node [above] {\large{\strategyFinal}}  node [below] {\large{\cubes}}([yshift=0cm]solve.west);


\node at ($(cubingBG.north) + (0,-0.3)$) {\textbf{\large Cubing}};
\node at ($(learningBG.north) + (0,-0.3)$) {\textbf{\large Strategy Learning}};
\node at ($(solvingBG.north) + (0,-0.3)$) {\textbf{\large Solving}};
\node at ($(learn.north) + (0,-0.3)$) {\large{\textbf{tuning}}};
\node at ($(validate.north) + (0,-0.3)$) {\large{\textbf{validation}}};

\end{tikzpicture}

}
    \vspace{-0.4cm}
\caption{Overview of the \sys-based solving procedure.} \label{fig:overview}
    \vspace{-0.4cm}
\end{figure}

\subsection{Collecting Suitable Cubes}\label{subsec:collection}

Designing a robust method to collect a suitable set of cubes for strategy learning is not straightforward. Simply taking a random subset of the initial partition can run into a number of pitfalls. First, if all the sampled cubes are trivial (e.g., can be solved with very low cost), then they cannot effectively distinguish different solving strategies, as each strategy would have a similar (if not the same) cost. On the other hand, if the sampled cubes contain very challenging cubes, then solving them repeatedly with different solving strategies would incur a large overhead. We argue that an ideal cube collection procedure should take the cube difficulty (as defined by a cost range) into account and \emph{guarantee} to return a given number of cubes of the specified difficulties. Such cubes are informative enough to distinguish between candidate strategies yet efficient to evaluate. Alg.~\ref{alg:createCubes} describes such a cube collection procedure, referred to as \algCollect.

\begin{algorithm}[t]
\small
\begin{algorithmic}[1]
\State {\bfseries Input:} Formula \formula, initial partition \cubes, number of cubes \sampleTarget, set of strategies \strategyspace
\State {\bfseries Parameters:} number of cubes to create during re-partitioning \paramOnlineCubes,  minimal cost \paramMinCost, maximal cost \paramMaxCost
\State {\bfseries Output:} \tuple{\sat/\unsat/\unknown, \cubesCollected}; \cubesCollected  is a set of \sampleTarget cubes if the first output is \unknown, and empty otherwise. 
\Function{\algCollect}{$\formula, \cubes, \sampleTarget, \strategyspace$}
\State $\cubesCollected \assigned \varnothing$ \label{line:init-collect}
\While {$\abs{\cubesCollected} < \sampleTarget$}
\State $\cube \assigned \funcSample(\cubes)$ \label{line:collect-start}
\State $\strategy \assigned \funcSample(\strategyspace)$ \label{line:collect-strategy}
\State $\result \assigned \funcCheck(\formula\land \cube, \strategy, \paramMaxCost)$ \label{line:collect-solve}
\If {$\result = \sat$}  {\textbf{return} \tuple{\sat, \varnothing}} \label{line:collect-sat}
\ElsIf {$\result = \unsat$ and $\funcEval(\formula\land\cube, \strategy) \geq \paramMinCost$} 
\State {$\cubesCollected\assigned \cubesCollected\cup \set{\cube}$} \label{line:collect-good}
\ElsIf {$\result = \unknown$} 
\State {$\cubes \assigned \cubes \cup \set{\cube' \land \cube \,|\, \cube' \in \funcCube(\formula \land \cube, \paramOnlineCubes)}$} \label{line:collect-unknown}
\EndIf
\State {$\cubes.remove(\cube)$}
\If {$\abs{\cubes} = 0$}  {\textbf{return} \tuple{\unsat, \varnothing}} \label{line:collect-solved}
\EndIf
\EndWhile
\State {\textbf{return} \tuple{\unknown, \cubesCollected}}
\EndFunction
\end{algorithmic}
\caption{Identifying a set of suitable cubes}\label{alg:createCubes}
\end{algorithm}

The \algCollect procedure takes as input \tuple{\formula, \cubes, \sampleTarget, \strategyspace}, where \formula is a formula, \cubes is the current set of unsolved cubes, $\sampleTarget\in \nat$ is the target number of cubes, and \strategyspace is a set of strategies. The procedure makes use of the following sub-procedures: 1) $\funcCheck(\formula, \strategy, \maxCost)$ solves the formula \formula with solving strategy \strategy and cost budget \maxCost; the method returns \sat/\unsat if the \formula is solved, and \unknown otherwise; 2) $\funcEval(\formula, \strategy)$ returns the cost of solving the formula \formula with solving strategy \strategy; and 3) $\funcCube(\formula, \onlineCubes)$ partitions the formula \formula into $\onlineCubes$ cubes ($k > 1$). We assume the \funcCube procedure is \emph{sound}, i.e., \formula and $(\bigvee_{\cube \in\funcCube(\formula, k)} \formula\land\cube)$ are equi-satisfiable. The \algCollect procedure is also parameterized by \paramOnlineCubes--the number of cubes to create during re-partitioning ($\paramOnlineCubes > 1$), and \paramMinCost/\paramMaxCost--the minimal/maximal cost.
The procedure gradually grows a set of suitable cubes \cubesCollected (initialized at line \ref{line:init-collect} as an empty set) and terminates either when sufficient cubes that fall in the desired cost range have been collected (i.e., $\abs{\cubesCollected} = \sampleTarget$), or when the original formula is solved in this process.

Concretely, \algCollect repeatedly removes an unsolved cube \cube from \cubes (line~\ref{line:collect-start}), selects a strategy $\strategy\in\strategyspace$ (line~\ref{line:collect-strategy}), and attempts to solve the cube using the strategy with cost budget \paramMaxCost (line~\ref{line:collect-solve}). If the \funcCheck method returns \sat, then the process terminates and returns \tuple{\sat, \varnothing}, concluding that the original formula is satisfiable (line~\ref{line:collect-sat}); if the \funcCheck method returns \unsat and the cost of the strategy \strategy on solving the cube \cube is larger than the minimal cost \paramMinCost, then the cube \cube is added to the collection of suitable cubes \cubesCollected (line~\ref{line:collect-good}); if \cube is not solved within the cost budget, then \cube is further partitioned into smaller cubes using the function \funcCube (line~\ref{line:collect-unknown}) and those cubes are added back to \cubes, which maintains the current set of unsolved cubes. If there is no unsolved cube left, then the procedure returns \tuple{\unsat, \varnothing}, concluding that the original formula is unsatisfiable (line~\ref{line:collect-solved}). Finally, if the procedure exited the while loop without terminating, the procedure returns \tuple{\unknown, \cubesCollected}.

The following theorem characterizes the functionality of Alg.~\ref{alg:createCubes} and holds by construction:
\begin{theorem}\label{prop:functional-collect}
If $\algCollect(\formula, \cubes, \sampleTarget, \strategyspace)$ returns \tuple{\unknown, \cubesCollected}, then $\abs{\cubesCollected} = \sampleTarget$ and for each cube \cube in \cubesCollected, there is some solving strategy $\strategy\in\strategyspace$ such that $\paramMinCost \leq \funcEval(\formula\land\cube, \strategy) \leq \paramMaxCost$.
\end{theorem}

The following theorem characterizes the soundness of Alg.~\ref{alg:createCubes} and holds because re-partitioning maintains equi-satisfiability:
\begin{theorem}\label{prop:sound-collect}
 Given a set of cubes \cubes such that $(\bigvee_{\cube \in\cubes} \formula\land\cube)$ and $\formula$ are equi-satisfiable, and a sound \funcCube procedure, the formula \formula is satisfiable if $\algCollect(\formula, \cubes, \sampleTarget, \strategyspace)$ returns \tuple{\sat,\varnothing}, and unsatisfiable if it returns \tuple{\unsat,\varnothing}.   
\end{theorem}

To summarize, given a partitioning of the formula \formula described by a set of cubes \cubes, \algCollect either solves the formula or returns a 
given number of cubes of specified difficulties. The former case occurs either when an easy \sat case is encountered or when all cubes are easily unsatisfiable. 

The termination of Alg.~\ref{alg:createCubes} is not always guaranteed, even if all of its sub-routines terminate. For example, if the \funcCheck method keeps returning \unknown on all cubes, then neither will $|\cubesCollected|$ grow, nor will $|\cubes|$ reduce to 0. Conceptually, termination relies on the choices of the strategy space \strategyspace, the cost metric, and the cubing procedure \funcCube. We specify a practical sufficient condition for the termination of Alg.~\ref{alg:createCubes}.
\begin{definition}[Cubing trace]
\label{def:trace}
A \emph{cubing trace} \trace starting from formula $\formula$ with respect to $\funcCube(\cdot,\paramOnlineCubes)$ is defined as a cube sequence $\cube_0,\;\cube_1,\;\dots$ where
$\cube_0 = \top$ and $\cube_{i+1} \in
  \funcCube(\formula \land \cube_{i},\paramOnlineCubes)$.
We say the trace \trace \emph{reaches} cube $\cube_i$ in $i$ steps.
\end{definition}

\begin{restatable}[Termination via bounded cost reduction]{theorem}{terminationalt}\label{prop:termination-alt}
\noindent
Assume that all sub-routines of \algCollect terminate.
And assume that the cubing procedure $\funcCube(\cdot,\paramOnlineCubes)$ satisfies the following property with respect to a set of strategies \strategyspace, and a maximal cost \paramMaxCost:
\begin{description}
\item[(Bounded cost reduction)]%
For any formula \formula, every cubing trace starting from \formula with respect to $\funcCube(\cdot, \paramOnlineCubes)$ will reach, in a finite number of steps, a cube \cube such that
\[\max_{\strategy\in\strategyspace}
\funcEval(\formula\!\land\!\cube,\strategy) <\paramMaxCost.
\]
\end{description}
\noindent
Then $\algCollect(\formula, \cubes, \sampleTarget, \strategyspace)$ terminates.
\end{restatable}
\begin{proof}
App.~\ref{app:proofs}.
\end{proof}

Intuitively, the bounded cost reduction assumption requires that \funcCube must lead to sub-formulas that are sufficiently easy  to solve by any of the solving strategies in \strategyspace.

\paragraph{\textbf{Efficiently implementing Alg.~\ref{alg:createCubes}}.} There are a few things to note when it comes to implementing Alg.~\ref{alg:createCubes}. First, the procedure is highly parallelizable: the input cubes \cubes can be viewed as a dynamically shrinking and growing work queue to be processed in parallel. Second, eagerly invoking the \funcCube method (line \ref{line:collect-unknown}) can unnecessarily incur time and memory overhead when there are still unexamined cubes in \cubes. Therefore, one could lazily re-partition \unknown cubes only when all the cubes in \cubes have been examined, but not enough cubes have been collected. Finally, to further reduce the number of \funcCube calls, one could consider adaptively increasing (e.g., multiplying by some factor $r$) the value of \paramMaxCost for re-partitioned cubes. Thus, the cost budget for any cube \cube is $\paramMaxCost \cdot r^{n}$, where $n$ is the number of re-splits that leads to \cube ($n=0$ for original cubes). With this approach, Thms.~\ref{prop:sound-collect} and ~\ref{prop:termination-alt} still holds, and Thm.~\ref{prop:functional-collect} holds by changing \paramMaxCost to $\paramMaxCost \cdot r^{n}$.

\subsection{Extending Cube-and-Conquer with Online Strategy Learning}\label{subsec:taco}

Alg.~\ref{alg:sdac} introduces an extension of the traditional \cnc solving procedure with an online learning phase (\algSDAC). The goal is to dynamically optimize solving strategies based on fractions of the original formula.

The \algSDAC procedure takes as input a formula \formula and a strategy space \strategyspace (with a default strategy $\strategy_0$), and returns either \sat or \unsat. Apart from the cube collection procedure \algCollect, \algSDAC also depends on the following sub-procedures: 1) $\funcCheck$ is the same as described in the previous section; we assume that when the cost budget \paramMaxCost is unlimited ($\infty$), \funcCheck returns either \sat or \unsat; 2) $\funcTune(\formulas, \strategyspace)$ optimizes the solving strategy with respect to a set of formulas \formulas over a strategy space \strategyspace; and 3) $\funcValidate(\formulas, \strategy, \strategy_0)$ compares the performance (i.e., total cost) of a new strategy \strategy with that of the default strategy $\strategy_0$ on a set of formulas \formulas and returns the final strategy used for solving the unsolved cubes.

\begin{algorithm}[t]
\small
\begin{algorithmic}[1]
\State {\bfseries Input:} Formula \formula, strategy space \strategyspace (with a default strategy $\strategy_0$)
\State {\bfseries Parameters:} number of cubes to create initially \paramInitialCubes, number of cubes for tuning \paramSampleTarget, number of cubes for validation \paramSampleTargetOracleTraining
\State {\bfseries Output:} \sat or \unsat
\Function{\algSDAC}{$\formula, \strategyspace$}
\State {$\cubes \assigned \funcCube(f, \paramInitialCubes)$} \label{line:sdac-init} \Comment{Initial partition}
\State {$\strategyFinal \assigned \strategy_0$}
\State {$res, \cubesTune \assigned \algCollect(\formula, \cubes, \paramSampleTarget, \strategyspace)$} \tikzmark{phase1top} \label{line:sdac-collect-tune}
\If {$res \in \set{\sat, \unsat}$} {{\textbf{return}} res}
\EndIf
\State {$\strategy \assigned \funcTune(\set{\formula \land \cube \,|\, \cube\in \cubesTune}, \strategyspace)$} \tikzmark{phase1bot} \label{line:sdac-tune}
\If {$\strategy \neq \strategy_0$}
\State {$res, \cubesVal \assigned \algCollect(\formula, \cubes, \paramSampleTargetOracleTraining, \set{\strategy})$} \tikzmark{phase2top}   \label{line:sdac-collect-val}
\If {$res \in \set{\sat, \unsat}$} {{\textbf{return}} res} \quad \tikzmark{phaseleft}
\EndIf
\State {$\strategyFinal \assigned \funcValidate(\set{\formula \land \cube \,|\, \cube\in \cubesVal}, \strategy, \strategy_0)$} \tikzmark{phase2bot} \label{line:sdac-val}  \tikzmark{phase2bot}
\EndIf
\For {each $\cube \in \cubes$} 
\If {$\funcCheck(\formula\land \cube, \strategyFinal, \infty) = \sat$}  {\textbf{return} \sat}
\EndIf
\EndFor
\State {\textbf{return} \unsat}
\EndFunction
\end{algorithmic}
\caption{Cube-and-conquer with self-driven algorithmic configuration}\label{alg:sdac}
\AddNote{phase1top}{phase2bot}{phaseleft}{Strategy Learning}
\end{algorithm}

We now describe the execution of \algSDAC. The procedure first creates an initial partition \cubes like traditional \cnc, and then enters the strategy learning phase. During this phase, \algSDAC first attempts to collect \paramSampleTarget cubes from \cubes for tuning (line~\ref{line:sdac-collect-tune}). Importantly, we assume \cubes is \emph{passed by reference}, meaning that after \algCollect terminates, \cubes contains the latest set of cubes yet to be solved. If the original formula is not solved by \algCollect, the \funcTune method is invoked on the set of formulas
\formulas=\set{\formula \land \cube \,|\, \cube\in \cubesTune} to
 optimize the solving strategy (line~\ref{line:sdac-tune}). If a better strategy for \cubesTune is not found, the strategy learning phase terminates and the remaining cubes are solved with the default strategy $\strategy_0$. Otherwise, a different set of \paramSampleTargetOracleTraining cubes, denoted \cubesVal, is collected to validate the new strategy \strategy (line~\ref{line:sdac-collect-val}). Similarly, the original formula might be solved during this second round of cube collection. If this is not the case, the \funcValidate method is invoked to decide on the final solving strategy \strategyFinal (line~\ref{line:sdac-val}), which is then used to solve the remaining set of cubes. 

\paragraph{\textbf{Design choice: strategy validation.}} We instantiated the \funcValidate method as the following: if the new strategy $\strategy$ performs better on \cubesVal, then return \strategy, otherwise, return $\function{sequentialPortfolio}(\strategy, \paramMaxCost, \strategy_0)$.\footnote{This stands for solving the cube using \strategy with a cost budget of \paramMaxCost, and fallback to the default strategy $\strategy_0$.} Alternatively, one could simply fall back to $\strategy_0$ instead of the sequential portfolio strategy, or even select a different candidate strategy by repeating the tuning phase with a larger \paramSampleTarget. 

\paragraph{\textbf{Design choice: strategies for collection.}} We would like to call attention to a nuance in the choice of the fourth input to \algCollect, which specifies the set of candidate strategies to pick from when examining a cube. For the tuning cubes \cubesTune, \algSDAC uses \strategyspace as this set of strategies (line~\ref{line:sdac-collect-tune}), which means it is \emph{unbiased} towards any strategies and the collected cubes fall in the suitable range for different strategies. Another plausible design choice is to always collect cubes with the default strategy $\set{\strategy_0}$. We did not opt for the latter since \cubesTune might be biased towards cubes that $\strategy_0$ already performs well on, leaving less room for improvement. On the other hand, when collecting the validation cubes \cubesVal, \algSDAC only uses the newly learned strategy $\strategy$ (line~\ref{line:sdac-collect-val}). This means \algSDAC is \emph{optimistic} and only rejects \strategy if the default strategy $\strategy_0$ performs better than \strategy on a set of cubes that \strategy performs well on.

\subsection{Tuning the solving strategy on the cubes}
\label{subsec:strategy}

We now describe our instantiation of the \funcTune procedure  (line~\ref{line:sdac-tune}, Alg.~\ref{alg:sdac}).
Given a set of collected cubes \cubesTune, the goal of \funcTune is to optimize the solving strategy in a given strategy space \strategyspace: 
\[
\underset{\strategy \in \mathcal{\strategyspace}} {\text{minimize}} \quad  \cost(\strategy) := \Sigma_{\cube \in \cubesTune} \funcEval(\formula \land \cube, \strategy).
\]
When \strategyspace is small, one could evaluate all strategies in \strategyspace. 
But this is unaffordable when $|\strategyspace|$ is large, especially in the online strategy learning setting. In this section, we consider the scenario where one stochastically optimizes the solving strategy while being only allowed to examine \paramNumMCMCSamples strategies in the strategy space \strategyspace, with $\paramNumMCMCSamples \ll \abs{\strategyspace}$. One possible approach is Markov-Chain Monte-Carlo (MCMC) sampling, which in our setting can be used to generate a sequence of solving strategies with the desirable property that the sequence converges to strategies with the lowest cost.

A widely-used MCMC method is the Metropolis-Hastings (M-H) Alg.~\cite{mh}, instantiated in the context of \sys as follows: 
\begin{enumerate}
    \item Choose a current strategy \strategy;
    \item Propose to replace the current strategy with a new one $\strategy'$, sampled from a \emph{proposal distribution} $q(\strategy' | \strategy)$;
    \item If $\cost(\strategy') \leq \cost(\strategy)$,
    accept $\strategy'$ as the current strategy;
    \item Otherwise, accept $\strategy'$ as the current strategy with some probability $a(\strategy{\rightarrow}\strategy')$ (e.g., a probability inversely proportional to the increase in the cost);
    \item Go to step 2.
\end{enumerate}
This process is repeated until \paramNumMCMCSamples samples are drawn. Intuitively, under this scheme, a better proposal is always accepted, while a proposal that increases the cost may still be accepted. In other words, the algorithm greedily navigates towards a better strategy whenever possible, but can also overcome local minima. In our implementation, the acceptance probability is computed using a standard method~\cite{mcmc}. First, $\cost(\strategy)$ is transformed into a probability distribution $p(\strategy) \propto \exp(-\beta \cdot \cost(\strategy))$,
where $\beta > 0$ is a configurable parameter. The acceptance probability is then computed as:
\begin{align*}
a(\strategy{\rightarrow}\strategy') &= \min\left(1, \frac{p\left(\strategy'\right)}{p\left(\strategy\right)}\right)  \\
&= \min\left(1, \exp\left(\beta \cdot \left(\cost(\strategy) - \cost(\strategy')\right)\right)\right).
\end{align*}
Under this acceptance probability, the larger the total cost increases going from the current strategy $\strategy$ to the proposed strategy $\strategy'$, the lower the probability of accepting $\strategy'$ as the new current strategy. On the other hand, the larger $\beta$ is, the more reluctant we are to move to a worse proposal.

\paragraph{\textbf{Design choice: proposal distribution.}} To ensure the aforementioned convergence property of MCMC, the proposal distribution must be both \emph{symmetric} and \emph{ergodic}.\footnote{
A proposal distribution $q$ is symmetric if $q(\strategy' | \strategy) = q(\strategy | \strategy')$ for any $\strategy, \strategy'\in\strategyspace$
and is ergodic if there is a non-zero probability of reaching a strategy $\strategy\in\strategyspace$ from any other strategy $\strategy'\in\strategyspace$ in a finite number of steps.
}
For discrete search spaces (which is the setting we consider in this paper), a common proposal distribution is the symmetric random walk, which moves to one of the \emph{neighbors} of the current sample at uniform random. This proposal distribution is both symmetric and ergodic. We define the neighbors of a strategy as all strategies for which $\leq k$ parameter values are different. We use $k=2$ in our implementation.

\paragraph{\textbf{Design choice: MCMC initialization.}} The better (i.e., low-cost) the initial strategy is, the shorter it takes for MCMC-sampling to converge~\cite{roy2020convergence}. While a natural choice of the initial strategy is the default strategy $\strategy_0$, we found that it is empirically advantageous to ``probe'' every $1$-neighbor of $\strategy_0$ and initialize the MCMC-sampling with the best strategy seen in this process.

The sampling scheme presented above is in the same spirit as approaches used in the automated configuration literature~\cite{hoos2021automated}, with emphasis on identifying low-cost strategies in a lightweight manner. Borrowing more insights from that literature or leveraging existing off-line tuning tools such as SMAC~\cite{lindauer-jmlr22a} are interesting directions for future work.
\section{ Instantiations of \sys}\label{sec:impl}

We implemented a prototype of the \sys-extended \cnc procedure in \textsc{Python3}. In addition to parallelizing the cube collection using the strategy described in Sec.~\ref{subsec:collection}, the prototype also parallelizes the tuning, validation, and solving processes, by solving the cubes in parallel. Our prototype takes as input a formula in file format and runs \cnc on the formula with or without \sys. 
The prototype supports solving two types of logical formulas: SAT formulas and neural network verification (NNV) queries. 

\paragraph{SAT-solving.} 
For SAT-solving, the prototype uses \marchcu~\cite{heule2011cube} as the cuber and \kissat~\cite{BiereFallerFazekasFleuryFroleyksPollitt-SAT-Competition-2024-solvers} as the solver, and takes formula in CNF format. \marchcu is a popular cubing tool used in previous \cnc work~\cite{heule2011cube,ozdemir2021sat,heisinger2020distributed}, and \kissat is a state-of-the-art SAT-solver. We use the number of conflicts as the cost metric, which has been used as a proxy for runtime in prior offline tuning work~\cite{beskyd2022domain}. We choose not to use runtime as the cost metric since it makes the execution non-deterministic. The strategy space consists of combinations of possible values of 9 \kissat command-line parameters that we believe have a significant impact on the branching decisions. We allow at most four parameters to deviate from the default value, resulting in a total of 349 unique parameter settings. Tab.~\ref{tab:strategyspace-kissat} shows the parameters and their candidate values. We evaluate the effect of using a different strategy space in Sec.~\ref{subsec:ablation}.

\begin{table}[t]
\centering
\scriptsize
\vspace{-0.2cm}
\caption{Parameters in the strategy space for {\kissat}} 
\vspace{-0.2cm}
\begin{tabular}{ccc|ccc}
\toprule
Parameter & default & alts & Parameter & default & alts\\
\midrule
\param{bump} & 1 & 0 & \param{phase} & 1 & 0\\
\param{bumpreasons} & 1 & 0 & \param{stable} & 1 & 0, 2 \\
\param{chrono} & 1 & 0 & \param{target} & 1 & 0 \\
\param{eliminate} & 1 & 0 & \param{tumble} & 1 & 0 \\
\param{forcephase} & 0 & 1
 \\ 
\bottomrule
\end{tabular}
\label{tab:strategyspace-kissat}
\vspace{-0.5cm}
\end{table}


\paragraph{Neural Network Verification.}  For NNV, the prototype uses the \marabou~\cite{katz2019marabou,wu2024marabou} verifier as both the cuber and the solver, and takes \marabou's input query format. \marabou is a state-of-the-art SMT-based neural network verification tool. We chose NNV with \marabou as the second case study because
\cnc-like approach has been shown to improve \marabou's performance~\cite{katz2019marabou,wu2020parallelization}. Since \marabou performs DPLL(T)-like search without conflict-driven clause learning, we use the number of decisions as the cost metric instead of the number of conflicts. We extended \marabou with a command line option that limits the number of decisions. Two parameters are chosen for tuning: \param{pl-split-freq} and \param{branch}. The former specifies the frequency of performing a case-split on internal neurons when input-splitting is enabled (e.g., 1 means always splitting on internal neurons; 5 means splitting on an internal neuron every 4 input-splittings). The latter determines the heuristics for selecting which internal neuron to branch on. The strategy space is described in Tab.~\ref{tab:strategyspace-marabou}. In total, $\strategyspace_{\marabou}$ contains 12 candidate solving strategies.

\begin{table}[ht]
\scriptsize
\vspace{-0.2cm}
\centering
\caption{Parameters in the strategy space for \marabou} 
\vspace{-0.2cm}
\begin{tabular}{ccc}
\toprule
Parameter & default & alts\\
\midrule
\param{pl-split-freq} & 10 & 1, 2, 5 \\
\param{branch} & pseudo-impact & babsr, polarity \\
\bottomrule
\end{tabular}
\label{tab:strategyspace-marabou}
\vspace{-0.2cm}
\end{table}

We prove in App.~\ref{app:proofs} that the bounded cost reduction assumption that guarantees termination (Thm.~\ref{prop:termination-alt}) can  be satisfied for both applications. 
The prototype can also be extended to handle other logical formulas by providing the corresponding implementations of the \funcCube, \funcCheck, and \funcEval methods. 
App.~\ref{app:impl} describes other implementation details.
\section{Experimental Evaluation}\label{sec:eval-sat}

We conduct a number of experiments in the two aforementioned automated reasoning applications to evaluate the proposed method. The main questions we investigate are:
\begin{enumerate}
\itemindent=-4pt
\item Can \sys learn new strategies that lead to performance gain? [Yes]
\item Do the new strategies offer new problem-solving insights? [In several cases]
\end{enumerate}
In addition, we perform ablation studies to understand the effect of various ingredients of \sys.

\subsection{Case study: SAT-Solving}

In SAT-solving, \cnc is recognized as a method for tackling challenging instances that cannot be efficiently solved by sequential solver or a portfolio strategy. 
To study the effectiveness of \sys, we focus on benchmarks studied in existing \cnc literature, which were known to be challenging for state-of-the-art SAT-solvers: 
\begin{itemize}
\item \cncBench: benchmarks studied in the original \cnc work~\cite{heule2011cube}; 
\item \crux: an arithmetic verification benchmark family~\cite{ritirc2017column}
that was studied in a distributed \cnc work \solver{Paracooba}~\cite{heisinger2020distributed}.
\end{itemize}
In addition, to investigate the effectiveness of \sys on the latest hard problems, we propose to evaluate \cnc-based techniques on a new set of benchmarks that consists of unsolved problems from recent SAT competitions:
\begin{itemize} 
\item \scBench: benchmarks from recent SAT Competitions (22--24) that were not solved by any participating tools.
\end{itemize}

We perform a direct comparison between \cnc with and without \sys. We use \marchcu to partition a given benchmark with maximal cube depth 15. This means at most 32768 cubes are generated, a number on par with what was used in prior work~\cite{heisinger2020distributed}. During the solving phase, both configurations solve the cubes in the same order (modulo cubes that were already solved by \sys); this ensures fairness of runtime comparison, especially when it comes to satisfiable benchmarks. 
The number of cubes collected for tuning (\paramSampleTarget) is 50; the number of cubes collected for validation (\paramSampleTargetOracleTraining) is 25; the cost range (\paramMinCost and \paramMaxCost) of tuning cubes is between 500 and 10000, and that of validation cubes is between 10000 to 50000; for re-partitioning during \algCollect, we use \marchcu to create a maximal number of 64 new cubes. Since re-partitioning is performed lazily, it only occurred on two of the evaluated benchmarks. 
The number of MCMC samples (\paramNumMCMCSamples) is 20. 

In addition, we also compare the performance of our prototype with the state-of-the-art parallel SAT solver \painless, which won the parallel track of SAT Competition 2024~\cite{saoudi2024pl}. 

Experiments were performed on a cluster equipped with Dell PowerEdge R6525 CPU servers featuring 2.6-GHz AMD CPU cores. All three configurations were given 24 cores (spawning 23 workers) and 192GB of memory. For \cncBench and \crux benchmarks, we used a wall-clock timeout of 12 hours. For \scBench benchmarks, we used a wall-clock timeout of 1 hour. 
App.~\ref{app:exp-setup} provides additional details about the experimental setup.

\begin{table}[t!]
\setlength\tabcolsep{1.5pt}
\centering
\caption{Results on the \cncBench and \crux benchmarks. \coltimecubing, \coltimelearning, and \coltimesolving denotes cubing, strategy learning, and solving time; the total runtime is $\coltimecubing + \coltimesolving$ for \cnc, and $\coltimecubing + \coltimesolvinglearning$ for \cncTuneAndVal. \colscorecommon denotes the total conflicts on cubes commonly solved by \cnc and \cncTuneAndVal using different solving strategies. The asterisk ($^*$) near a benchmark means re-partitioning happened; for other benchmarks, \cnc and \cncTuneAndVal solved the same set of cubes. Finally, the last column (\coltimetotal) shows the runtime of \painless.} 
\vspace{-0.2cm}
\scriptsize
\label{tab:cnc-crux}
\begin{tabular}{lcccccccc}
\toprule\vspace{-0.3cm}\\
\cmidrule(lr){3-8}
\cmidrule(lr){9-9}
 & & & \multicolumn{2}{c}{\tabtitle{\cnc}} & \multicolumn{3}{c}{\tabtitle{\cncTuneAndVal}} & \tabtitle{\gray{\painless}} \\
\cmidrule(lr){4-5}\cmidrule(lr){6-8}\cmidrule(lr){9-9}
\tabtitle{Benchmark} & \tabtitle{Res.} & \tabtitle{\coltimecubing} & \tabtitle{\coltimesolving} & \tabtitle{\colscorecommon} & \tabtitle{\coltimesolvinglearning} & \tabtitle{\colscorecommon} & \tabtitle{\coltimelearning} & \tabtitle{\gray{\coltimetotal}} \\
9dlx\_vliw\_at\_b\_iq8 & UNS & 406.5 & 20601 & 1.157e+9 & \best{17409} & \best{4.412e+8} & 781 & \gray{182} \\
9dlx\_vliw\_at\_b\_iq9 & UNS & 627.0 & 26135 & 1.424e+9 & \best{20457} & \best{4.983e+8} & 978 & \gray{275} \\
AProVE07-25 & UNS & 3.9 & 213 & 6.954e+7 & \best{159} & \best{5.805e+7} & 24  & \gray{486} \\
dated-5-19-u* & UNS & 75.0 & 268 & 4.881e+6 & \best{192} & \best{3.297e+6} & 55 & \gray{522} \\
eq.atree.braun.12 & UNS & 1.0 & 76 & 4.245e+7 & \best{60} & \best{1.590e+7} & 60 & \gray{282} \\
eq.atree.braun.13 & UNS & 1.1 & 242 & 1.553e+8 & \best{139} & \best{6.145e+7} & 52  & \gray{1088} \\
gss-24-s100 & SAT & 14.6 & 1652 & 1.978e+8 & \best{1314} & \best{1.657e+8} & 98 & \gray{272} \\
gss-26-s100 & SAT & 8.4 & 11933 & 2.230e+9 & \best{8585} & \best{1.815e+9} & 185 & \gray{578} \\
gus-md5-14 & UNS & 157.8 & 9607 & 1.879e+8 & \best{6970} & \best{1.724e+8} & 408 & \gray{41486} \\
ndhf\_xits\_09\_UNS & UNS & 16.5 & 535 & 1.094e+8 & \best{501} & \best{9.528e+7} & 43 & \gray{820} \\
rbcl\_xits\_09\_UNK & UNS & 6.0 & 276 & 1.087e+8 & \best{268} & \best{9.724e+7} & 26 & \gray{684} \\
rpoc\_xits\_09\_UNS & UNS & 6.9 & \best{334} & 1.032e+8 & 385 & \best{9.827e+7} & 33 & \gray{273} \\
total-10-17-u* & UNS & 166.0 & 488 & 2.800e+5 & \best{238} & \best{2.687e+5} & 158 & \gray{268} \\
total-5-15-u & UNS & 101.6 & 4683 & 2.294e+8 & \best{2303} & \best{1.739e+8} & 56 & \gray{2075} \\
\midrule
cruxmiter32seed0 & UNS & 0.2 & 114 & 1.355e+8 & \best{100} & \best{1.297e+8} & 9 & \gray{329} \\
cruxmiter32seed1 & UNS & 0.2 & 113 & 1.355e+8 & \best{101} & \best{1.297e+8} & 10 & \gray{299} \\
cruxmiter32seed2 & UNS & 0.2 & 171 & 1.769e+8 & \best{161} & \best{1.751e+8} & 10 & \gray{190} \\
cruxmiter32seed3 & UNS & 0.2 & 463 & 3.689e+8 & \best{283} & \best{3.562e+8} & 12 & \gray{603} \\
cruxmiter32seed4 & UNS & 0.2 & 746 & 5.048e+8 & \best{510} & \best{4.527e+8} & 12 & \gray{1149} \\
cruxmiter32seed5 & UNS & 0.2 & 173 & 2.096e+8 & \best{142} & \best{1.913e+8} & 11 & \gray{360} \\
cruxmiter32seed6 & UNS & 0.2 & 204 & 2.870e+8 & \best{161} & \best{2.691e+8} & 12 & \gray{453} \\
cruxmiter32seed7 & UNS & 0.2 & 291 & 3.230e+8 & \best{211} & \best{2.975e+8} & 11 & \gray{138} \\
cruxmiter32seed8 & UNS & 0.2 & 297 & 4.192e+8 & \best{288} & \best{4.191e+8} & 12 & \gray{854} \\
cruxmiter32seed9 & UNS & 0.2 & 530 & 5.119e+8 & \best{337} & \best{4.726e+8} & 10 & \gray{2549} \\
\bottomrule
\end{tabular}
\end{table}

\subsubsection{Results on the \cncBench and \crux benchmarks.} 

Tab.~\ref{tab:cnc-crux} compares the performance of \cnc and \cncTuneAndVal on the first two benchmark sets. We omitted one satisfiable \cncBench benchmark on which both \cnc and \cncTuneAndVal timed out. We report the time to create the initial cubes (\coltimecubing). For \cnc, we report the time to solve the cubes (\coltimesolving); for \cncTuneAndVal, we report the total time of the strategy learning phase and the solving phase (\coltimesolvinglearning). We also report the total number of conflicts on cubes that were commonly solved by \cnc and \cncTuneAndVal using different solving strategies (\colscorecommon). We report \painless's runtime in the last column.

Impressively, \sys was able to find a new solving strategy that reduces the total number of conflicts on all \cnc and \crux benchmarks. Moreover, the reduction in total cost consistently translates to a reduction in runtime on all but one case. The reduction is over 20\% in 15 out of the 24 benchmarks and in several cases exceeds 40\% (e.g., \benchmark{eq.atree.braun.13}).  Upon closer examination, on all but two instances, the new strategy found by tuning was validated and used to solve the remaining cubes. For the two remaining cases (\benchmark{rbcl\_xits\_09\_UNK} and \benchmark{rpoc\_xits\_09\_UNS}), \cncTuneAndVal fell back to a sequential portfolio strategy (as described in Sec.~\ref{subsec:strategy}). 

We notice that while the strategy-learning phase (\coltimelearning) typically constitutes a small fraction of the total solving time, this is not always the case. In particular, for \benchmark{eq.atree.braun.12}, \cncTuneAndVal solved all the cubes during the strategy learning phase. Upon closer examination, we found that \sys successfully identified a new solving strategy, and in the process of collecting validation cubes, all the cubes were solved using this new strategy before sufficient validation cubes were collected. This highlights the fact that the \algCollect method does contribute to solving the cubes, and its overhead mainly comes from unsolved cubes (line \ref{line:collect-unknown}, Alg.~\ref{alg:createCubes}). Usually, we observe that only a small fraction of cubes fall into this category.

Another interesting result is that \cncTuneAndVal and \painless exhibit complementary performance. \cncTuneAndVal outperformed \painless  on 7 out of the 14 \cncBench benchmarks and dominated \painless on all \crux benchmarks. This suggests \cnc is competitive to state-of-the-art solver \painless on benchmark sets considered in \cnc literature.


App.~\ref{app:strategy} lists the solving strategies found by \sys on all benchmarks. We highlight a few interesting observations:
\begin{itemize}
    \item For 7 of the 12 unsatisfiable \cnc benchmarks, \sys learns \param{stable=0}. This coincides with \kissat's pre-defined \param{unsat} mode, which is intended for unsatisfiable instances. This suggests that \sys can independently arrive at solving strategies consistent with expectation. 
 \item However, there are cases where \sys uncovers unexpectedly useful strategies. One interesting example comes from the two benchmarks with the prefix \benchmark{eq.atree.braun}, where \sys learns the configuration \param{bump=0}, \param{tumble=0}, and \param{stable=2}. Upon closer examination, \param{bump=0} and \param{tumble=0} combined contributes most significantly to the performance gain, reducing the runtime on \benchmark{eq.atree.braun.13} by more than 40\%. Interestingly, setting only \param{tumble=0} yields a moderate runtime reduction of approximately 15\%, whereas setting only \param{bump=0}  worsens the runtime by more than $5\times$.
Upon closer examination, \param{bump=0} results in a completely static branching order, while \param{tumble=0} keeps \kissat's indexing of variables as defined in the original encoding (instead of heuristically re-ordering them). Setting both \param{bump=0} and \param{tumble=0} results in a static branching order that splits on the first variables in the original encoding—essentially mimicking a BDD (Binary Decision Diagram) ordering. This represents a novel insight into solving this benchmark, discovered by \sys.
\item Finally, on the \crux benchmarks, \cncTuneAndVal consistently learns subsets of \param{stable=2}, \param{target=0}, and \param{phase=0}. This suggests \cncTuneAndVal can robustly arrive at similar strategies  for similar benchmarks. 
\end{itemize}

\begin{table}[t]
\setlength\tabcolsep{1.5pt}
\centering

\caption{Results on unsolved benchmarks (name abbreviated) from SC'22-24 that were solved by either \cnc or \cncTuneAndVal with 24 cores and a one-hour wall-clock timeout. } 
\vspace{-0.2cm}
\scriptsize
\label{tab:sc}

\begin{tabular}{lcccccccc}
\toprule\vspace{-0.3cm}\\
\cmidrule(lr){3-8}
\cmidrule(lr){9-9}
 & & & \multicolumn{2}{c}{\tabtitle{\cnc}} & \multicolumn{3}{c}{\tabtitle{\cncTuneAndVal}} & \tabtitle{\gray{\painless}} \\
\cmidrule(lr){4-5}\cmidrule(lr){6-8}\cmidrule(lr){9-9}
\tabtitle{Benchmark} & \tabtitle{Res.} & \tabtitle{\coltimecubing} & \tabtitle{\coltimesolving} & \tabtitle{\colscorecommon} & \tabtitle{\coltimesolvinglearning} & \tabtitle{\colscorecommon} & \tabtitle{\coltimelearning} & \tabtitle{\gray{\coltimetotal}} \\
W4.2.90.10.1SS & UNS & 1.3 & 183 & 1.096e+8 & \best{153} & \best{8.083e+7} & 14 & \gray{TO} \\
W4.3.90.10.1DA & UNS & 5.1 & 441 & 1.310e+8 & \best{253} & \best{8.857e+7} & 23 & \gray{TO} \\
W5.3.50.10.2DA & UNS & 9.2 & 2058 & 5.279e+8 & \best{1015} & \best{3.490e+8} & 24 & \gray{TO} \\
grs\_192\_256 & UNS & 178.4 & \best{1390} & 9.760e+7 & 1963 & \best{6.926e+7} & 166 & \gray{TO} \\
multiplier.14.14 & UNS & 0.4 & \best{165} & 2.005e+7 & 183 & \best{1.890e+7} & 59 & \gray{512} \\
php15\_mixed\_15 & UNS & 0.3 & TO & - & \best{1150} & - & 11 & \gray{2852} \\
sin.depth.miter.1 & UNS & 1503.1 & 1117 & 2.396e+7 & 1127 & \best{2.382e+7} & 1127 & \gray{TO} \\
\bottomrule
\end{tabular}
\vspace{-0.2cm}
\end{table}

\subsubsection{Unsolved benchmarks from SC'22-24.}

We now discuss results on \scBench, a total of 73 benchmarks unsolved during recent SAT Competitions. Most of those benchmarks remain out-of-reach for \cnc and \cncTuneAndVal with the allocated computational resources. On 16 benchmarks, our prototype ran out of time or memory during the initial cubing phase. Nonetheless, \cnc and \cncTuneAndVal combined solved 7 out of the 57 remaining benchmarks. The results are shown in Tab.~\ref{tab:sc}. We see that \cncTuneAndVal consistently reduces the total number of conflicts and solves one benchmark which \cnc alone is unable to tackle. Benchmark \benchmark{grs\_192\_256} is an interesting case, where the total number of conflicts is significantly reduced by \sys but its runtime increases. Upon closer examination, we discovered that while \cncTuneAndVal did solve most cubes faster, there is one cube that took the new strategy over 1000 seconds and the default strategy less than 300 seconds. We hypothesize that the solver got to this cube late in the solver process, leading to a slow-down in the solving time. We have not discovered other cases where a single cube dominated the runtime. In the future, it could be interesting to learn not only the solving strategy but also cube-ordering on the fly so that hard cubes could be front-loaded for better load balancing.

Scatter plots like Fig.~\ref{fig:cherry} are shown in App.~\ref{app:figures}. For the completeness of the studies, we also evaluate all SC'24 benchmarks in App.~\ref{app:sc24}. \cncTuneAndVal solved all instances solvable by \cnc and uniquely solved 10 additional instances. 

\subsubsection{Ablation studies.}\label{subsec:ablation}
We conduct ablation studies of \sys on \benchmark{cruxmiter32seed4}, which was closely studied in a distributed-\cnc work~\cite{heisinger2020distributed}. In particular, we stress-test \sys by examining the effect of 1) reducing the number of tuning cubes to just 1 (\texttt{-Cubes}); 2) tuning on the same cube and accepting the new strategy without validation (\texttt{-Cubes-Val}); 3) using a much larger strategy space  (\texttt{+Stra.}); 
and 4) using this expanded strategy space without parameter probing for MCMC initialization (\texttt{+Stra.-Prob.}). 
We also varied the random seed (default is 0) of \cncTuneAndVal from 0 to 9.

The total runtime of \cncTuneAndVal for each configuration-seed pair is shown in Figure~\ref{fig:ablation}. For reference, default \cnc took 746 seconds on this benchmark. Looking at the default configuration of \cncTuneAndVal (the first row), we observe that as we vary the random seeds, \cncTuneAndVal's runtime does vary, due to variations in the learned strategies and in the time it took to tune/validate on the collected cubes. 
Nonetheless, it consistently boosts performance over \cnc. 
On the other hand, when we sample much fewer cubes for tuning (the second row), the chance of over-fitting increases. Indeed, with random seeds 2 and 4, a solving strategy that does not generalize was learned. However, the runtime performance was \emph{not hurt} w.r.t. the default \cnc,
as validation successfully rejected this solving strategy. The third row shows the effect of directly committing to the learned strategy without validation. Unsurprisingly, for seeds 2 and 4, the solving strategy that would have been rejected significantly degrades the runtime performance. Upon closer examination, this unfavorable strategy turned out to be \param{stable=0}, \kissat's \texttt{unsat} mode. This again speaks to the need of per-instance strategy adaptation. The effects of larger strategy space and MCMC initialization are discussed in App.~\ref{app:ablation}.

\begin{figure}[t]
\includegraphics[width=\linewidth]{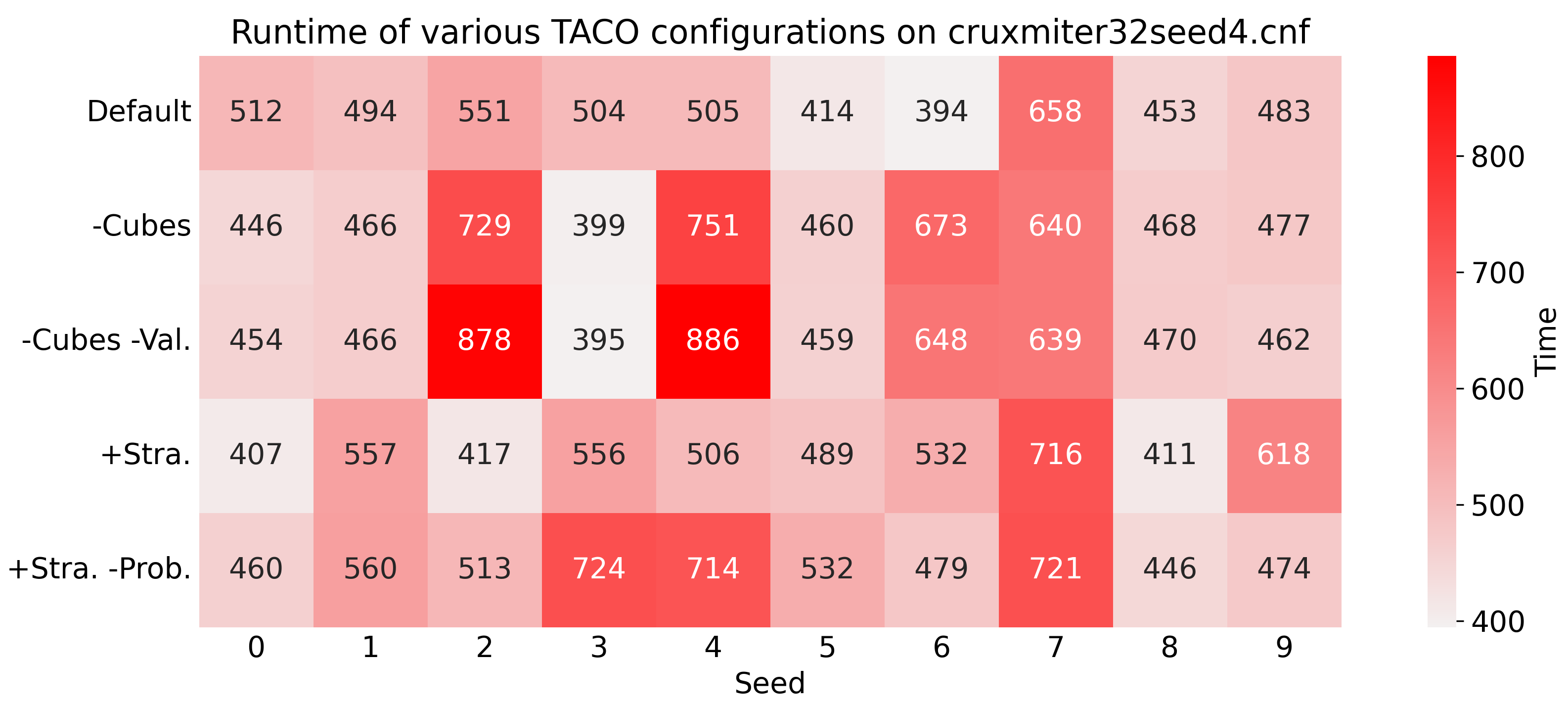}
    \vspace{-0.4cm}
    \caption{Ablation studies of the effect of, from top to bottom, less tuning cubes (\texttt{-Cubes}), no validation (\texttt{-Cubes-Val.}), larger strategy space (\texttt{+Stra.}), and no probing for MCMC initialization (\texttt{+Stra.-Prob.}). Vanilla \cnc took 746 seconds.}
    \label{fig:ablation}
    \vspace{-0.4cm}
\end{figure}


\subsection{Case Study: Neural Network Verification}\label{sec:eval-nnv}

We now shift our attention to a second automated reasoning task: solving neural network verification (NNV) queries. We used two of the four benchmark sets used for evaluating the latest version of \marabou~\cite{wu2024marabou}: 
\begin{itemize}
    \item \nap: 235 verification benchmarks concerned with checking whether certain neuron activation pattern implies certain output classification~\cite{gopinath2019property}; and
    \item \altloop: 259 benchmarks concerned with checking the existence of infinite loops in a robot navigation system~\cite{amir2023verifying}.
\end{itemize}
We omitted the two other benchmark sets. One is already easily solvable; the other involves invoking \marabou repeatedly as a sub-procedure through its Python API, which creates engineering challenges for our prototype as it makes sub-process calls to the \marabou executable.

We use \marabou to create 4096 cubes with its default partitioning strategy. We created fewer cubes in this experiment because we allocated fewer computational resources for each solver run, as described below. The number of cubes collected for tuning (\paramSampleTarget) is 30; the number of cubes collected for validation (\paramSampleTargetOracleTraining) is 15; the cost range (\paramMinCost and \paramMaxCost) of tuning cubes is between 10 and 50, and that of validation cubes is between 50 to 100; for re-partitioning during \algCollect, we use \marabou to create 8 new cubes. Note that the cost ranges are significantly smaller than the ones used in SAT-solving. This is because \marabou performs extensive theory reasoning at each search state~\cite{katz2017reluplex} and expands the search tree much slower than SAT solvers. We believe similar considerations would apply when extending \sys to other lazy-DPLL(T)-based solvers. In addition to \cnc and \cncTuneAndVal, we also ran the parallel solving mode of \marabou~\cite{wu2020parallelization} (\dncmarabou), which performs \cnc-based parallel-solving with dynamic re-partitioning. For \dncmarabou, we create 16 initial sub-problems and leave its other hyper-parameters as default. Each job is assigned 8 cores (thus 7 workers are spawned), 64GB of memory, and a one-hour wall-clock timeout.

\begin{table}[t!]
\setlength\tabcolsep{1.6pt}
\centering
\caption{Effect of \sys on benchmarks used to evaluate Marabou. All configurations were given 8 cores and 1 hour timeout.}
\vspace{-0.2cm}
\scriptsize
\label{tab:nnv}
\begin{tabular}{ccccccccccccccc}
\toprule
& \multicolumn{2}{c}{\tabtitle{\cnc}} 
& \multicolumn{5}{c}{\tabtitle{\cncTuneAndVal}}
& \multicolumn{2}{c}{\tabtitle{\gray{\dncmarabou}}}
 \\
\cmidrule(lr){2-3}\cmidrule(lr){4-8}\cmidrule(lr){9-10}
\tabtitle{Family} (\#) & 
\tabtitle{\colsolved} &
\tabtitle{\coltimetotal} & 
\tabtitle{\colsolved} &
\tabtitle{\coltimetotal} & \tabtitle{\coltimelearning} & 
\tabtitle{\colnumupdated} &
\tabtitle{\colreductiononupdated} &
\gray{\tabtitle{\colsolved}} &
\gray{\tabtitle{\coltimetotal}} & 
\\
\benchmark{NAP} (235) & 163  & 1206538
& \best{216} & \best{33606} & 8955 & 0 & - 
& \gray{216} & \gray{420312}
\\
\benchmark{AltLoop} (259) & 259  & 93764
& 259 & \best{66532} & 2489 & 53 & 44\% 
& \gray{258} & \gray{108459}
\\
\bottomrule
\end{tabular}
\vspace{-0.2cm}
\end{table}

\subsubsection{Evaluation on \nap and \altloop.}

Tab.~\ref{tab:nnv} summarizes the runtime performance of \cnc and \cncTuneAndVal. Since cubing with \marabou took less than 1 second on those benchmarks, we did not split the runtime into cubing time and solving time. We report the total solving time (\coltimetotal), and for \cncTuneAndVal, we also report the time spent on strategy learning (\coltimelearning), the number of benchmarks on which a new solving strategy has been found (\colnumupdated), and the percentage of overall runtime reduction on those benchmarks. Results on individual benchmarks are reported in App.~\ref{app:detailed-nnv}.

\cncTuneAndVal exhibits interesting behaviors on both benchmark sets. On the \nap benchmarks, \cncTuneAndVal solves significantly more benchmarks than \cnc. Upon closer examination, this significant performance gain is \emph{not} due to new strategies learned, but in fact due to satisfiable cubes discovered during the cube collection process. While default \cnc attempts to solve each cube before moving on to the next cube, the \algCollect method creates a \emph{depth-limited} (or more precisely, cost-limited) breadth-first search pattern, where the solver scrolls through and attempts at cubes with a low-cost budget. This turns out to be the preferable approach for finding satisfying assignments on this benchmark set. \dncmarabou also solved 216 \nap benchmarks, though with a higher average solving time. Upon closer examination, \cncTuneAndVal and \dncmarabou together solved 230 (all \sat) out of the 235 benchmarks, while \cnc does not contribute any unique solutions. To our knowledge, this is the first study to show that almost none of the activation patterns in the NAP benchmark set can actually guarantee a fixed prediction, despite their previously reported statistical correlation with specific outputs~\cite{gopinath2019property}.

As for the \altloop benchmarks, both \cnc and \cncTuneAndVal solved all the instances, while \cncTuneAndVal is overall more efficient and results in a 29\% reduction of total solving time ($\frac{93764 - 66532}{93764}$). We found that for the vast majority of the benchmarks, \sys did not find a different solving strategy, suggesting that the default solving strategy is already reasonable. However, there are 53 benchmarks where a new solving strategy is found. On those benchmarks, \cncTuneAndVal reduces the total runtime by 44\% (from 46761 seconds to 26161 seconds). 

Examining the new strategies found by \sys in those 53 cases (App.~\ref{app:detailed-nnv}), we discovered that they always involves performing less input splitting (and more case splits on internal neurons instead). Interestingly, this deviates from the conventional wisdom~\cite{wu2020parallelization,bunel2020branch,jia2021verifying} that input splitting is preferable for neural networks with low input dimensions (networks in \altloop have input dimension $\leq 20$). 
In fact, 45 out of the 53 new strategies completely turn off input splitting (by setting \param{pl-split-freq} to 1). This suggests that even on benchmarks from the ``same distribution,'' there might not be a uniformly optimal solving strategy.

In App.~\ref{app:seq-taco}, we report result of running a sequential \sys configuration. 
We found that compared with the default sequential mode of \marabou, this configuration can solve roughly twice more \nap benchmarks and  speed up the solving of \altloop benchmarks by more than 2x. This suggests \sys has the potential to improve performance in a sequential solving setting as well.
\section{Related Work}\label{sec:related}


\subsubsection{Algorithm tuning and selection.} Our work is directly motivated by the success of offline meta-algorithmic design approaches~\cite{hoos2021automated} such as automated configuration~\cite{hutter2007boosting,hutter2007automatic,khudabukhsh2016satenstein,hutter2014algorithm,ansotegui2015model,leyton2009empirical} and per-instance algorithm selection~\cite{xu2008satzilla,scott2021machsmt,xu2010hydra,singh2009avatarsat}. Automated configuration focuses on optimizing parameter settings of an algorithm for a fixed set of problems, using either local search or performance prediction techniques. Per-instance algorithm selection techniques utilize machine learning to predict the most effective algorithm for a given formula. 
\sys differs from both approaches in that it moves meta-algorithmic design online. In addition, compared to offline per-instance algorithmic selection, \sys does not require feature extraction (i.e., representing a formula as a set of features), which is in itself a complex problem.

\subsubsection{Online learning.} The general idea of online learning is not new in automated reasoning. The \solver{MapleSAT} solver~\cite{DBLP:conf/sat/LiangGPC16,liang2018machine} was among the first to explore such a direction. For example, in \solver{MapleSAT}, branching is formulated as a multi-armed bandit problem and the estimated reward of branching on each variable is updated throughout the solving.  
More broadly, algorithms with adaptive components are common in automated reasoning, such as dynamic local search~\cite{li2007combining}, branching heuristics~\cite{moskewicz2001chaff,cherif2021combining}, and restart strategies~\cite{biere2008adaptive}. 
In contrast, \sys focuses on customizing high-level solving strategies instead of designing adaptive low-level heuristics. There has also been work on choosing solving strategies online when solving a sequence of related problems~\cite{pimpalkhare2021medleysolver,wu2023lightweight}. Our approach conducts online strategy learning when solving a single formula. 

\subsubsection{Partitioning strategies.} Our work leverages partitioning strategies devised in cube-and-conquer-based automated reasoning. Cube-and-conquer was originally designed for tackling challenging SAT problems~\cite{heule2011cube,heule2016solving}, by partitioning a problem into easier and balanced sub-problems using lookahead techniques~\cite{heule2009look}. Various other partitioning strategies for SAT and SMT have been devised~\cite{nejati2020machine,nejati2017propagation,heule2011cube,wilson2023partitioning,hyvarinen2015search,wu2020parallelization,zhao2024distributed,hyvarinen2021lookahead}.

\section{Conclusion and Next Steps}\label{sec:conclusion}

In this paper, we explored the feasibility of customizing the solving strategy for a given problem by learning from the problem itself. We introduced \sys, an online learning methodology that extends a traditional cube-and-conquer solving procedure. Experimental evaluation in two automated reasoning applications showed that \sys can consistently boost the performance of \cnc, uncover effective instance-specific solving strategy, and in many cases outperform state-of-the-art solvers in the respective domains.

Our long-term vision is that \emph{intelligent agents should both efficiently conduct low-level reasoning and adaptively shift their high-level solving strategy}---this work is one step towards that direction.
%
We believe the application of \sys to more automated reasoning tools (e.g., general-purpose SMT solvers) is a promising next step and might reveal that new design choices are required.  It would also be interesting to study the applicability of \sys in the cloud setting, where the efficient sharing of training data and strategies becomes less straightforward. Finally, we believe deep integration of \sys-like approach into the search procedure of modern solvers 
would be a significant next challenge to tackle. 


%
%
\newpage
\bibliographystyle{kr}
\bibliography{references}

\newpage
\onecolumn
\appendix
\section{Proofs}\label{app:proofs}

\terminationalt*  
\begin{proof}
Assume, for contradiction, that
$\algCollect(\formula,\cubes,\sampleTarget,\strategyspace)$
does not terminate, this means the loop executes infinitely many times. Each iteration selects one cube
$\cube\in\cubes$, runs
$\funcCheck(\formula\!\land\!\cube,\strategy,\paramMaxCost)$, and then either terminates immediately (if the result is \sat),  or remove $\cube$ from \cubes (if the result is \unsat), or replaces $\cube$ by $\paramOnlineCubes>1$ sub-cubes (if the result is \unknown). Therefore, an infinite execution must receive the outcome
$\unknown$ in infinitely many iterations; otherwise, $|\cubes|$ would eventually drop to~$0$ (if the procedure does not terminate earlier). Hence, $\funcCube$ must be invoked infinitely many times. 
%
This means we can  build an infinite sequence $\trace_\infty:= \top,\cube_0,\cube_1,\dots$ with the following properties:
1) $\cube_0$ is produced by invoking $\funcCube(\formula \land \cube, \paramOnlineCubes)$ for some cube $\cube$ in the initial~$\cubes$; 2) $\cube_{i+1}$ is produced by $\funcCube(\formula \land \cube_i, \paramOnlineCubes)$ when \funcCheck\ returns \unknown\ on $\cube_i$; and 3) \funcCheck\ returns \unknown\ on \emph{every} $\cube_i$. Sequence $\trace_\infty$ is a \emph{cubing trace} starting from $\formula\land\cube$ (Def.~\ref{def:trace}).

However, by the bounded cost reduction assumption of the theorem, every cubing trace will reach, in a finite number of steps, a cube that can be decided within budget \paramMaxCost by any solving strategy $\strategy\in\strategyspace$. Let $\cube_d$ be that cube for trace $\trace_\infty$. On $\cube_d$, \funcCheck\ can return only \sat\ or \unsat, never \unknown. This contradicts the conclusion that $\funcCheck$ returns \unknown for all cubes in $\trace_\infty$. Therefore the assumption of non-termination is impossible, and the procedure must terminate.
\end{proof}

Next, for SAT-solving and Neural Network Verification respectively, we define conditions on the \funcCube procedure, the solving strategies, and the cost metric that guarantee the bounded cost reduction property, which in turn guarantees the termination of \algCollect.

\begin{lemma}[Bounded cost reduction for SAT solving]
\label{lem:bounded-cost-sat}
Let $\formula$ be a CNF formula. Suppose the partitioning procedure
$\funcCube(\formula,\paramOnlineCubes)$ satisfies:
\begin{itemize}
\item every returned cube $\cube$ is a conjunction of unit clauses whose variables occur in $\formula$;
\item whenever possible, $\cube$ contains a unit clause that does \emph{not} yet occur in $\formula$.
\end{itemize}
Let $\strategyspace$ be any set of SAT--solving strategies that perform
exhaustive \emph{unit propagation} before making their first decision.
Then the bounded--cost--reduction assumption of
Theorem~\ref{prop:termination-alt} holds with cost budget
$\paramMaxCost>1$, where the cost metric is the number of conflicts.
\end{lemma}

\begin{proof}
Let $n$ be the number of variables in $\formula$, and let
$\cube$ be a cube reachable from $\formula$ in at least $n$ steps.
The conjunction $\formula\land\cube$ must contain at least
$n$ different unit clauses. Hence, either all $n$ variables are assigned, or some variable and its negation both appear as unit clauses. In either case, unit propagation can decide the satisfiability of $\formula\land\cube$ with at most one conflict.
\end{proof}

The cubing tool \marchcu satisfies the characteristics of $\funcCube$ described above, except that when the given formula is already solvable by unit propagation, \marchcu will directly conclude $\unsat$ instead of generating any new cubes.
\begin{lemma}[Bounded cost reduction for Neural Network Verification]
\label{lem:bounded-cost-nnv}
Define a neural network verification query on ReLU neural networks as a conjunction of linear constraints together with ReLU constraints of the form $y = \max(0,x)$, where $x$ and $y$ are real variables. Let $\formula$ be such a query. Suppose the partitioning procedure
$\funcCube(\formula,\paramOnlineCubes)$ satisfies:
\begin{itemize}
\item every cube $\cube$ it returns is a conjunction of linear constraints over variables occurring in $\formula$;
\item $\cube$ \emph{fixes} at least one unfixed ReLU constraint $y = \max(0,x)$, that is, it asserts either $(x \ge 0 \land y = x)$ or $(x < 0 \land y = 0)$.
\end{itemize}
Let $\strategyspace$ be any set of NNV strategies that can decide the satisfiability of a set of linear arithmetic constraints without performing any case splits. Then the bounded--cost--reduction assumption of
Theorem~\ref{prop:termination-alt} holds with cost budget
$\paramMaxCost>0$, where the cost metric is the number of case splits on ReLU constraints.
\end{lemma}

\begin{proof}
Let $n$ be the number of ReLU constraints in $\formula$, and let
$\cube$ be a cube reachable from $\formula$ in at least $n$ steps. The conjunction $\formula\land\cube$ fixes every ReLU to a linear phase. Consequently, satisfiability can be decided by solving only the linear constraints, which requires no case splits.
\end{proof}

The lemma can be extended to neural networks with other piecewise-linear activation functions.
\marabou{} supports two cubing methods: one based on input splitting (splitting the intervals of variables representing the inputs to the neural network) and another based on ReLU splitting. The first method can be modified to satisfy the bounded-cost-reduction assumption by eventually switching to ReLU splitting, or by performing ReLU splitting periodically.

\newpage
\section{Additional Implementation Details}\label{app:impl}


\subsection{Choice of the cost metric} Importantly, the assumption that the number of conflicts is a good proxy for runtime is not always true. For example, if we turn off learned clause deletion, the number of conflicts would likely reduce, but the runtime can significantly increase due to overhead in unit propagation. This fact needs to be taken into account when choosing which parameter to include in the strategy space: the parameter's impact on the cost metric and its impact on the runtime must \emph{align}. If a parameter setting reduces the cost metric but simultaneously increases runtime, the tuning procedure would be misled to favor the parameter setting. We believe that coming up with a deterministic, easy-to-measure, and consistently faithful performance proxy is a highly non-trivial research topic. 

\subsection{Evaluating strategies on collected cubes} In tuning and validation, in order to evaluate a solving strategy \strategy, we need to solve a set of cubes with \strategy to compute its cost on those cubes. We impose a cost budget of \paramMaxCost for each cube. If a cube \cube is not solved within the budget, its cost on the \cube is treated as $2 \cdot \paramMaxCost$ (PAR2). This avoids bad solving strategies taking a long time to finish on the cubes. We have also tried PAR5 or PAR10, and found that the performance is not sensitive to this choice.

\newpage
\section{Experimental Evaluation}

\subsection{Additional details on experimental setup}\label{app:exp-setup}
\paragraph{Unsolved benchmarks from recent SAT Competitions} This information is retrieved using queries like ``\texttt{result = unknown and track = main\_2024}'' on a recently introduce SAT benchmark database~\cite{iser_et_al:LIPIcs.SAT.2024.18}.

\paragraph{Cube creation} For SAT-solving, we used \marchcu to partition a given benchmark with maximal cube depth 15. This means at most 32768 cubes are generated. \marchcu generates cubes close or equal to this number on most benchmarks. However, on some benchmarks (in particular, the two benchmarks where re-partitioning occurred), only thousands of cubes were created. On the other hand, if the maximal depth is unspecified, \marchcu would dynamically decide the number of cubes. We did not use this option since we found that \marchcu runs out of time or memory on many of the benchmarks under this mode.

\paragraph{\painless version} We built \painless from the source code available on Github.\footnote{\url{https://github.com/S-Mazigh/painless/tree/satcomp-24}} The \painless configurations used during the SAT competition are available from the solver download link at \url{https://satcompetition.github.io/2024/index.html}.\footnote{
See \texttt{parallel/painless/painless2/docker/painless-images/leader/run\_solver.sh}
} The concrete \painless configuration we used is \texttt{-v=1 -c=23 \\ -shr-strat=3 -shr-sleep=100000 -sbva-count=12 -ls-after-sbva=2\\ -sbva-timeout=1000}.

\subsection{Continued Discussion of Ablation studies}\label{app:ablation}

For the ablation studies, we consider a much larger strategy space that, in addition to Tab.~\ref{tab:strategyspace-kissat}, also includes 8 additional tunable parameters: \param{and}, \param{equivalences}, \param{extract}, \param{ifthenelse}, \param{minimize}, \param{phasesaving}, \param{rephase}, \param{substitute}. 
This increases the size of the strategy space to 3911.

We now discuss the effect of using this larger strategy space (the fourth row of Fig.~\ref{fig:ablation}). We observe that with this larger strategy space, \cncTuneAndVal still managed to consistently find solving strategies that improve runtime performance within the limited rounds of MCMC-sampling (20). This speaks to the efficacy of the MCMC-sampling procedure described in Sec.~\ref{subsec:strategy}. Comparing the fourth and fifth rows of Fig.~\ref{fig:ablation} illustrates the effect of parameter probing. If we initialize MCMC-search with the default solving strategy instead of the best strategy found through probing, there appear to be fewer cases where a significantly better solving strategy is found. Another way to read this result is that, for a larger strategy space, more samples are needed. Considering that one can only afford a limited number of samples in the online solving scenario, it would be an interesting next step to leverage developer expertise to devise a rich yet compact set of promising solving strategies in order to optimize the performance of \sys. 

\newpage
\section{Solving Strategies Found by CnC+TACO on the Evaluated SAT Benchmarks}\label{app:strategy}

\begin{table}[H]
\centering
\small
\label{tab:sat-config}
\begin{tabular}{lr}
\toprule
\tabtitle{Benchmark} & 
\tabtitle{Solving strategy}\tablefootnote{Only parameters that differ from the default are shown.} \\
9dlx\_vliw\_at\_b\_iq8 & chrono=0 stable=0 \\
9dlx\_vliw\_at\_b\_iq9 & chrono=0 stable=0 \\
AProVE07-25 & chrono=0 stable=0 \\
dated-5-19-u & chrono=0 phase=0 target=0 \\
eq.atree.braun.12 & bump=0 phase=0 stable=2 tumble=0 \\
eq.atree.braun.13 & bump=0 stable=2 tumble=0 \\
gss-24-s100 & forcephase=1 stable=2 \\
gss-26-s100 & stable=2 \\
gus-md5-14 & forcephase=1 \\
ndhf\_xits\_09\_UNS & stable=0 \\
rbcl\_xits\_09\_UNK & stable=0 \\
rpoc\_xits\_09\_UNS & forcephase=1 phase=0 stable=0 \\
total-10-17-u & chrono=0 stable=0 \\
total-5-15-u & chrono=0 target=0 \\
\midrule
cruxmiter32seed0 & stable=2 target=0\\
cruxmiter32seed1 & stable=2 target=0 \\
cruxmiter32seed2 & target=0 \\
cruxmiter32seed3 & phase=0 stable=2 \\
cruxmiter32seed4 & forcephase=1 stable=2 \\
cruxmiter32seed5 & phase=0 stable=2 target=0 \\
cruxmiter32seed6 & phase=0 stable=2 target=0 \\
cruxmiter32seed7 & phase=0 stable=2 target=0 \\
cruxmiter32seed8 & phase=0 stable=2 \\
cruxmiter32seed9 & stable=2 target=0 \\
\midrule
WS\_400\_24\_90\_10.apx\_1\_DS\_ST & stable=2 \\
WS\_400\_32\_90\_10.apx\_1\_DC\_AD & stable=2 target=0 \\
WS\_500\_32\_50\_10.apx\_2\_DC\_AD & stable=2 target=0 \\
grs\_192\_256 & chrono=0 phase=0 stable=0 \\
multiplier\_14bits\_\_miter\_14 & stable=2 target=0 \\
php15\_mixed\_15percent\_blocked & stable=2 target=0 \\
sin\_depth\_miter\_1 & stable=0 \\
\bottomrule
\end{tabular}
\end{table}

\newpage
\section{Scatter Plots of Solver Performance on Individual Cubes}\label{app:figures}

This section shows three scatter plots for each benchmark. The left plot shows the costs (\# conflicts) of cubes used for tuning and validation when solved using the learned solving strategy and the default solving strategy. The middle plot shows the costs on cubes commonly solved by \cnc and \cncTuneAndVal using different solving strategies (excluding cubes used for tuning). The right plot shows the runtime on those cubes. The sums of the values are also shown in the figures. 

\subsection{\cncBench benchmarks}

\begin{figure}[H]
    \centering
    \includegraphics[width=\textwidth]{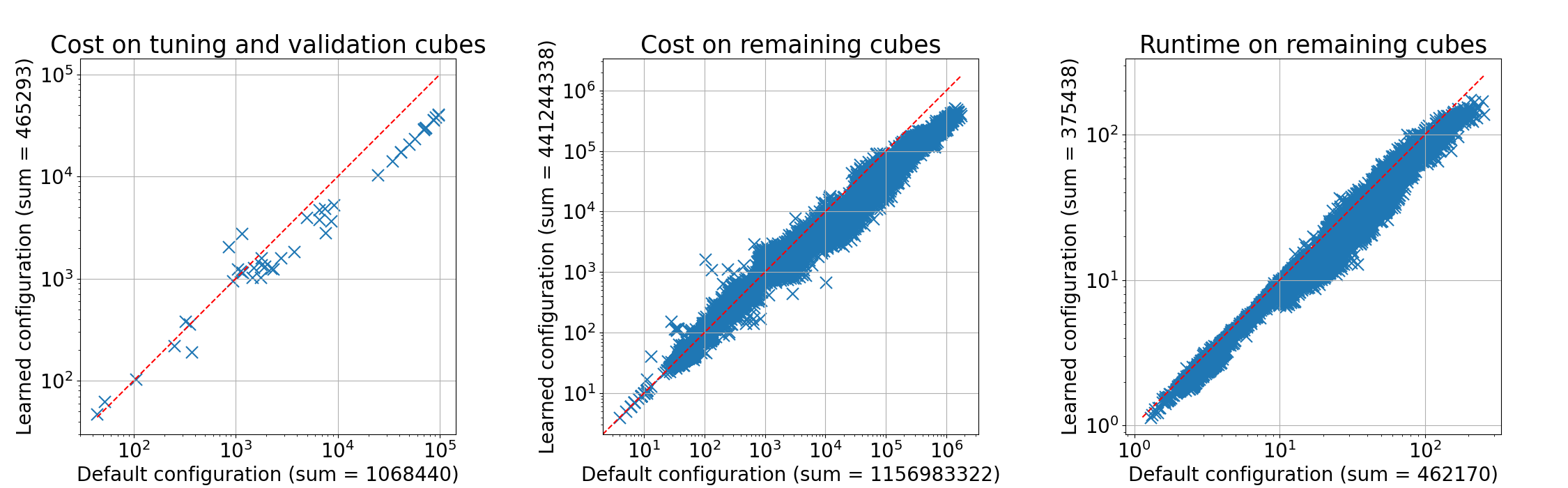}
    \caption{9dlx\_vliw\_at\_b\_iq8}
    \label{fig:combined_benchmark9dlx_vliw_at_b_iq8.cnf}
\end{figure}

\begin{figure}[H]
    \centering
    \includegraphics[width=\textwidth]{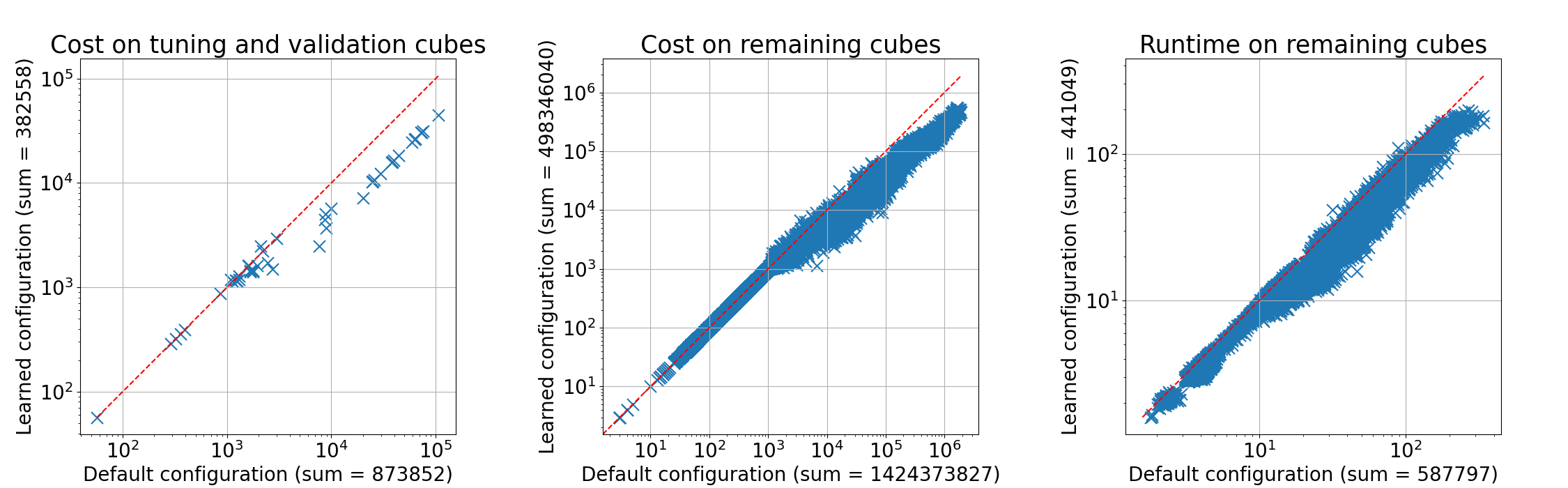}
    \caption{9dlx\_vliw\_at\_b\_iq9}
    \label{fig:combined_benchmark9dlx_vliw_at_b_iq9.cnf}
\end{figure}

\begin{figure}[H]
    \centering
    \includegraphics[width=\textwidth]{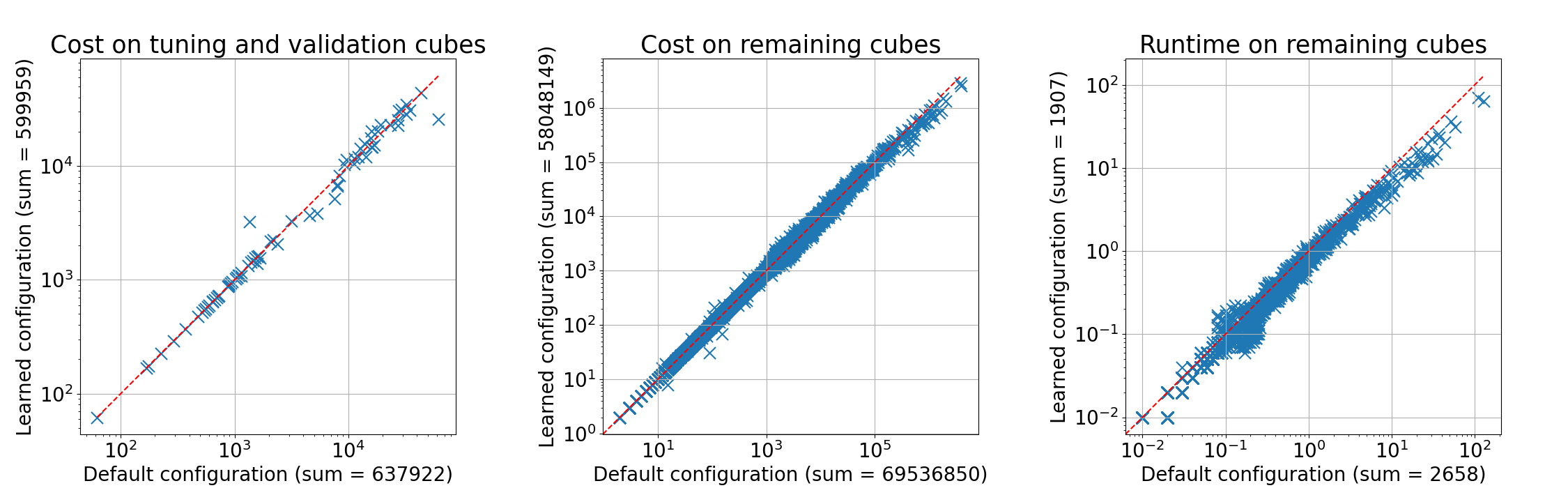}
    \caption{AProVE07-25}
    \label{fig:combined_benchmarkAProVE07-25.cnf}
\end{figure}

\begin{figure}[H]
    \centering
    \includegraphics[width=\textwidth]{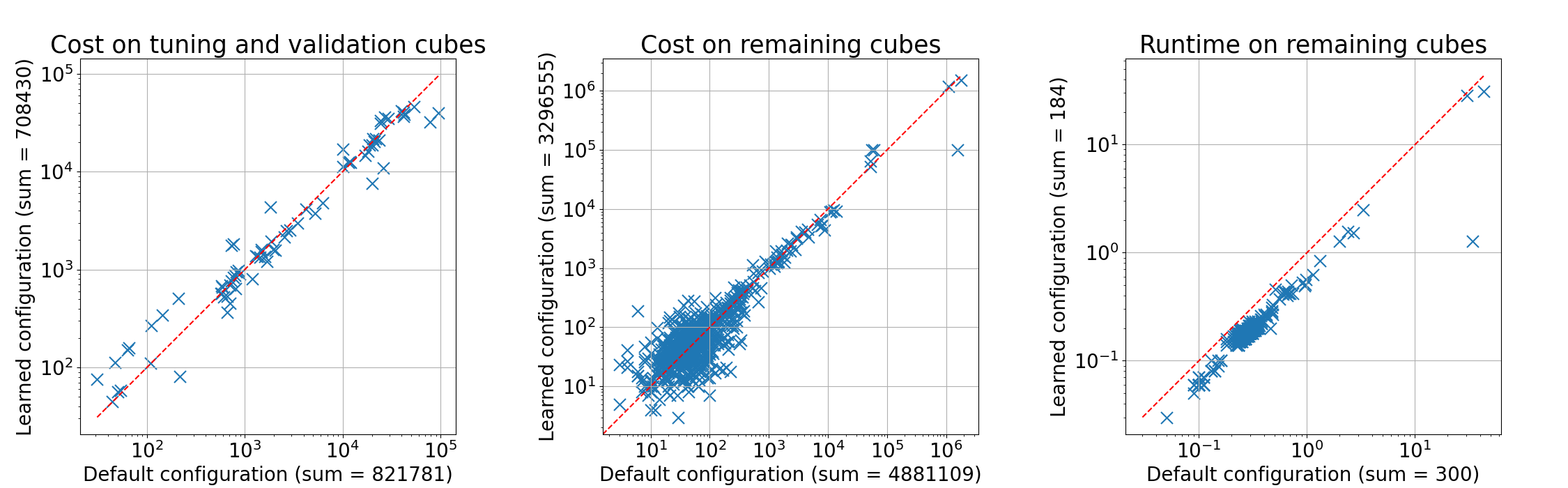}
    \caption{dated-5-19-u}
    \label{fig:combined_benchmarkdated-5-19-u.cnf}
\end{figure}

\begin{figure}[H]
    \centering
    \includegraphics[width=\textwidth]{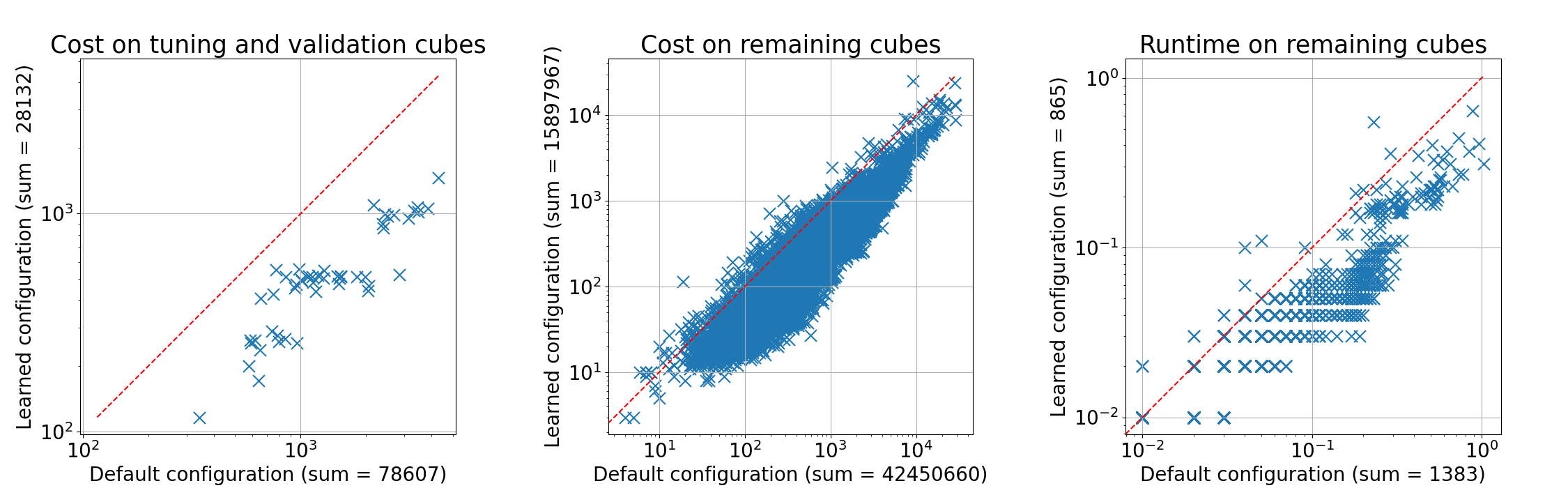}
    \caption{eq.atree.braun.12}
    \label{fig:combined_benchmarkeq.atree.braun.12.cnf}
\end{figure}

\begin{figure}[H]
    \centering
    \includegraphics[width=\textwidth]{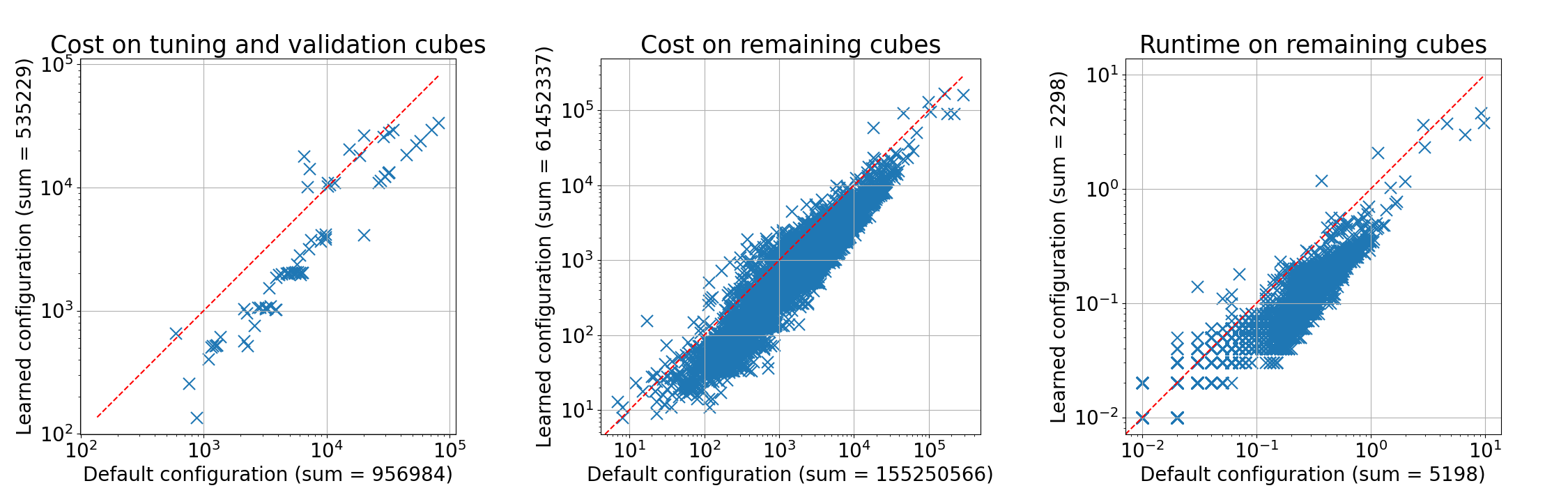}
    \caption{eq.atree.braun.13}
    \label{fig:combined_benchmarkeq.atree.braun.13.cnf}
\end{figure}

\begin{figure}[H]
    \centering
    \includegraphics[width=\textwidth]{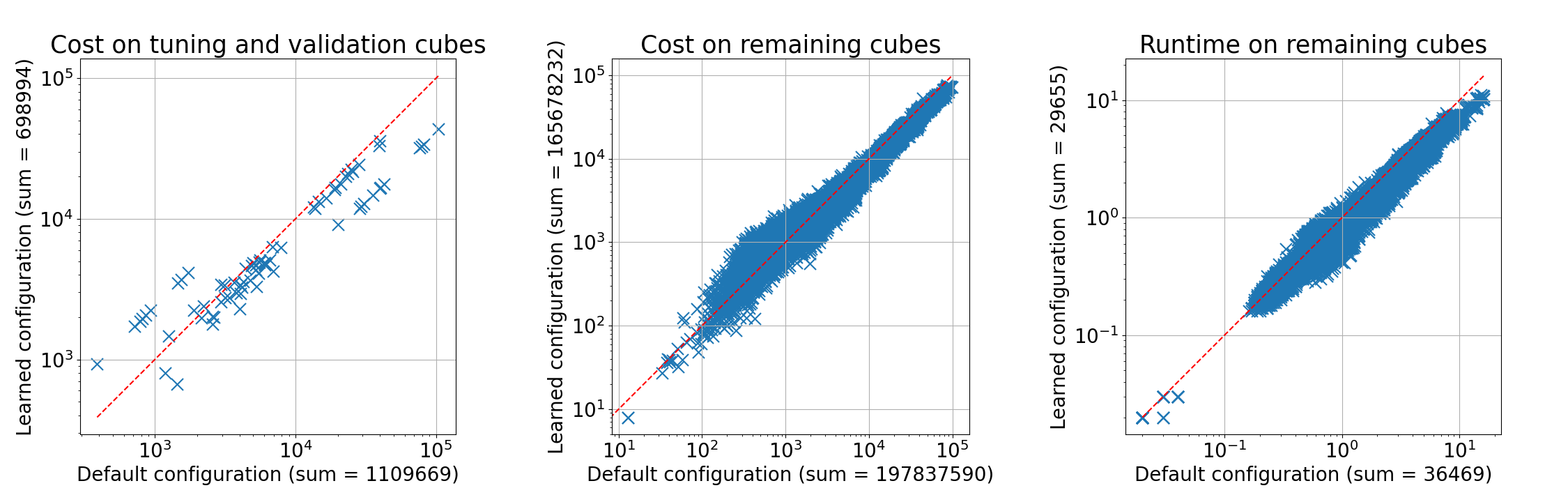}
    \caption{gss-24-s100}
    \label{fig:combined_benchmarkgss-24-s100.cnf}
\end{figure}

\begin{figure}[H]
    \centering
    \includegraphics[width=\textwidth]{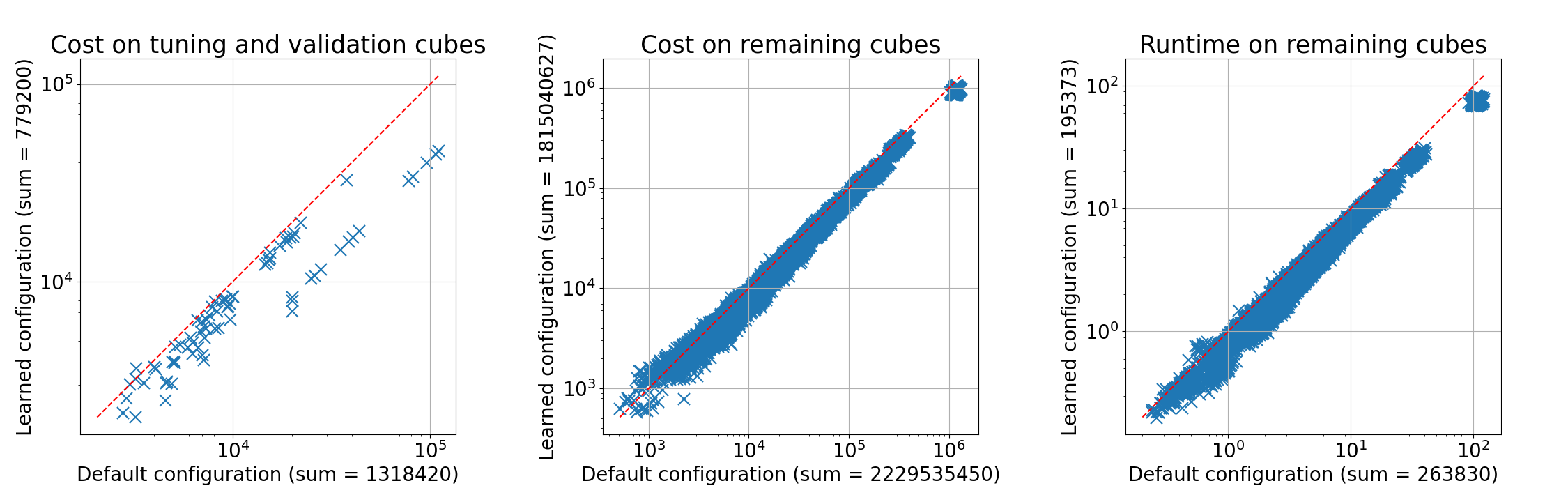}
    \caption{gss-26-s100}
    \label{fig:combined_benchmarkgss-26-s100.cnf}
\end{figure}

\begin{figure}[H]
    \centering
    \includegraphics[width=\textwidth]{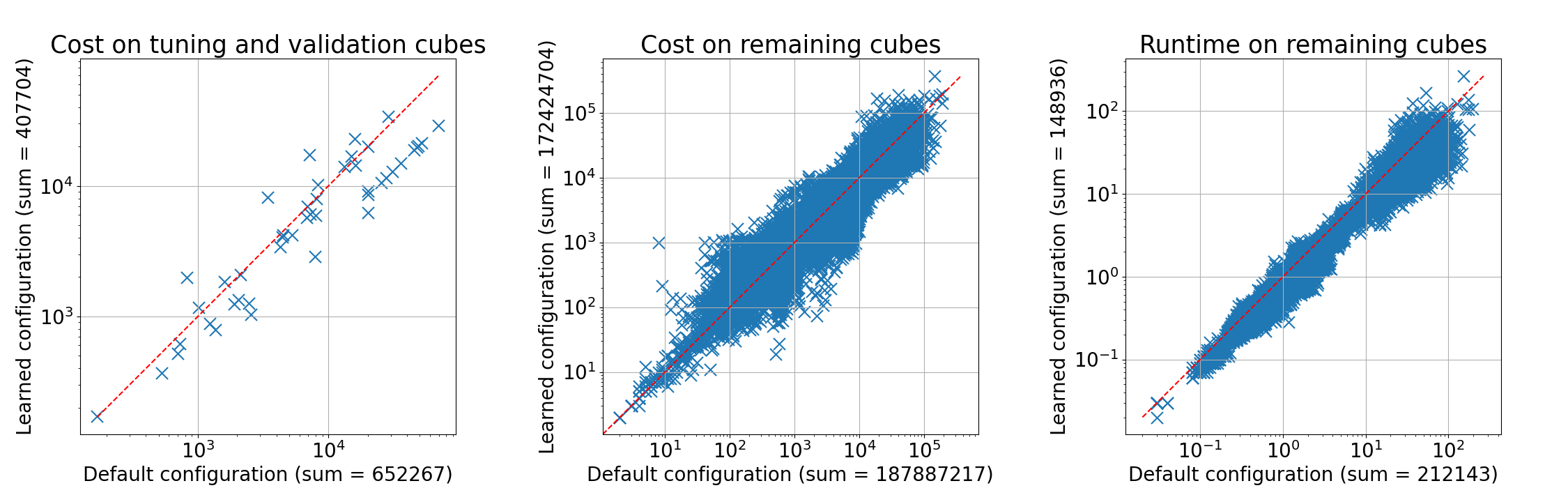}
    \caption{gus-md5-14}
    \label{fig:combined_benchmarkgus-md5-14.cnf}
\end{figure}

\begin{figure}[H]
    \centering
    \includegraphics[width=\textwidth]{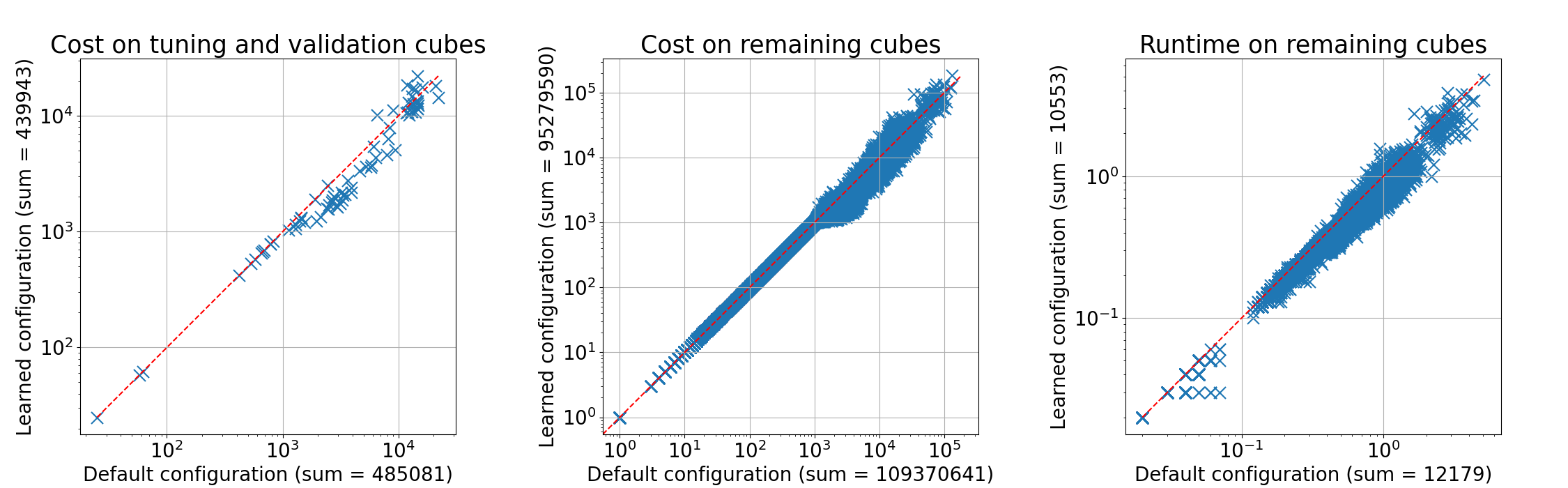}
    \caption{ndhf\_xits\_09\_UNS}
    \label{fig:combined_benchmarkndhf_xits_09_UNS.cnf}
\end{figure}

\begin{figure}[H]
    \centering
    \includegraphics[width=\textwidth]{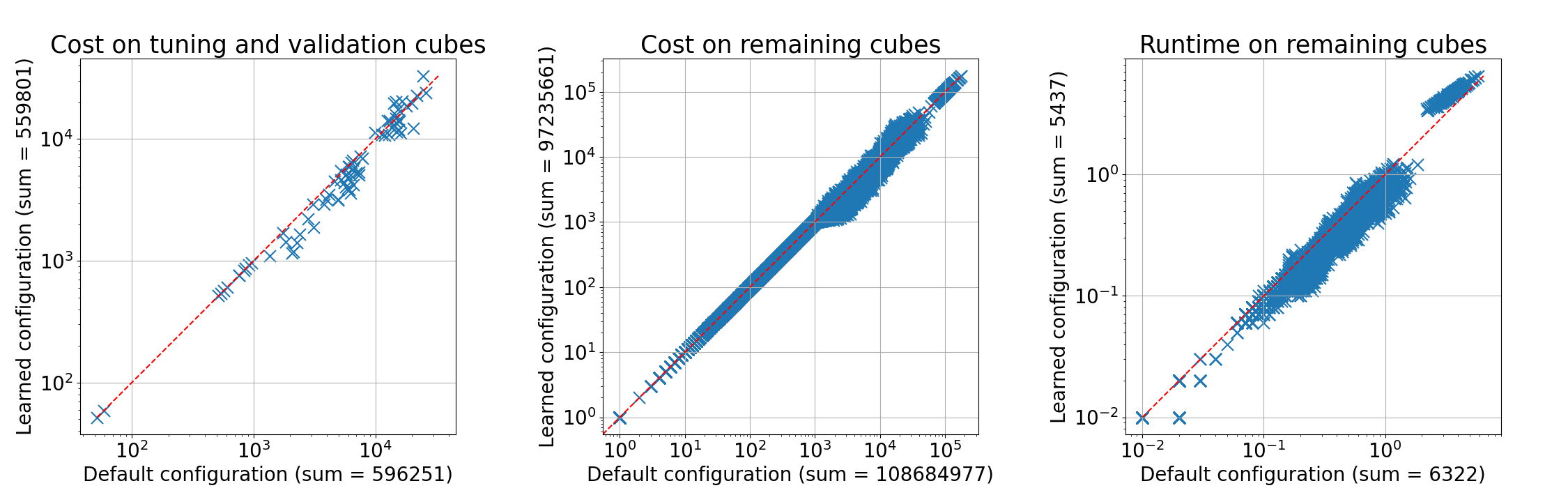}
    \caption{rbcl\_xits\_09\_UNK}
    \label{fig:combined_benchmarkrbcl_xits_09_UNK.cnf}
\end{figure}

\begin{figure}[H]
    \centering
    \includegraphics[width=\textwidth]{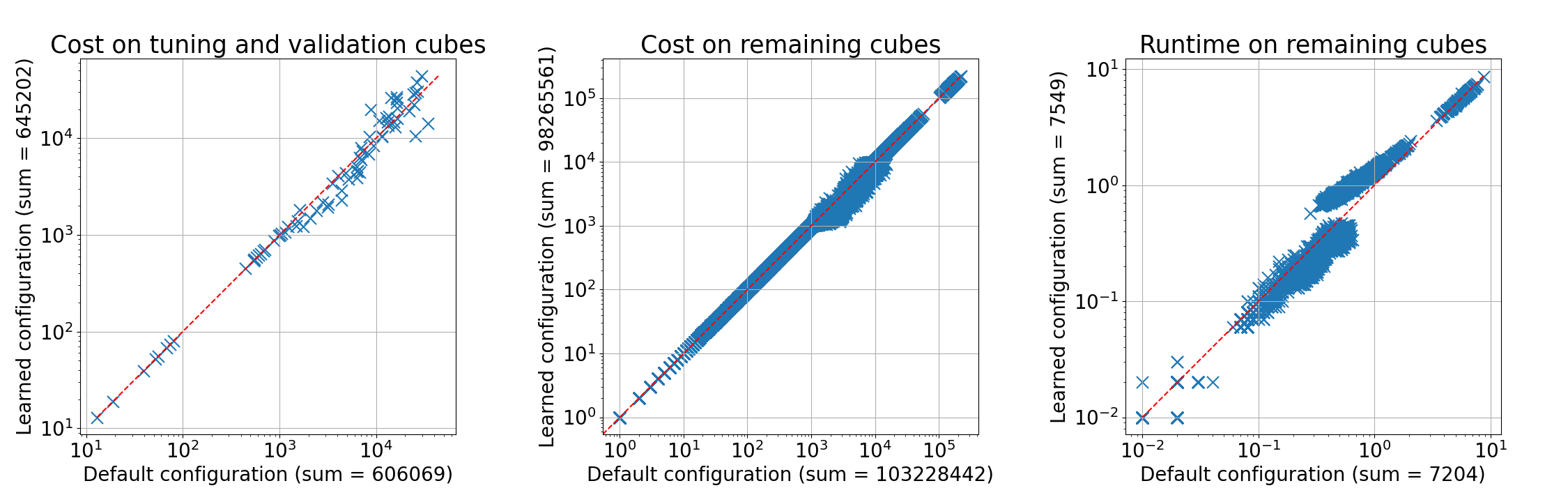}
    \caption{rpoc\_xits\_09\_UNS}
    \label{fig:combined_benchmarkrpoc_xits_09_UNS.cnf}
\end{figure}

\begin{figure}[H]
    \centering
    \includegraphics[width=\textwidth]{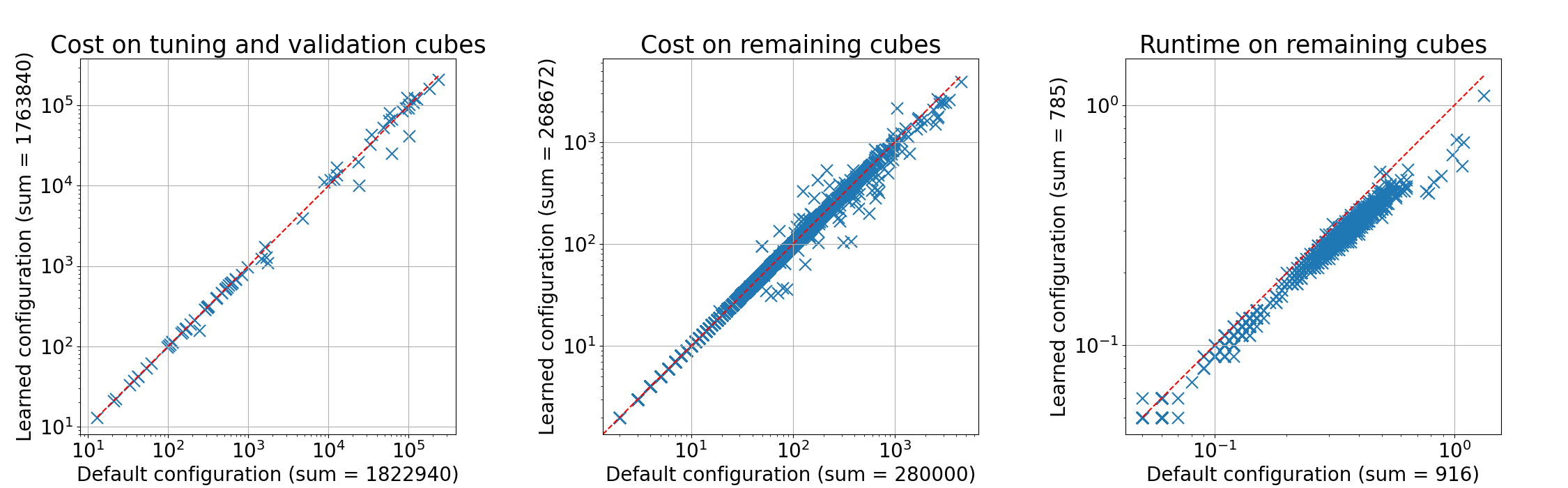}
    \caption{total-10-17-u}
    \label{fig:combined_benchmarktotal-10-17-u.cnf}
\end{figure}

\begin{figure}[H]
    \centering
    \includegraphics[width=\textwidth]{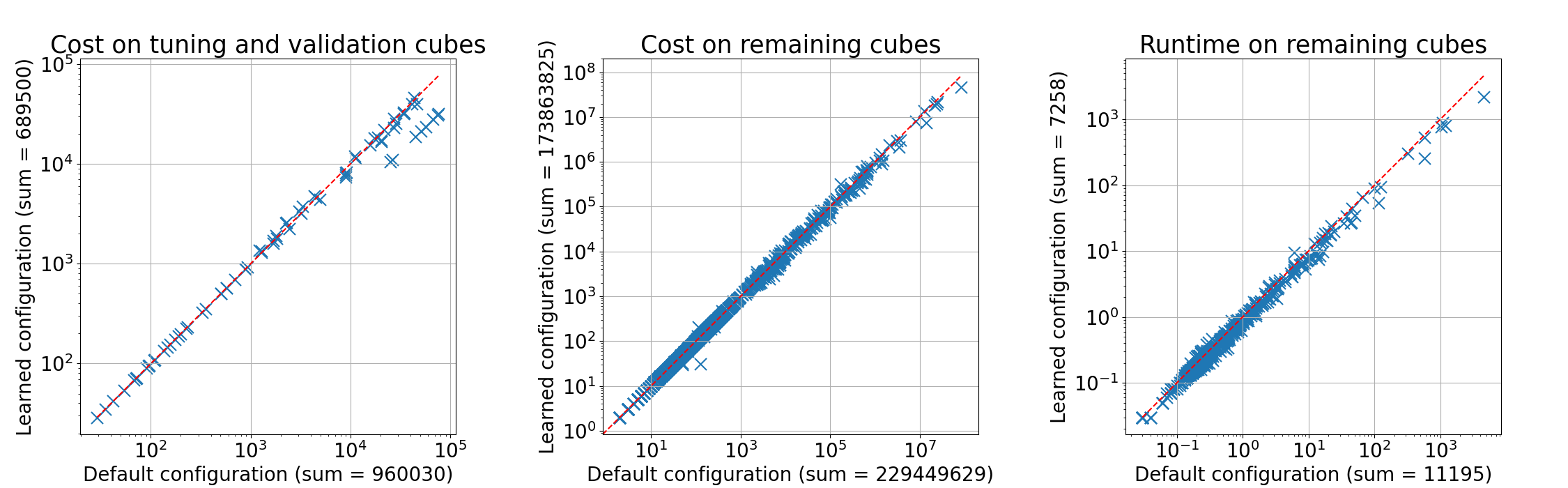}
    \caption{total-5-15-u}
    \label{fig:combined_benchmarktotal-5-15-u.cnf}
\end{figure}

\subsection{\crux benchmarks}

\begin{figure}[H]
    \centering
    \includegraphics[width=\textwidth]{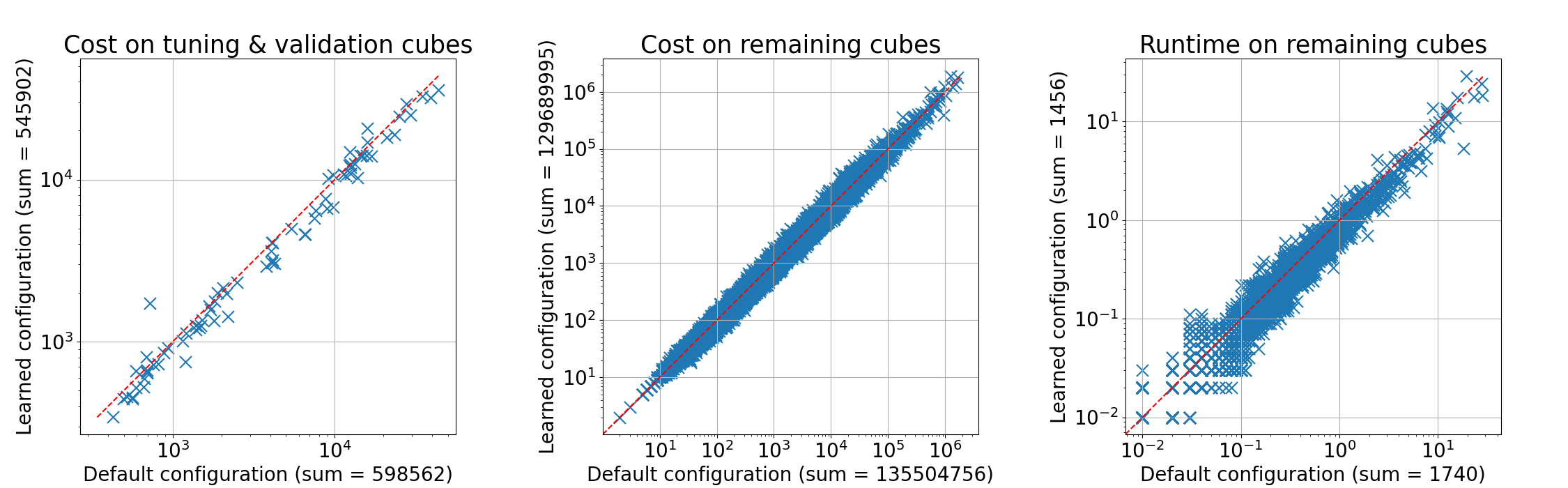}
    \caption{cruxmiter32seed0}
    \label{fig:combined_benchmarkcruxmiter32seed0.cnf}
\end{figure}

\begin{figure}[H]
    \centering
    \includegraphics[width=\textwidth]{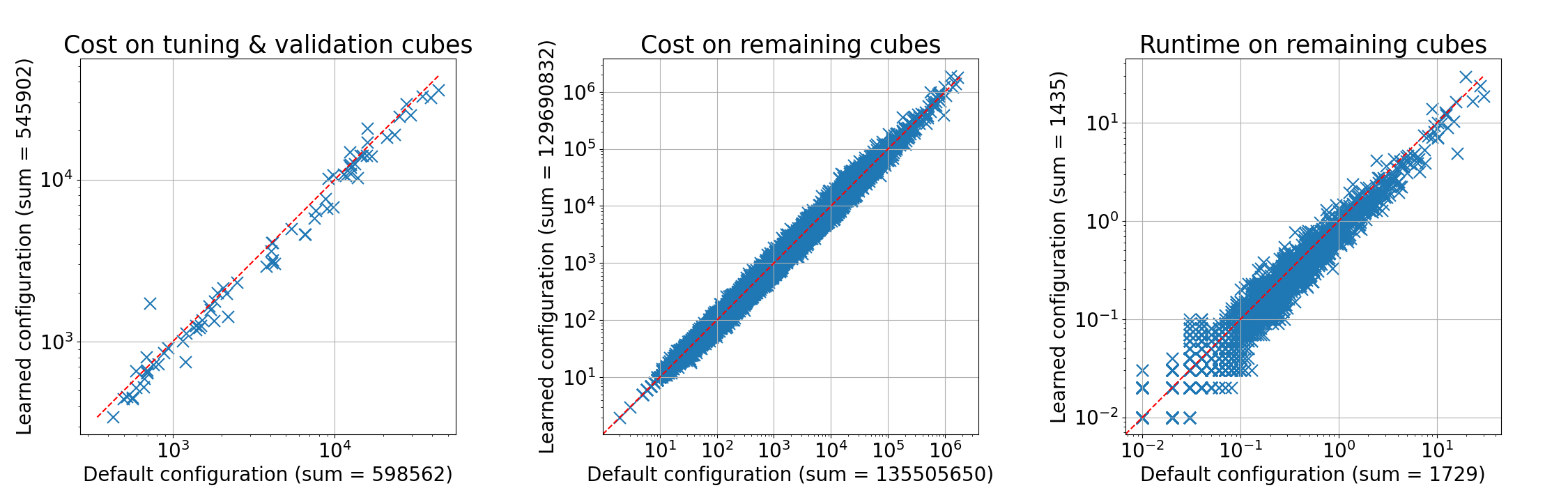}
    \caption{cruxmiter32seed1}
    \label{fig:combined_benchmarkcruxmiter32seed1.cnf}
\end{figure}

\begin{figure}[H]
    \centering
    \includegraphics[width=\textwidth]{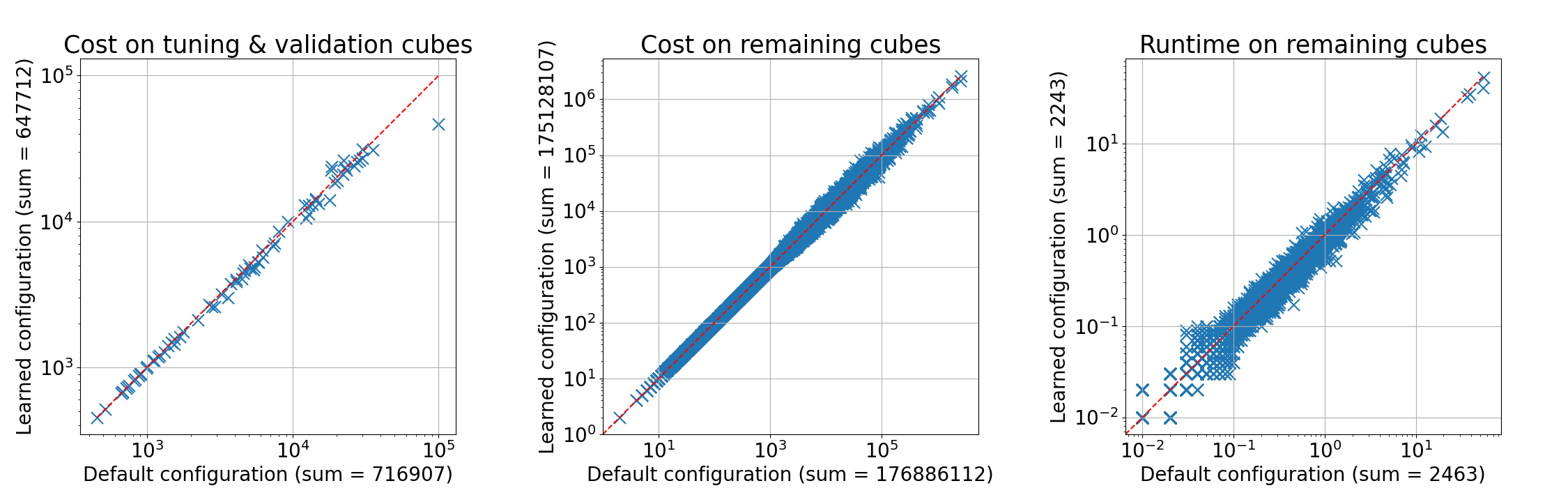}
    \caption{cruxmiter32seed2}
    \label{fig:combined_benchmarkcruxmiter32seed2.cnf}
\end{figure}

\begin{figure}[H]
    \centering
    \includegraphics[width=\textwidth]{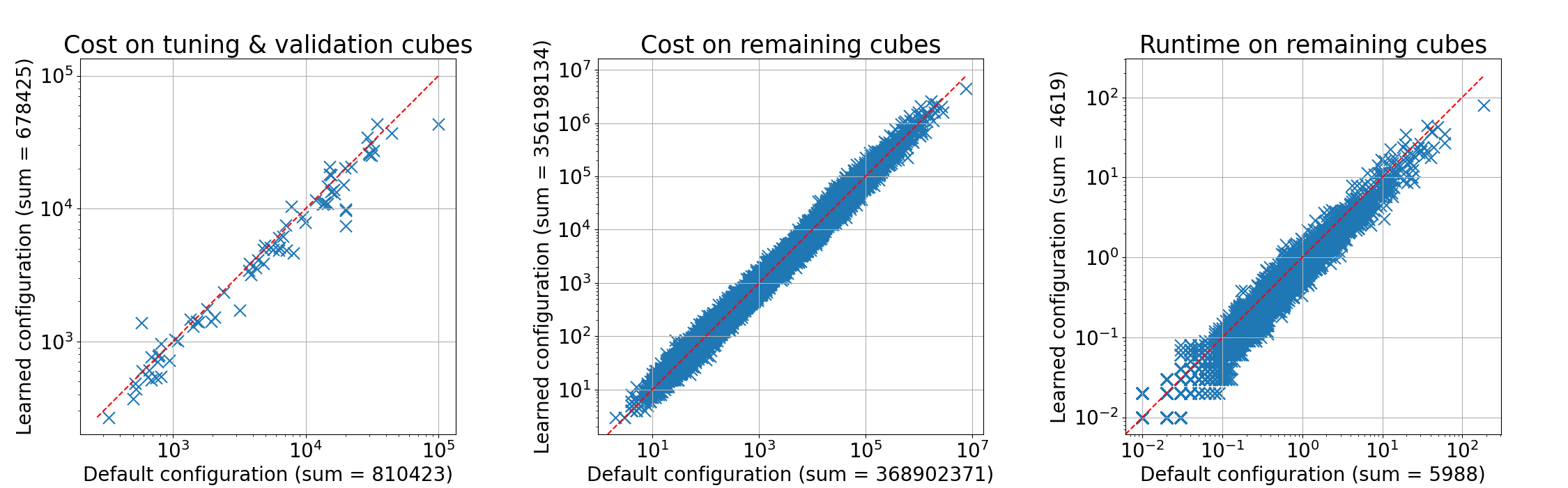}
    \caption{cruxmiter32seed3}
    \label{fig:combined_benchmarkcruxmiter32seed3.cnf}
\end{figure}

\begin{figure}[H]
    \centering
    \includegraphics[width=\textwidth]{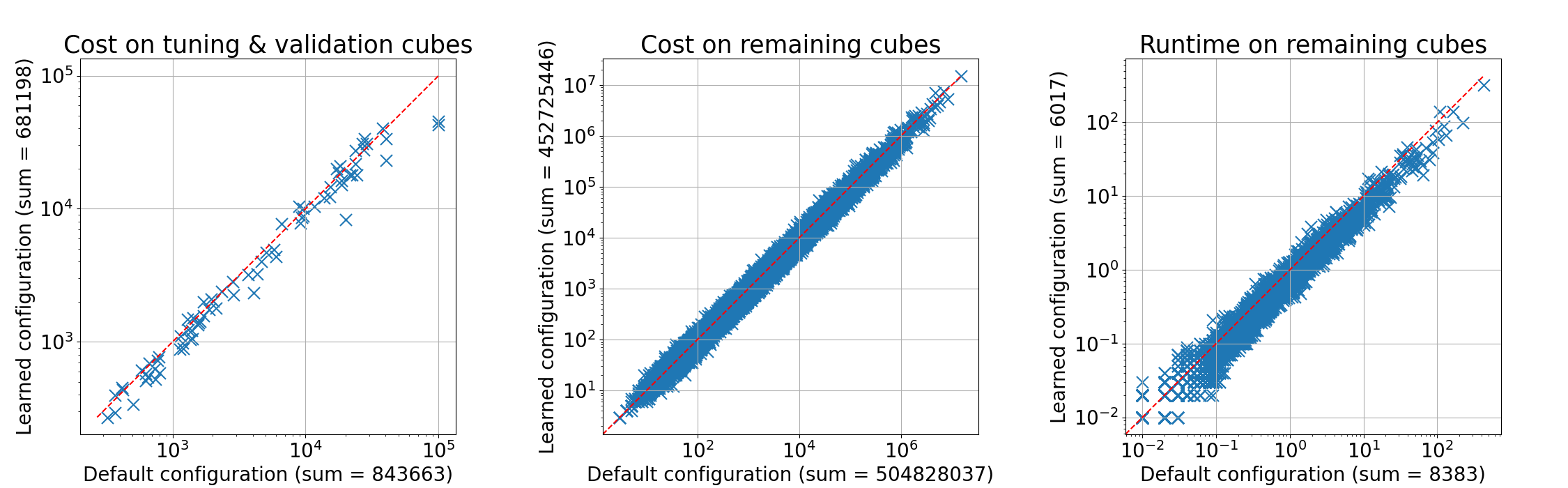}
    \caption{cruxmiter32seed4}
    \label{fig:combined_benchmarkcruxmiter32seed4.cnf}
\end{figure}

\begin{figure}[H]
    \centering
    \includegraphics[width=\textwidth]{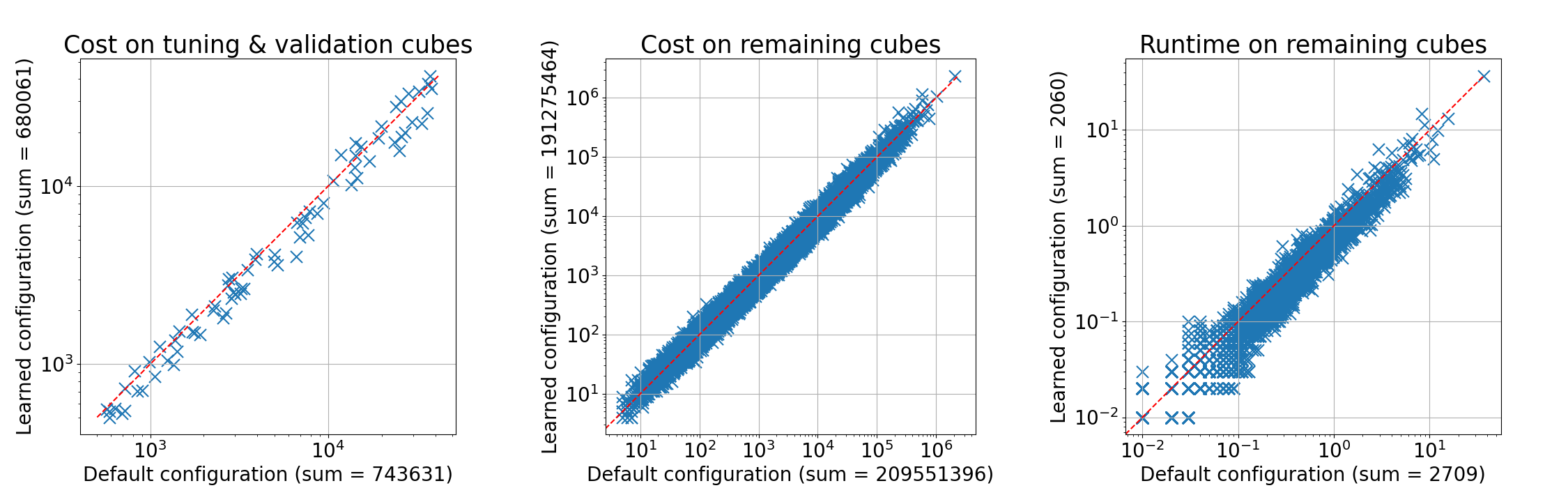}
    \caption{cruxmiter32seed5}
    \label{fig:combined_benchmarkcruxmiter32seed5.cnf}
\end{figure}

\begin{figure}[H]
    \centering
    \includegraphics[width=\textwidth]{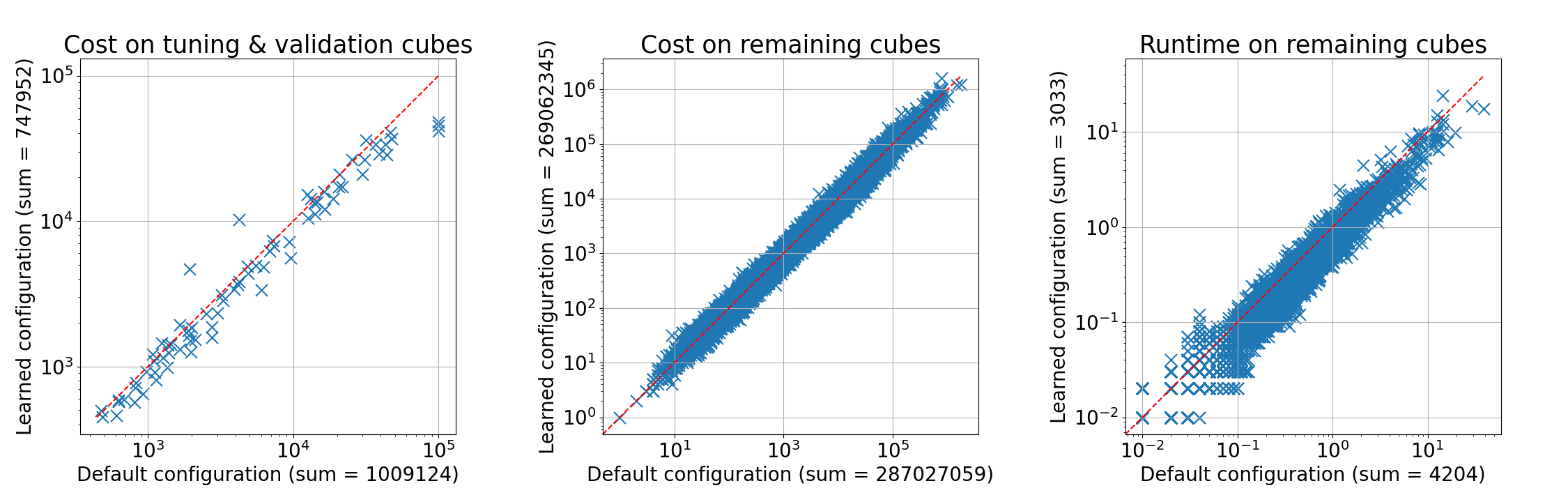}
    \caption{cruxmiter32seed6}
    \label{fig:combined_benchmarkcruxmiter32seed6.cnf}
\end{figure}

\begin{figure}[H]
    \centering
    \includegraphics[width=\textwidth]{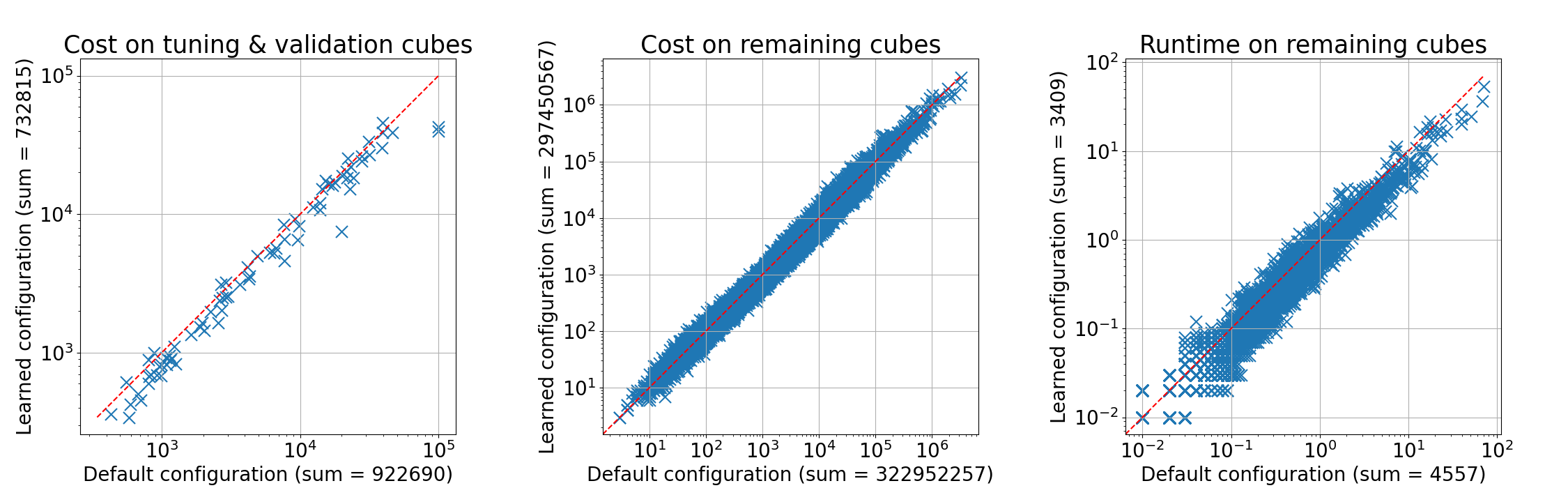}
    \caption{cruxmiter32seed7}
    \label{fig:combined_benchmarkcruxmiter32seed7.cnf}
\end{figure}

\begin{figure}[H]
    \centering
    \includegraphics[width=\textwidth]{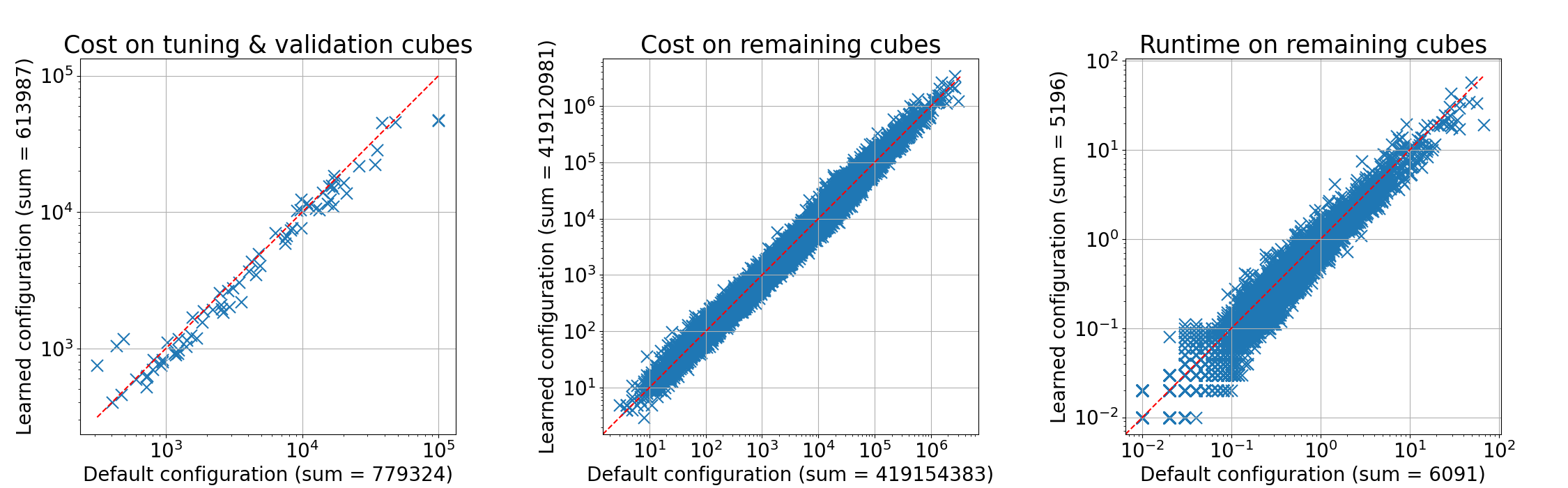}
    \caption{cruxmiter32seed8}
    \label{fig:combined_benchmarkcruxmiter32seed8.cnf}
\end{figure}

\begin{figure}[H]
    \centering
    \includegraphics[width=\textwidth]{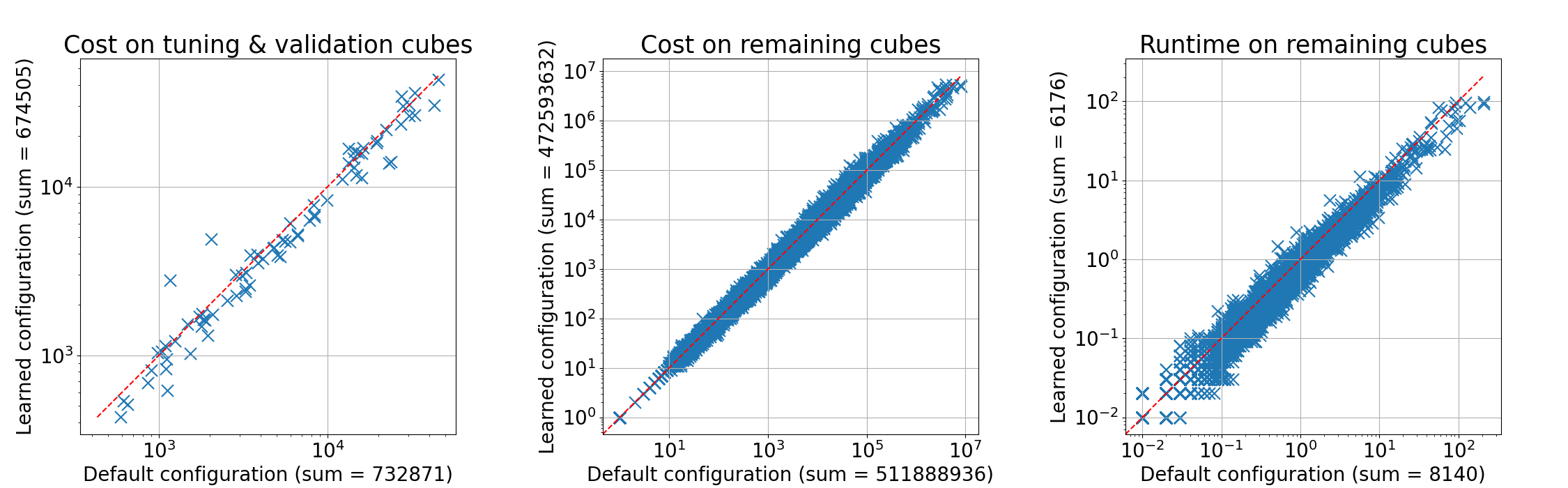}
    \caption{cruxmiter32seed9}
    \label{fig:combined_benchmarkcruxmiter32seed9.cnf}
\end{figure}

\subsection{\scBench}    

\begin{figure}[H]
    \centering
    \includegraphics[width=\textwidth]{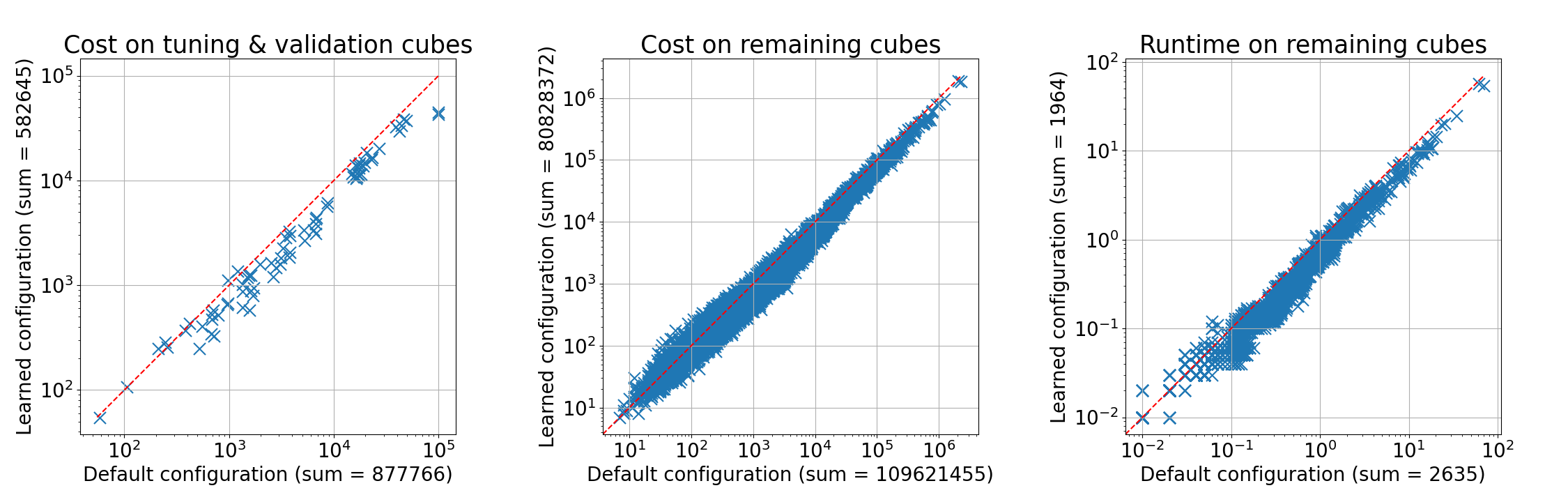}
    \caption{WS\_400\_24\_90\_10.apx\_1\_DS\_ST}
    \label{fig:combined_benchmarkWS_400_24_90_10.apx_1_DS_ST.cnf}
\end{figure}

\begin{figure}[H]
    \centering
    \includegraphics[width=\textwidth]{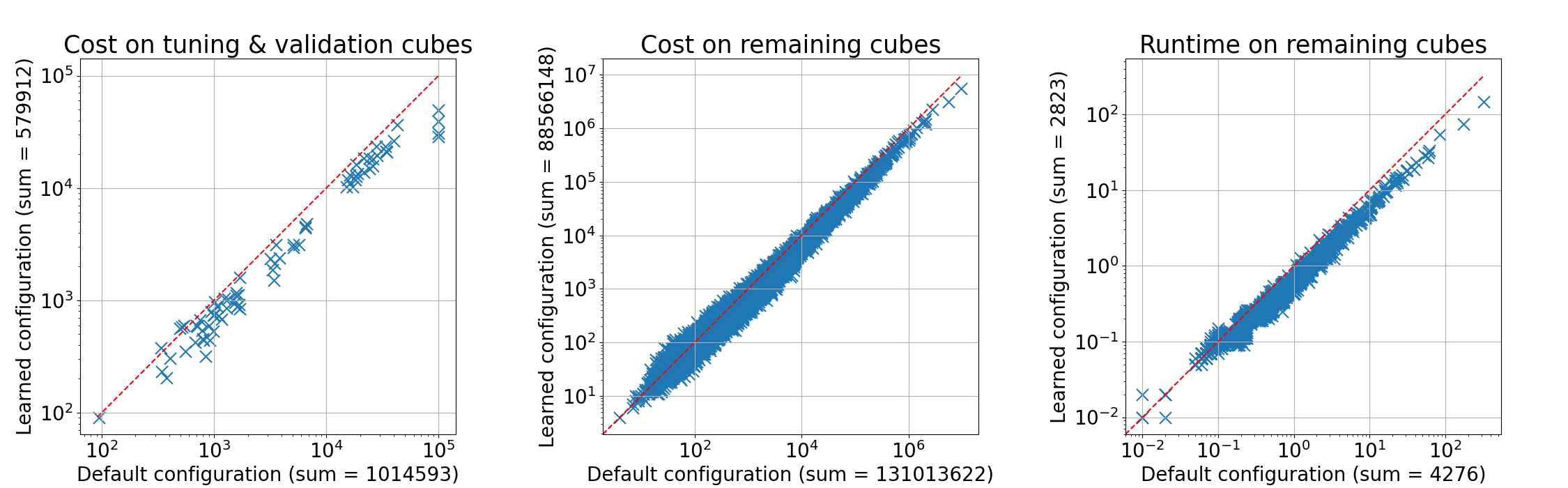}
    \caption{WS\_400\_32\_90\_10.apx\_1\_DC\_AD}
    \label{fig:combined_benchmarkWS_400_32_90_10.apx_1_DC_AD.cnf}
\end{figure}

\begin{figure}[H]
    \centering
    \includegraphics[width=\textwidth]{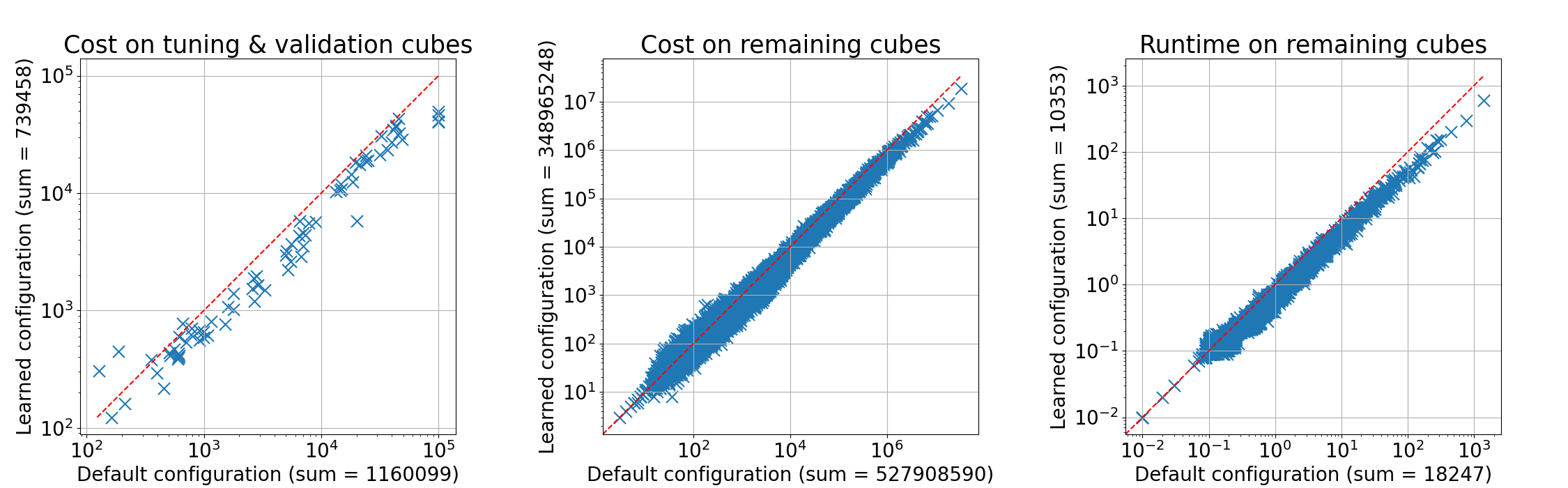}
    \caption{WS\_500\_32\_50\_10.apx\_2\_DC\_AD}
    \label{fig:combined_benchmarkWS_500_32_50_10.apx_2_DC_AD.cnf}
\end{figure}

\begin{figure}[H]
    \centering
    \includegraphics[width=\textwidth]{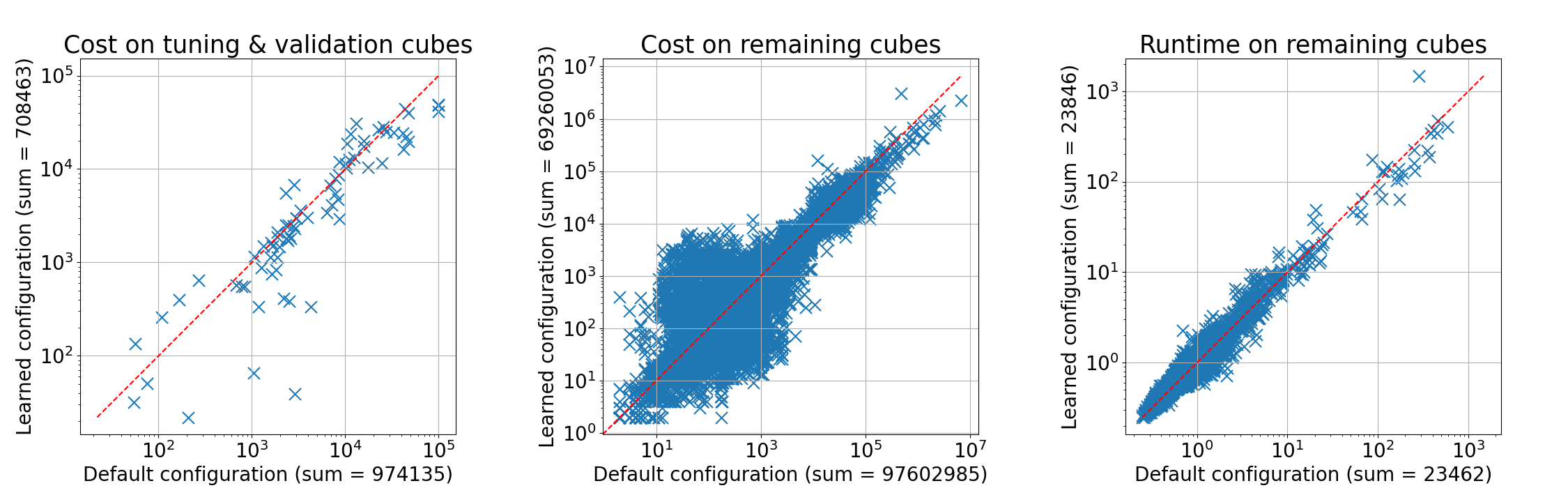}
    \caption{grs\_192\_256}
    \label{fig:combined_benchmarkgrs_192_256.cnf}
\end{figure}

\begin{figure}[H]
    \centering
    \includegraphics[width=\textwidth]{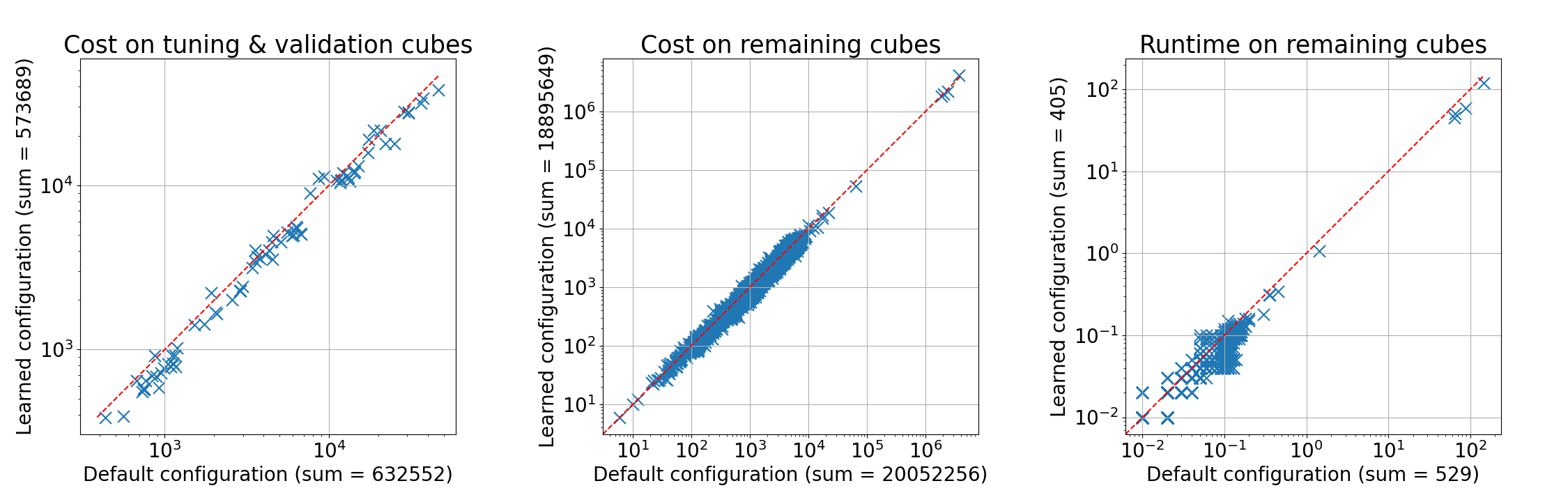}
    \caption{multiplier\_14bits\_\_miter\_14}
    \label{fig:combined_benchmarkmultiplier_14bits__miter_14.cnf}
\end{figure}

\begin{figure}[H]
\centering
\includegraphics[width=0.33\textwidth]{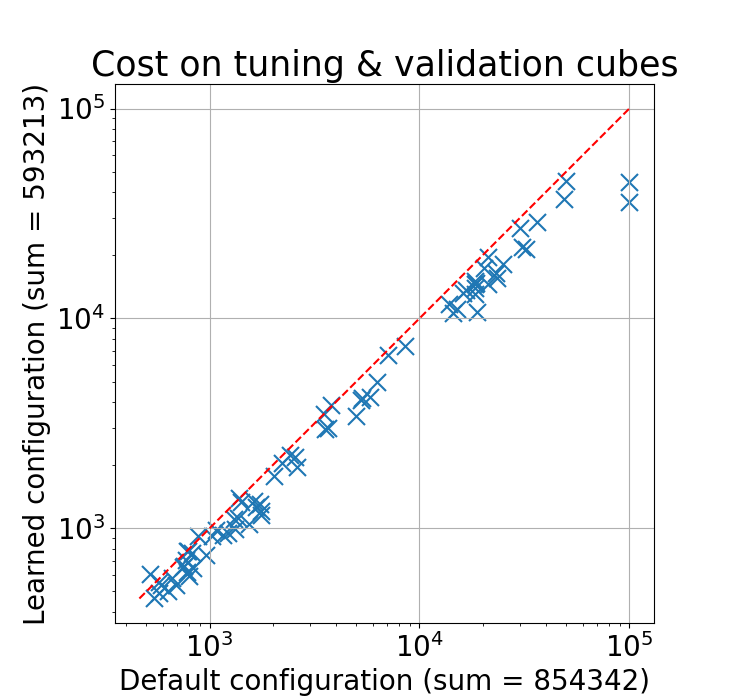}
\caption{php15\_mixed\_15percent\_blocked}
\label{fig:combined_benchmarkphp15_mixed_15percent_blocked.cnf}
\end{figure}

\begin{figure}[H]
    \centering
    \includegraphics[width=\textwidth]{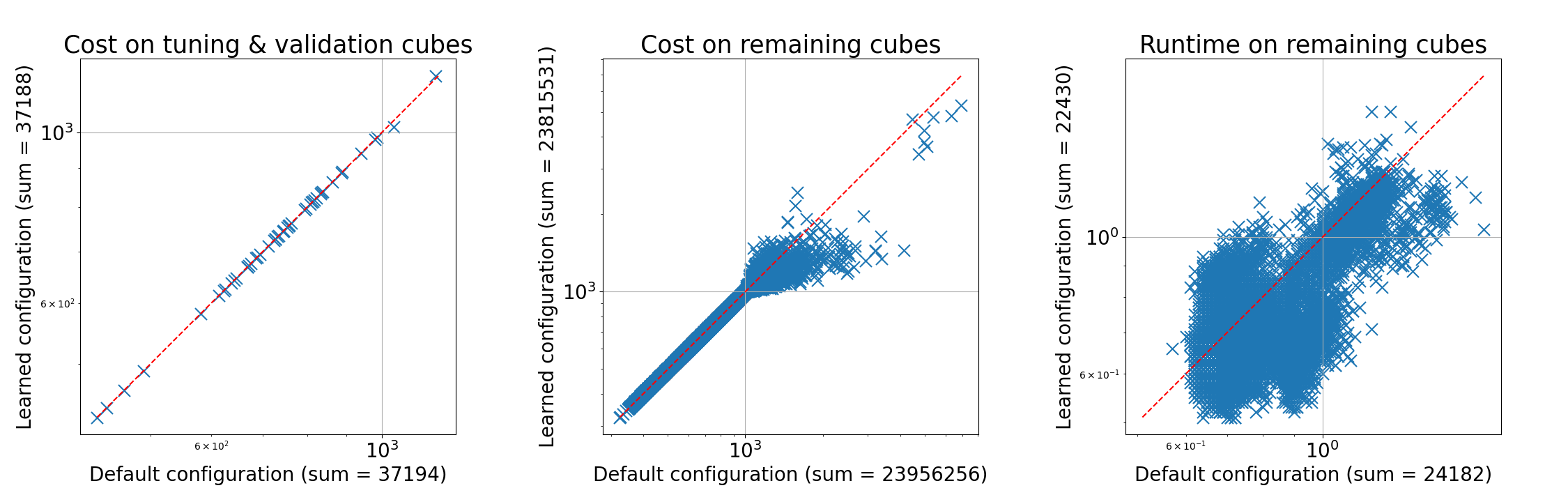}
    \caption{sin\_depth\_miter\_1}
    \label{fig:combined_benchmarksin_depth_miter_1.cnf}
\end{figure}

\newpage
\section{Evaluation on SC'24 Benchmarks}
\label{app:sc24}

For the completeness of the study, we evaluate \sys on the 400 SC'24 benchmarks. We note that \cnc is not the best strategy for solving this benchmark set: according to the competition results\footnote{\url{https://satcompetition.github.io/2024/downloads/satcomp24slides.pdf}}, within 500 seconds, sequential \kissat can already solve over half of the benchmarks, and \painless can solve over 75\%. On the other hand, \cnc is intended for solving challenging problems that cannot be efficiently tackled by a sequential solver or a portfolio strategy.

We run \cnc and \cncTuneAndVal with a one-hour wall clock timeout, 8 cores and 64G memory. As shown in Tab.~\ref{tab:sc24all}, \cncTuneAndVal solved all instances solved by \cnc, and uniquely solved 10 additional instances.

\begin{table}[h]
\centering
\caption{Comparing \cnc and \cncTuneAndVal on all benchamrks from SAT competitions 2024.}
\small
\label{tab:sc24all}
\begin{tabular}{lcccccc}
\toprule
 \tabtitle{Config.} & \tabtitle{Slv.} & 
 \tabtitle{Time} &
 \tabtitle{Uniq.} &
\tabtitle{Slv.} & 
 \tabtitle{Time} &
 \tabtitle{Uniq.} \\

 \cmidrule(lr){2-4}  \cmidrule(lr){5-7}
\cnc & 115 & 27956 & 0 & 121 & 42500  & 0 \\
\cncTuneAndVal & \best{123} & 38515 & \best{8} & \best{123} & 48507 & \best{2} \\
\bottomrule
\end{tabular}
\end{table}

A scatter plot of the runtime performance of both configurations is shown in Fig.~\ref{fig:scatter}. We observe that \sys tends to induce a small overhead for easy benchmarks (bottom left corner). On the other hand, \sys did manage to reduce the runtime for a number of harder unsatisfiable instances. Interestingly, we observe that in some cases \cncTuneAndVal finds satisfying assignment during the cube collection phase, which results in significant speed up on the corresponding satisfiable benchmarks.

\begin{figure}[h]
    \centering
\includegraphics[width=0.5\linewidth]{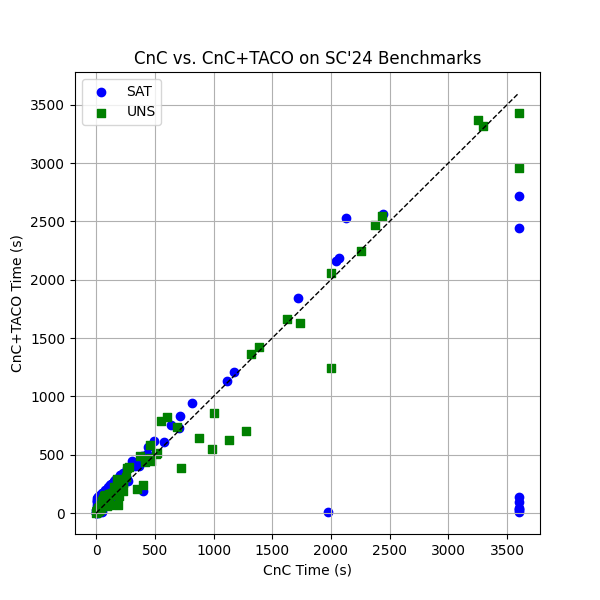}
\caption{Runtime of \cnc and \cncTuneAndVal on the SC'24 benchmarks}
\label{fig:scatter}
\end{figure}

\newpage
\section{Evaluation of sequential \sys}\label{app:seq-taco}

We also ran a sequential configuration \cncTuneAndValSeq that is the same as \cncTuneAndVal except that it spawns just one worker, and compared it against the default sequential mode of \marabou. Each job is given 8GB memory and a 4-hour wall-clock timeout. As shown in Tab.~\ref{tab:nnv-seq}, \cncTuneAndValSeq solves significantly more \nap benchmarks than \marabou. This is again due to the cost-limited breadth-first search pattern enabled by \sys. On the other hand, \cncTuneAndValSeq solved all of the 259 \altloop benchmarks, while \marabou timed out on one of them. The average runtime of \cncTuneAndValSeq is more than 2x faster than \marabou on the \altloop benchmarks. This suggests that \sys has the potential to improve the performance of an automated reasoning procedure in a sequential solving setting as well. 

\begin{table}[h]
\centering
\caption{Comparing sequential \sys with \marabou. Both configurations were given 1 core and a 4-hour timeout.} 
\small
\label{tab:nnv-seq}
\begin{tabular}{ccccccc}
\toprule
 & \multicolumn{3}{c}{\tabtitle{\cncTuneAndValSeq}} & 
 \multicolumn{2}{c}{\tabtitle{\marabou}} \\
\cmidrule(lr){2-4}\cmidrule(lr){5-6}
\tabtitle{Family} (\#) & 
\tabtitle{\colsolved} &
\tabtitle{\coltimetotal} & \tabtitle{\coltimelearning} & \tabtitle{\colsolved} & \tabtitle{\coltimetotal} \\
\benchmark{NAP} (235) 
& \best{209} & \best{301089} & 48605 & 113 & 669170 \\
\benchmark{Altloop} (259)
& \best{259} & \best{104182} & 19243 & 258 & 243271 \\
\bottomrule
\end{tabular}
\end{table}

\newpage
\section{Detailed Results on Neural Network Verification Benchmarks}\label{app:detailed-nnv}

\subsection{\cnc vs. \cncTuneAndVal}
\begin{center}
\scriptsize
\begin{xltabular}{\textwidth}{@{}llrrrrrrrr@{}}
\toprule
& & \multicolumn{2}{c}{\tabtitle{\cnc}} & \multicolumn{4}{c}{\tabtitle{\cncTuneAndVal}} &
 \multicolumn{2}{c}{\tabtitle{\dncmarabou}} 
\\
\cmidrule(lr){3-4} \cmidrule(lr){5-8} \cmidrule(lr){9-10}
\tabtitle{Family} & \tabtitle{Benchmark}  & \tabtitle{Res.} & \coltimetotal & \tabtitle{Res.} & \coltimetotal & \coltimelearning & \tabtitle{Config.} & \tabtitle{Res.} & \coltimetotal \\
AltLoop & REI.id76.ep90543 & UNS & 69.4 & UNS & 78.7 & 15.2 &  & UNS & 33.3 \\
AltLoop & REI.id35.ep94624 & SAT & 0.3 & SAT & 0.6 & 0.4 & (slv. during collect.) & SAT & 0.3 \\
AltLoop & REI.id267.ep96300 & UNS & 14.3 & UNS & 17.8 & 17.7 & (slv. during collect.) & UNS & 0.3 \\
AltLoop & REI.id321.ep55647 & UNS & 20.4 & UNS & 37.3 & 36.4 &  & UNS & 6.5 \\
AltLoop & REI.id91.ep95105 & UNS & 104.5 & UNS & 310.0 & 16.3 &  & UNS & 66.4 \\
AltLoop & REI.id176.ep98430 & UNS & 143.6 & UNS & 191.3 & 17.8 &  & UNS & 98.3 \\
AltLoop & REI.id176.ep96138 & UNS & 61.2 & UNS & 68.9 & 12.6 &  & UNS & 40.5 \\
AltLoop & REI.id73.ep84023 & UNS & 72.4 & UNS & 101.0 & 17.0 &  & UNS & 52.1 \\
AltLoop & REI.id343.ep63040 & UNS & 23.1 & UNS & 40.0 & 22.2 & br.=pol. p.s.f=1 & UNS & 9.9 \\
AltLoop & REI.id90.ep84292 & UNS & 212.0 & UNS & 234.4 & 15.9 &  & UNS & 545.9 \\
AltLoop & REI.id114.ep94581 & SAT & 13.4 & SAT & 0.8 & 0.6 & (slv. during collect.) & SAT & 0.4 \\
AltLoop & REI.id180.ep69795 & SAT & 0.6 & SAT & 1.2 & 1.0 & (slv. during collect.) & SAT & 0.3 \\
AltLoop & REI.id215.ep97421 & SAT & 35.4 & SAT & 2.8 & 2.7 & (slv. during collect.) & SAT & 0.3 \\
AltLoop & REI.id296.ep74854 & UNS & 44.0 & UNS & 56.6 & 13.9 &  & UNS & 17.8 \\
AltLoop & REI.id298.ep85589 & UNS & 76.2 & UNS & 94.2 & 13.9 &  & UNS & 63.1 \\
AltLoop & REI.id299.ep81175 & UNS & 40.2 & UNS & 56.8 & 17.8 &  & UNS & 21.6 \\
AltLoop & REI.id318.ep94480 & SAT & 111.7 & SAT & 2.0 & 1.9 & (slv. during collect.) & SAT & 0.3 \\
AltLoop & REI.id535.ep87226 & SAT & 0.8 & SAT & 1.7 & 1.6 & (slv. during collect.) & SAT & 0.4 \\
AltLoop & REI.id530.ep82130 & SAT & 0.5 & SAT & 0.7 & 0.6 & (slv. during collect.) & SAT & 0.3 \\
AltLoop & REI.id76.ep91480 & UNS & 65.5 & UNS & 112.8 & 18.8 &  & UNS & 44.0 \\
AltLoop & REI.id77.ep49509 & UNS & 34.4 & UNS & 48.2 & 18.7 &  & UNS & 12.9 \\
AltLoop & REI.id518.ep95639 & UNS & 54.8 & UNS & 63.2 & 16.0 &  & UNS & 36.6 \\
AltLoop & REI.id240.ep89317 & UNS & 31.9 & UNS & 43.1 & 17.0 & br.=pol. p.s.f=2 & UNS & 20.8 \\
AltLoop & REI.id540.ep94690 & SAT & 0.5 & SAT & 0.7 & 0.5 & (slv. during collect.) & SAT & 0.4 \\
AltLoop & REI.id73.ep91504 & SAT & 0.5 & SAT & 0.7 & 0.6 & (slv. during collect.) & SAT & 0.3 \\
AltLoop & REI.id76.ep98692 & UNS & 52.5 & UNS & 56.2 & 15.2 &  & SAT & 2.7 \\
AltLoop & REI.id35.ep85842 & SAT & 0.4 & SAT & 0.6 & 0.5 & (slv. during collect.) & SAT & 0.3 \\
AltLoop & REI.id518.ep95900 & UNS & 68.7 & UNS & 98.9 & 19.2 &  & UNS & 51.5 \\
AltLoop & REI.id343.ep63051 & UNS & 24.9 & UNS & 55.6 & 32.1 & br.=pol. p.s.f=1 & UNS & 10.1 \\
AltLoop & REI.id518.ep70223 & UNS & 47.2 & UNS & 58.7 & 13.0 &  & UNS & 25.1 \\
AltLoop & REI.id298.ep79143 & UNS & 79.2 & UNS & 62.2 & 18.9 & p.s.f=1 & UNS & 84.3 \\
AltLoop & REI.id512.ep97385 & SAT & 30.1 & SAT & 1.7 & 1.5 & (slv. during collect.) & SAT & 0.4 \\
AltLoop & REI.id530.ep91636 & SAT & 0.3 & SAT & 0.5 & 0.4 & (slv. during collect.) & SAT & 0.3 \\
AltLoop & REI.id35.ep61788 & SAT & 0.5 & SAT & 1.6 & 1.5 & (slv. during collect.) & SAT & 0.3 \\
AltLoop & REI.id67.ep87611 & UNS & 33.0 & UNS & 46.2 & 18.9 & p.s.f=5 & UNS & 23.6 \\
AltLoop & REI.id535.ep91323 & UNS & 64.7 & UNS & 57.2 & 12.4 &  & UNS & 13.7 \\
AltLoop & REI.id491.ep96326 & SAT & 135.8 & SAT & 1.4 & 1.3 & (slv. during collect.) & SAT & 0.4 \\
AltLoop & REI.id491.ep67924 & SAT & 0.5 & SAT & 0.5 & 0.4 & (slv. during collect.) & SAT & 0.4 \\
AltLoop & REI.id174.ep97933 & SAT & 0.4 & SAT & 0.5 & 0.4 & (slv. during collect.) & SAT & 0.3 \\
AltLoop & REI.id114.ep94159 & SAT & 15.8 & SAT & 0.8 & 0.6 & (slv. during collect.) & SAT & 0.5 \\
AltLoop & REI.id165.ep96974 & SAT & 14.5 & SAT & 2.2 & 2.1 & (slv. during collect.) & SAT & 0.6 \\
AltLoop & REI.id535.ep76916 & UNS & 47.8 & UNS & 56.0 & 13.3 &  & UNS & 15.6 \\
AltLoop & REI.id196.ep86142 & SAT & 0.6 & SAT & 0.9 & 0.8 & (slv. during collect.) & SAT & 0.4 \\
AltLoop & REI.id240.ep98891 & UNS & 46.2 & UNS & 71.5 & 20.6 & br.=pol. p.s.f=1 & UNS & 27.2 \\
AltLoop & REI.id298.ep75768 & UNS & 75.7 & UNS & 109.8 & 15.9 &  & UNS & 75.5 \\
AltLoop & REI.id91.ep54663 & UNS & 34.7 & UNS & 66.4 & 20.1 &  & UNS & 19.1 \\
AltLoop & REI.id502.ep71027 & SAT & 0.6 & SAT & 0.7 & 0.6 & (slv. during collect.) & SAT & 0.4 \\
AltLoop & REI.id540.ep57595 & UNS & 431.8 & UNS & 118.8 & 22.1 & p.s.f=1 & SAT & 483.1 \\
AltLoop & REI.id352.ep51510 & UNS & 24.9 & UNS & 30.7 & 22.6 &  & UNS & 6.9 \\
AltLoop & REI.id549.ep87887 & SAT & 0.6 & SAT & 1.0 & 0.8 & (slv. during collect.) & SAT & 0.4 \\
AltLoop & REI.id399.ep96449 & SAT & 91.4 & SAT & 111.2 & 17.4 &  & SAT & 5.3 \\
AltLoop & REI.id165.ep95976 & SAT & 3.4 & SAT & 1.1 & 0.9 & (slv. during collect.) & SAT & 0.5 \\
AltLoop & REI.id158.ep53499 & SAT & 26.8 & SAT & 32.4 & 16.9 & p.s.f=1 & SAT & 0.5 \\
AltLoop & REI.id201.ep48219 & SAT & 4.6 & SAT & 0.7 & 0.5 & (slv. during collect.) & SAT & 0.3 \\
AltLoop & REI.id165.ep94522 & SAT & 3.1 & SAT & 1.0 & 0.9 & (slv. during collect.) & SAT & 0.4 \\
AltLoop & REI.id201.ep48186 & SAT & 4.5 & SAT & 0.5 & 0.4 & (slv. during collect.) & SAT & 0.3 \\
AltLoop & REI.id308.ep94787 & UNS & 47.0 & UNS & 68.4 & 13.7 &  & UNS & 23.4 \\
AltLoop & REI.id457.ep79411 & SAT & 0.5 & SAT & 0.5 & 0.4 & (slv. during collect.) & SAT & 0.4 \\
AltLoop & REI.id303.ep94994 & SAT & 0.6 & SAT & 1.9 & 1.7 & (slv. during collect.) & SAT & 0.3 \\
AltLoop & REI.id180.ep77624 & SAT & 0.6 & SAT & 1.3 & 1.2 & (slv. during collect.) & SAT & 0.3 \\
AltLoop & REI.id457.ep81404 & SAT & 0.5 & SAT & 0.6 & 0.4 & (slv. during collect.) & SAT & 0.3 \\
AltLoop & REI.id239.ep73100 & UNS & 217.7 & UNS & 82.9 & 18.2 & p.s.f=1 & UNS & 309.1 \\
AltLoop & REI.id215.ep96277 & SAT & 0.4 & SAT & 0.7 & 0.6 & (slv. during collect.) & SAT & 0.3 \\
AltLoop & REI.id299.ep92705 & UNS & 114.8 & UNS & 71.7 & 20.3 & p.s.f=1 & UNS & 54.9 \\
AltLoop & REI.id308.ep97397 & SAT & 12.3 & SAT & 25.6 & 11.4 &  & SAT & 0.3 \\
AltLoop & REI.id35.ep61486 & SAT & 1.8 & SAT & 1.5 & 1.4 & (slv. during collect.) & SAT & 0.5 \\
AltLoop & REI.id512.ep70217 & SAT & 17.2 & SAT & 32.8 & 13.9 &  & SAT & 5.5 \\
AltLoop & REI.id240.ep98115 & UNS & 53.0 & UNS & 50.0 & 16.8 & p.s.f=1 & UNS & 30.6 \\
AltLoop & REI.id319.ep84434 & SAT & 0.3 & SAT & 0.7 & 0.5 & (slv. during collect.) & SAT & 0.4 \\
AltLoop & REI.id158.ep53113 & SAT & 11.2 & SAT & 4.6 & 4.5 & (slv. during collect.) & SAT & 0.4 \\
AltLoop & REI.id81.ep89983 & UNS & 97.3 & UNS & 96.6 & 14.1 &  & UNS & 142.9 \\
AltLoop & REI.id444.ep96160 & UNS & 69.0 & UNS & 89.3 & 11.7 &  & UNS & 50.1 \\
AltLoop & REI.id158.ep63427 & UNS & 28.8 & UNS & 40.1 & 20.4 & br.=pol. p.s.f=1 & UNS & 14.3 \\
AltLoop & REI.id491.ep70134 & SAT & 3.4 & SAT & 0.6 & 0.4 & (slv. during collect.) & SAT & 0.3 \\
AltLoop & REI.id321.ep93439 & UNS & 28.1 & UNS & 38.3 & 27.8 & p.s.f=5 & UNS & 9.2 \\
AltLoop & REI.id239.ep90118 & SAT & 0.4 & SAT & 1.7 & 1.6 & (slv. during collect.) & SAT & 0.3 \\
AltLoop & REI.id530.ep94463 & SAT & 0.6 & SAT & 0.6 & 0.5 & (slv. during collect.) & SAT & 0.3 \\
AltLoop & REI.id286.ep94925 & UNS & 36.0 & UNS & 48.0 & 17.3 &  & UNS & 353.6 \\
AltLoop & REI.id549.ep99098 & SAT & 0.5 & SAT & 0.8 & 0.7 & (slv. during collect.) & SAT & 0.3 \\
AltLoop & REI.id267.ep96426 & UNS & 15.2 & UNS & 16.2 & 16.1 & (slv. during collect.) & UNS & 0.3 \\
AltLoop & REI.id319.ep86630 & SAT & 0.3 & SAT & 0.6 & 0.4 & (slv. during collect.) & SAT & 0.3 \\
AltLoop & REI.id540.ep57610 & UNS & 470.0 & UNS & 117.5 & 22.9 & p.s.f=1 & UNS & 1300.8 \\
AltLoop & REI.id528.ep89331 & SAT & 0.4 & SAT & 3.4 & 3.2 & (slv. during collect.) & SAT & 0.5 \\
AltLoop & REI.id425.ep83340 & SAT & 0.4 & SAT & 0.7 & 0.6 & (slv. during collect.) & SAT & 0.3 \\
AltLoop & REI.id399.ep69315 & UNS & 26.2 & UNS & 42.2 & 17.7 &  & UNS & 13.4 \\
AltLoop & REI.id267.ep95713 & UNS & 17.4 & UNS & 13.6 & 13.5 & (slv. during collect.) & UNS & 0.3 \\
AltLoop & REI.id215.ep97333 & SAT & 0.5 & SAT & 3.1 & 3.0 & (slv. during collect.) & SAT & 0.3 \\
AltLoop & REI.id298.ep97154 & SAT & 0.5 & SAT & 1.3 & 1.2 & (slv. during collect.) & SAT & 0.3 \\
AltLoop & REI.id299.ep99426 & UNS & 55.0 & UNS & 55.4 & 19.4 & p.s.f=1 & UNS & 33.5 \\
AltLoop & REI.id319.ep84478 & SAT & 0.4 & SAT & 0.6 & 0.5 & (slv. during collect.) & SAT & 0.3 \\
AltLoop & REI.id379.ep80359 & SAT & 7.5 & SAT & 3.1 & 3.0 & (slv. during collect.) & SAT & 0.3 \\
AltLoop & REI.id549.ep99331 & SAT & 0.4 & SAT & 0.8 & 0.6 & (slv. during collect.) & SAT & 0.3 \\
AltLoop & REI.id535.ep87193 & SAT & 0.6 & SAT & 1.6 & 1.4 & (slv. during collect.) & SAT & 0.4 \\
AltLoop & REI.id379.ep71969 & SAT & 95.6 & SAT & 128.9 & 16.6 &  & SAT & 9.1 \\
AltLoop & REI.id425.ep93745 & UNS & 406.6 & UNS & 311.9 & 19.1 & p.s.f=5 & UNS & 364.3 \\
AltLoop & REI.id381.ep85470 & SAT & 103.8 & SAT & 3.6 & 3.4 & (slv. during collect.) & SAT & 37.2 \\
AltLoop & REI.id540.ep57775 & UNS & 707.5 & UNS & 147.7 & 27.2 & p.s.f=1 & UNS & 2343.8 \\
AltLoop & REI.id425.ep96981 & SAT & 0.6 & SAT & 0.8 & 0.6 & (slv. during collect.) & SAT & 0.3 \\
AltLoop & REI.id128.ep68334 & UNS & 109.9 & UNS & 74.2 & 18.5 & p.s.f=1 & UNS & 69.7 \\
AltLoop & REI.id318.ep85341 & SAT & 23.7 & SAT & 36.9 & 14.5 &  & SAT & 0.5 \\
AltLoop & REI.id68.ep95777 & SAT & 0.6 & SAT & 0.7 & 0.5 & (slv. during collect.) & SAT & 0.4 \\
AltLoop & REI.id393.ep90550 & SAT & 16.5 & SAT & 7.7 & 7.5 & (slv. during collect.) & SAT & 0.3 \\
AltLoop & REI.id393.ep97111 & SAT & 0.6 & SAT & 19.6 & 18.8 &  & SAT & 0.4 \\
AltLoop & REI.id286.ep95719 & UNS & 33.3 & UNS & 46.3 & 16.6 &  & UNS & 11.8 \\
AltLoop & REI.id393.ep97199 & SAT & 17.9 & SAT & 19.1 & 17.8 & p.s.f=1 & SAT & 0.5 \\
AltLoop & REI.id502.ep84337 & SAT & 0.6 & SAT & 0.6 & 0.5 & (slv. during collect.) & SAT & 0.4 \\
AltLoop & REI.id165.ep94134 & SAT & 2.7 & SAT & 0.9 & 0.8 & (slv. during collect.) & SAT & 0.4 \\
AltLoop & REI.id239.ep79112 & UNS & 247.7 & UNS & 87.8 & 20.2 & p.s.f=1 & UNS & 616.5 \\
AltLoop & REI.id444.ep78487 & UNS & 188.7 & UNS & 224.3 & 18.1 &  & UNS & 274.1 \\
AltLoop & REI.id502.ep90122 & SAT & 0.5 & SAT & 0.7 & 0.5 & (slv. during collect.) & SAT & 0.3 \\
AltLoop & REI.id90.ep99033 & UNS & 205.7 & UNS & 231.8 & 14.2 &  & UNS & 218.5 \\
AltLoop & REI.id319.ep62661 & SAT & 0.4 & SAT & 0.7 & 0.6 & (slv. during collect.) & SAT & 0.4 \\
AltLoop & REI.id512.ep70283 & SAT & 16.7 & SAT & 35.7 & 18.0 &  & SAT & 29.7 \\
AltLoop & REI.id457.ep46821 & SAT & 0.5 & SAT & 0.7 & 0.6 & (slv. during collect.) & SAT & 0.4 \\
AltLoop & REI.id308.ep98363 & SAT & 17.9 & SAT & 4.9 & 4.7 & (slv. during collect.) & SAT & 0.5 \\
AltLoop & REI.id502.ep86449 & SAT & 0.7 & SAT & 0.7 & 0.6 & (slv. during collect.) & SAT & 0.3 \\
AltLoop & REI.id303.ep95038 & SAT & 0.4 & SAT & 1.9 & 1.7 & (slv. during collect.) & SAT & 0.4 \\
AltLoop & REI.id444.ep87148 & UNS & 94.9 & UNS & 110.6 & 10.7 &  & UNS & 114.1 \\
AltLoop & REI.id318.ep75036 & SAT & 116.5 & SAT & 1.6 & 1.5 & (slv. during collect.) & SAT & 0.3 \\
AltLoop & REI.id47.ep99600 & SAT & 2.5 & SAT & 2.8 & 2.7 & (slv. during collect.) & SAT & 0.3 \\
AltLoop & REI.id399.ep92145 & UNS & 47.0 & UNS & 48.4 & 17.5 & p.s.f=1 & UNS & 25.1 \\
AltLoop & REI.id176.ep99427 & UNS & 93.5 & UNS & 116.0 & 13.5 &  & UNS & 71.4 \\
AltLoop & REI.id457.ep78403 & SAT & 0.5 & SAT & 0.6 & 0.4 & (slv. during collect.) & SAT & 0.4 \\
AltLoop & REI.id491.ep54571 & SAT & 17.9 & SAT & 0.6 & 0.4 & (slv. during collect.) & SAT & 0.4 \\
AltLoop & REI.id286.ep89988 & UNS & 30.9 & UNS & 37.6 & 16.3 &  & UNS & 11.2 \\
AltLoop & REI.id234.ep71965 & SAT & 0.5 & SAT & 0.9 & 0.8 & (slv. during collect.) & SAT & 0.4 \\
AltLoop & REI.id176.ep92206 & UNS & 229.5 & UNS & 113.5 & 23.5 & p.s.f=1 & UNS & 313.6 \\
AltLoop & REI.id114.ep97993 & SAT & 15.1 & SAT & 0.7 & 0.6 & (slv. during collect.) & SAT & 0.3 \\
AltLoop & REI.id352.ep70312 & UNS & 19.8 & UNS & 32.5 & 31.5 & br.=pol. p.s.f=1 & UNS & 6.9 \\
AltLoop & REI.id196.ep41633 & SAT & 0.5 & SAT & 0.7 & 0.6 & (slv. during collect.) & SAT & 0.3 \\
AltLoop & REI.id518.ep68431 & UNS & 52.4 & UNS & 67.0 & 15.3 &  & UNS & 28.0 \\
AltLoop & REI.id518.ep96115 & UNS & 50.8 & UNS & 146.4 & 18.5 & br.=pol. p.s.f=1 & UNS & 42.1 \\
AltLoop & REI.id549.ep85789 & SAT & 0.7 & SAT & 0.6 & 0.5 & (slv. during collect.) & SAT & 0.3 \\
AltLoop & REI.id73.ep74414 & SAT & 0.7 & SAT & 2.5 & 2.5 & (slv. during collect.) & SAT & 0.5 \\
AltLoop & REI.id296.ep79125 & UNS & 44.5 & UNS & 52.1 & 12.3 &  & UNS & 18.1 \\
AltLoop & REI.id502.ep71445 & SAT & 0.5 & SAT & 0.6 & 0.5 & (slv. during collect.) & SAT & 0.3 \\
AltLoop & REI.id215.ep97355 & SAT & 0.4 & SAT & 2.7 & 2.6 & (slv. during collect.) & SAT & 0.3 \\
AltLoop & REI.id239.ep73504 & SAT & 9.3 & SAT & 18.9 & 16.4 & p.s.f=1 & SAT & 0.3 \\
AltLoop & REI.id81.ep91571 & UNS & 117.1 & UNS & 455.2 & 16.3 & br.=pol. p.s.f=1 & UNS & 362.6 \\
AltLoop & REI.id298.ep77187 & SAT & 5.8 & SAT & 23.6 & 22.1 & p.s.f=1 & SAT & 4.3 \\
AltLoop & REI.id444.ep77347 & UNS & 161.8 & UNS & 172.3 & 12.1 &  & UNS & 161.0 \\
AltLoop & REI.id530.ep82185 & SAT & 0.3 & SAT & 0.8 & 0.6 & (slv. during collect.) & SAT & 0.3 \\
AltLoop & REI.id73.ep74469 & SAT & 0.6 & SAT & 3.0 & 2.9 & (slv. during collect.) & SAT & 0.4 \\
AltLoop & REI.id530.ep82207 & SAT & 0.6 & SAT & 0.6 & 0.5 & (slv. during collect.) & SAT & 0.3 \\
AltLoop & REI.id68.ep95854 & SAT & 0.5 & SAT & 0.7 & 0.6 & (slv. during collect.) & SAT & 0.3 \\
AltLoop & REI.id303.ep95137 & SAT & 2.7 & SAT & 5.2 & 5.1 & (slv. during collect.) & SAT & 0.4 \\
AltLoop & REI.id381.ep93562 & UNS & 182.6 & UNS & 230.1 & 18.7 &  & UNS & 322.1 \\
AltLoop & REI.id77.ep99229 & UNS & 61.6 & UNS & 51.7 & 16.7 & p.s.f=1 & UNS & 23.9 \\
AltLoop & REI.id549.ep93366 & SAT & 0.3 & SAT & 1.1 & 0.9 & (slv. during collect.) & SAT & 0.4 \\
AltLoop & REI.id399.ep85537 & UNS & 46.2 & UNS & 68.4 & 18.2 &  & UNS & 26.0 \\
AltLoop & REI.id379.ep96909 & SAT & 101.5 & SAT & 63.1 & 21.4 & p.s.f=1 & SAT & 25.2 \\
AltLoop & REI.id352.ep49376 & UNS & 18.3 & UNS & 29.5 & 29.3 & (slv. during collect.) & UNS & 6.5 \\
AltLoop & REI.id343.ep96195 & UNS & 62.5 & UNS & 63.2 & 18.9 & p.s.f=5 & UNS & 31.9 \\
AltLoop & REI.id318.ep89536 & SAT & 67.2 & SAT & 2.2 & 2.1 & (slv. during collect.) & SAT & 0.5 \\
AltLoop & REI.id47.ep92041 & SAT & 0.6 & SAT & 1.0 & 0.8 & (slv. during collect.) & SAT & 0.3 \\
AltLoop & REI.id68.ep98804 & SAT & 0.3 & SAT & 0.6 & 0.5 & (slv. during collect.) & SAT & 0.3 \\
AltLoop & REI.id73.ep82370 & UNS & 225.2 & UNS & 79.8 & 15.6 & p.s.f=1 & UNS & 589.2 \\
AltLoop & REI.id321.ep86834 & UNS & 29.9 & UNS & 36.0 & 19.8 &  & UNS & 9.3 \\
AltLoop & REI.id318.ep96308 & SAT & 65.9 & SAT & 1.9 & 1.8 & (slv. during collect.) & SAT & 0.5 \\
AltLoop & REI.id239.ep71386 & UNS & 272.1 & UNS & 105.6 & 21.2 & p.s.f=1 & UNS & 611.2 \\
AltLoop & REI.id299.ep69103 & UNS & 54.2 & UNS & 67.2 & 13.5 &  & UNS & 28.6 \\
AltLoop & REI.id267.ep96333 & UNS & 16.2 & UNS & 20.4 & 20.3 & (slv. during collect.) & UNS & 0.3 \\
AltLoop & REI.id91.ep58172 & UNS & 32.4 & UNS & 46.1 & 16.4 &  & UNS & 16.2 \\
AltLoop & REI.id303.ep68520 & UNS & 128.7 & UNS & 146.7 & 13.7 &  & UNS & 118.0 \\
AltLoop & REI.id165.ep82822 & SAT & 0.4 & SAT & 0.8 & 0.7 & (slv. during collect.) & SAT & 0.4 \\
AltLoop & REI.id296.ep76414 & UNS & 53.5 & UNS & 72.4 & 14.8 &  & UNS & 17.9 \\
AltLoop & REI.id303.ep69708 & SAT & 7.3 & SAT & 27.6 & 23.7 & p.s.f=1 & SAT & 13.8 \\
AltLoop & REI.id47.ep96122 & SAT & 2.4 & SAT & 3.0 & 2.9 & (slv. during collect.) & SAT & 0.4 \\
AltLoop & REI.id67.ep82981 & UNS & 32.6 & UNS & 40.4 & 18.9 & br.=pol. p.s.f=1 & UNS & 14.3 \\
AltLoop & REI.id196.ep99092 & SAT & 0.5 & SAT & 0.9 & 0.8 & (slv. during collect.) & SAT & 0.4 \\
AltLoop & REI.id215.ep98767 & SAT & 0.5 & SAT & 2.7 & 2.5 & (slv. during collect.) & SAT & 0.3 \\
AltLoop & REI.id418.ep99947 & SAT & 52.0 & SAT & 1.8 & 1.7 & (slv. during collect.) & SAT & 0.2 \\
AltLoop & REI.id267.ep91426 & UNS & 16.7 & UNS & 13.1 & 13.0 & (slv. during collect.) & UNS & 0.3 \\
AltLoop & REI.id379.ep97133 & SAT & 21.9 & SAT & 40.7 & 15.0 &  & SAT & 1.6 \\
AltLoop & REI.id444.ep92635 & UNS & 159.3 & UNS & 82.0 & 15.2 & p.s.f=1 & UNS & 95.9 \\
AltLoop & REI.id321.ep92994 & UNS & 24.1 & UNS & 37.7 & 25.4 &  & UNS & 9.6 \\
AltLoop & REI.id491.ep67504 & SAT & 2.9 & SAT & 0.6 & 0.5 & (slv. during collect.) & SAT & 0.3 \\
AltLoop & REI.id528.ep95627 & SAT & 0.6 & SAT & 0.9 & 0.8 & (slv. during collect.) & SAT & 0.4 \\
AltLoop & REI.id77.ep99119 & UNS & 51.3 & UNS & 55.6 & 18.4 & p.s.f=1 & UNS & 23.6 \\
AltLoop & REI.id68.ep95942 & SAT & 0.6 & SAT & 0.6 & 0.5 & (slv. during collect.) & SAT & 0.3 \\
AltLoop & REI.id67.ep95115 & SAT & 4.5 & SAT & 1.5 & 1.3 & (slv. during collect.) & SAT & 0.3 \\
AltLoop & REI.id90.ep94747 & UNS & 239.9 & UNS & 265.8 & 16.3 &  & UNS & 244.2 \\
AltLoop & REI.id158.ep59378 & UNS & 33.9 & UNS & 40.9 & 17.5 & br.=pol. p.s.f=1 & UNS & 14.8 \\
AltLoop & REI.id90.ep97680 & UNS & 133.9 & UNS & 76.9 & 18.6 & p.s.f=1 & UNS & 82.5 \\
AltLoop & REI.id76.ep99198 & UNS & 90.1 & UNS & 84.8 & 13.2 &  & UNS & 39.5 \\
AltLoop & REI.id201.ep56910 & SAT & 6.7 & SAT & 0.6 & 0.5 & (slv. during collect.) & SAT & 0.2 \\
AltLoop & REI.id352.ep68171 & UNS & 30.1 & UNS & 43.2 & 17.6 & br.=pol. p.s.f=1 & UNS & 10.9 \\
AltLoop & REI.id240.ep91714 & UNS & 50.0 & UNS & 55.5 & 18.6 & br.=pol. p.s.f=1 & UNS & 28.8 \\
AltLoop & REI.id91.ep75275 & UNS & 44.5 & UNS & 60.6 & 13.8 &  & UNS & 26.3 \\
AltLoop & REI.id180.ep58329 & SAT & 4.1 & SAT & 1.0 & 0.8 & (slv. during collect.) & SAT & 0.3 \\
AltLoop & REI.id158.ep53554 & SAT & 24.4 & SAT & 39.2 & 14.7 &  & SAT & 0.5 \\
AltLoop & REI.id180.ep61293 & SAT & 7.1 & SAT & 2.7 & 2.6 & (slv. during collect.) & SAT & 0.5 \\
AltLoop & REI.id47.ep99688 & SAT & 2.4 & SAT & 1.7 & 1.6 & (slv. during collect.) & SAT & 0.3 \\
AltLoop & REI.id201.ep88825 & SAT & 0.3 & SAT & 0.8 & 0.7 & (slv. during collect.) & SAT & 0.3 \\
AltLoop & REI.id234.ep70556 & SAT & 0.6 & SAT & 1.2 & 1.0 & (slv. during collect.) & SAT & 0.2 \\
AltLoop & REI.id321.ep81083 & UNS & 26.9 & UNS & 52.1 & 29.3 & p.s.f=1 & UNS & 11.9 \\
AltLoop & REI.id457.ep69348 & SAT & 0.5 & SAT & 0.7 & 0.6 & (slv. during collect.) & SAT & 0.5 \\
AltLoop & REI.id286.ep90032 & UNS & 27.9 & UNS & 37.5 & 16.8 &  & UNS & 9.5 \\
AltLoop & REI.id114.ep67447 & SAT & 2.4 & SAT & 0.5 & 0.4 & (slv. during collect.) & SAT & 0.5 \\
AltLoop & REI.id68.ep95821 & SAT & 0.8 & SAT & 0.7 & 0.6 & (slv. during collect.) & SAT & 0.2 \\
AltLoop & REI.id343.ep55239 & UNS & 23.3 & UNS & 26.3 & 26.1 & (slv. during collect.) & UNS & 4.3 \\
AltLoop & REI.id418.ep92561 & SAT & 0.6 & SAT & 1.7 & 1.5 & (slv. during collect.) & SAT & 0.4 \\
AltLoop & REI.id535.ep76894 & UNS & 49.9 & UNS & 54.9 & 16.1 & p.s.f=5 & UNS & 13.2 \\
AltLoop & REI.id81.ep91408 & UNS & 115.6 & UNS & 139.7 & 12.4 &  & UNS & 404.9 \\
AltLoop & REI.id201.ep39282 & SAT & 6.2 & SAT & 0.6 & 0.5 & (slv. during collect.) & SAT & 0.3 \\
AltLoop & REI.id352.ep69872 & UNS & 22.4 & UNS & 32.3 & 20.8 &  & UNS & 8.1 \\
AltLoop & REI.id174.ep91107 & SAT & 0.5 & SAT & 0.5 & 0.3 & (slv. during collect.) & SAT & 0.3 \\
AltLoop & REI.id418.ep91524 & SAT & 0.5 & SAT & 1.9 & 1.8 & (slv. during collect.) & SAT & 0.4 \\
AltLoop & REI.id77.ep75183 & UNS & 31.0 & UNS & 46.2 & 16.1 &  & UNS & 14.7 \\
AltLoop & REI.id393.ep88789 & SAT & 7.5 & SAT & 19.6 & 17.6 & p.s.f=2 & SAT & 0.4 \\
AltLoop & REI.id77.ep91446 & UNS & 46.0 & UNS & 55.2 & 11.6 &  & UNS & 23.5 \\
AltLoop & REI.id174.ep66258 & SAT & 4.1 & SAT & 0.4 & 0.3 & (slv. during collect.) & SAT & 0.3 \\
AltLoop & REI.id425.ep98177 & SAT & 0.5 & SAT & 0.9 & 0.8 & (slv. during collect.) & SAT & 0.3 \\
AltLoop & REI.id76.ep92143 & UNS & 47.4 & UNS & 50.1 & 14.9 &  & UNS & 16.7 \\
AltLoop & REI.id512.ep75804 & SAT & 32.3 & SAT & 1.5 & 1.3 & (slv. during collect.) & SAT & 0.6 \\
AltLoop & REI.id240.ep85050 & UNS & 45.1 & UNS & 61.1 & 18.1 & br.=pol. p.s.f=1 & UNS & 23.2 \\
AltLoop & REI.id540.ep85904 & SAT & 0.6 & SAT & 0.7 & 0.6 & (slv. during collect.) & SAT & 0.5 \\
AltLoop & REI.id176.ep95516 & UNS & 75.6 & UNS & 134.8 & 17.2 & br.=pol. p.s.f=5 & UNS & 75.8 \\
AltLoop & REI.id67.ep99214 & SAT & 6.7 & SAT & 1.6 & 1.4 & (slv. during collect.) & SAT & 0.4 \\
AltLoop & REI.id114.ep84297 & UNS & 66.2 & UNS & 124.2 & 17.4 & br.=pol. p.s.f=1 & UNS & 43.4 \\
AltLoop & REI.id128.ep87384 & SAT & 111.7 & SAT & 3.1 & 2.9 & (slv. during collect.) & SAT & 0.5 \\
AltLoop & REI.id128.ep85878 & UNS & 346.1 & UNS & 92.6 & 15.8 & p.s.f=1 & UNS & 549.8 \\
AltLoop & REI.id35.ep81067 & SAT & 0.5 & SAT & 0.7 & 0.5 & (slv. during collect.) & SAT & 0.3 \\
AltLoop & REI.id381.ep92994 & SAT & 122.7 & SAT & 1.0 & 0.9 & (slv. during collect.) & SAT & 0.4 \\
AltLoop & REI.id180.ep58461 & SAT & 3.6 & SAT & 1.1 & 1.0 & (slv. during collect.) & SAT & 0.3 \\
AltLoop & REI.id418.ep99126 & SAT & 139.2 & SAT & 2.0 & 1.9 & (slv. during collect.) & SAT & 0.3 \\
AltLoop & REI.id399.ep62611 & UNS & 31.5 & UNS & 35.8 & 15.2 &  & UNS & 12.0 \\
AltLoop & REI.id296.ep77867 & UNS & 23.8 & UNS & 43.3 & 13.8 &  & UNS & 9.6 \\
AltLoop & REI.id296.ep74969 & UNS & 45.2 & UNS & 54.4 & 13.2 &  & UNS & 17.0 \\
AltLoop & REI.id343.ep73171 & UNS & 149.7 & UNS & 167.5 & 18.2 &  & UNS & 85.7 \\
AltLoop & REI.id47.ep79665 & SAT & 2.5 & SAT & 2.8 & 2.7 & (slv. during collect.) & SAT & 0.3 \\
AltLoop & REI.id512.ep97330 & SAT & 30.9 & SAT & 1.8 & 1.6 & (slv. during collect.) & SAT & 0.5 \\
AltLoop & REI.id234.ep83780 & SAT & 0.4 & SAT & 1.2 & 1.0 & (slv. during collect.) & SAT & 0.2 \\
AltLoop & REI.id286.ep89999 & UNS & 30.8 & UNS & 35.0 & 16.2 &  & UNS & 13.1 \\
AltLoop & REI.id528.ep55711 & SAT & 13.2 & SAT & 3.2 & 3.1 & (slv. during collect.) & SAT & 0.7 \\
AltLoop & REI.id128.ep97983 & UNS & 243.9 & UNS & 75.8 & 16.5 & p.s.f=1 & UNS & 203.0 \\
AltLoop & REI.id381.ep92541 & UNS & 129.8 & UNS & 149.1 & 15.1 &  & UNS & 152.6 \\
AltLoop & REI.id174.ep81057 & SAT & 0.4 & SAT & 0.6 & 0.4 & (slv. during collect.) & SAT & 0.3 \\
AltLoop & REI.id319.ep89649 & SAT & 0.4 & SAT & 0.6 & 0.5 & (slv. during collect.) & SAT & 0.3 \\
AltLoop & REI.id528.ep90341 & SAT & 21.1 & SAT & 3.4 & 3.3 & (slv. during collect.) & SAT & 3.1 \\
AltLoop & REI.id299.ep69400 & UNS & 50.5 & UNS & 79.9 & 17.2 &  & UNS & 27.8 \\
AltLoop & REI.id81.ep97775 & UNS & 154.7 & UNS & 242.3 & 18.0 &  & UNS & 632.9 \\
AltLoop & REI.id196.ep73462 & SAT & 0.5 & SAT & 1.1 & 0.9 & (slv. during collect.) & SAT & 0.3 \\
AltLoop & REI.id379.ep87708 & UNS & 297.9 & UNS & 92.3 & 18.8 & p.s.f=1 & UNS & 624.2 \\
AltLoop & REI.id174.ep98455 & SAT & 0.3 & SAT & 0.6 & 0.5 & (slv. during collect.) & SAT & 0.2 \\
AltLoop & REI.id81.ep97833 & UNS & 185.1 & UNS & 102.8 & 18.0 & p.s.f=1 & TO & 3601.4 \\
AltLoop & REI.id308.ep98830 & UNS & 48.5 & UNS & 70.7 & 12.2 &  & UNS & 28.0 \\
AltLoop & REI.id67.ep84585 & UNS & 39.3 & UNS & 42.0 & 16.7 & p.s.f=1 & UNS & 26.4 \\
AltLoop & REI.id234.ep71035 & SAT & 0.4 & SAT & 0.9 & 0.8 & (slv. during collect.) & SAT & 0.3 \\
AltLoop & REI.id91.ep58150 & UNS & 36.1 & UNS & 46.5 & 16.1 &  & UNS & 16.6 \\
AltLoop & REI.id196.ep43251 & SAT & 0.5 & SAT & 0.8 & 0.7 & (slv. during collect.) & SAT & 0.3 \\
AltLoop & REI.id381.ep85558 & SAT & 66.9 & SAT & 3.4 & 3.2 & (slv. during collect.) & SAT & 77.2 \\
AltLoop & REI.id418.ep92671 & SAT & 0.5 & SAT & 1.9 & 1.8 & (slv. during collect.) & SAT & 0.4 \\
AltLoop & REI.id528.ep89177 & SAT & 0.5 & SAT & 3.0 & 2.9 & (slv. during collect.) & SAT & 2.7 \\
AltLoop & REI.id393.ep97221 & SAT & 0.6 & SAT & 20.2 & 19.4 & p.s.f=1 & SAT & 0.5 \\
AltLoop & REI.id128.ep91042 & UNS & 53.1 & UNS & 62.7 & 19.7 & p.s.f=1 & UNS & 25.9 \\
AltLoop & REI.id308.ep99303 & SAT & 22.8 & SAT & 5.0 & 4.9 & (slv. during collect.) & SAT & 0.4 \\
AltLoop & REI.id425.ep69300 & SAT & 24.2 & SAT & 1.6 & 1.5 & (slv. during collect.) & SAT & 0.4 \\
AltLoop & REI.id234.ep83703 & SAT & 0.3 & SAT & 1.0 & 0.8 & (slv. during collect.) & SAT & 0.3 \\
NAP & cls4.id202 & SAT & 1526.8 & SAT & 3.0 & 2.9 & (slv. during collect.) & SAT & 157.8 \\
NAP & cls3.id114 & SAT & 1455.0 & SAT & 0.8 & 0.7 & (slv. during collect.) & SAT & 6.0 \\
NAP & cls4.id219 & SAT & 2202.5 & SAT & 2390.6 & 165.1 &  & SAT & 818.4 \\
NAP & cls4.id228 & SAT & 55.9 & SAT & 132.6 & 132.4 & (slv. during collect.) & TO & 3601.4 \\
NAP & cls4.id217 & TO & 3601.5 & SAT & 4.6 & 4.5 & (slv. during collect.) & SAT & 41.4 \\
NAP & cls0.id22 & TO & 3601.6 & SAT & 1.7 & 1.6 & (slv. during collect.) & SAT & 0.9 \\
NAP & cls0.id18 & SAT & 3265.2 & SAT & 1.8 & 1.6 & (slv. during collect.) & SAT & 0.9 \\
NAP & cls3.id151 & SAT & 2069.6 & SAT & 80.7 & 80.6 & (slv. during collect.) & SAT & 102.8 \\
NAP & cls3.id155 & SAT & 1514.2 & SAT & 2.0 & 1.9 & (slv. during collect.) & SAT & 0.9 \\
NAP & cls3.id143 & SAT & 2734.5 & SAT & 106.0 & 105.9 & (slv. during collect.) & SAT & 164.2 \\
NAP & cls0.id57 & SAT & 389.6 & SAT & 9.8 & 9.7 & (slv. during collect.) & SAT & 0.8 \\
NAP & cls3.id169 & SAT & 660.4 & SAT & 818.3 & 156.3 &  & TO & 3601.4 \\
NAP & cls0.id25 & SAT & 134.9 & SAT & 3.7 & 3.5 & (slv. during collect.) & SAT & 0.6 \\
NAP & cls2.id93 & SAT & 2198.7 & SAT & 2354.7 & 177.8 &  & SAT & 32.0 \\
NAP & cls4.id194 & SAT & 2393.8 & SAT & 2.0 & 1.9 & (slv. during collect.) & SAT & 178.8 \\
NAP & cls0.id36 & SAT & 384.1 & SAT & 18.0 & 17.9 & (slv. during collect.) & SAT & 1.0 \\
NAP & cls4.id214 & TO & 3601.5 & SAT & 0.9 & 0.9 & (slv. during collect.) & SAT & 43.5 \\
NAP & cls1.id82 & SAT & 345.0 & SAT & 29.4 & 29.4 & (slv. during collect.) & SAT & 132.4 \\
NAP & cls0.id16 & SAT & 20.0 & SAT & 5.7 & 5.6 & (slv. during collect.) & SAT & 0.9 \\
NAP & cls1.id87 & SAT & 2627.3 & SAT & 102.1 & 102.0 & (slv. during collect.) & SAT & 319.7 \\
NAP & cls2.id91 & SAT & 175.4 & SAT & 51.9 & 51.7 & (slv. during collect.) & SAT & 0.8 \\
NAP & cls0.id58 & SAT & 130.3 & SAT & 4.5 & 4.4 & (slv. during collect.) & SAT & 5.7 \\
NAP & cls0.id8 & TO & 3601.5 & SAT & 1.6 & 1.5 & (slv. during collect.) & SAT & 0.7 \\
NAP & cls0.id9 & SAT & 5.2 & SAT & 0.8 & 0.7 & (slv. during collect.) & SAT & 0.7 \\
NAP & cls1.id80 & SAT & 1.1 & SAT & 65.0 & 64.9 & (slv. during collect.) & SAT & 147.8 \\
NAP & cls4.id230 & TO & 3601.4 & SAT & 47.5 & 47.4 & (slv. during collect.) & SAT & 1.5 \\
NAP & cls3.id109 & SAT & 1032.6 & SAT & 6.3 & 6.2 & (slv. during collect.) & SAT & 1.0 \\
NAP & cls4.id205 & SAT & 255.3 & SAT & 49.1 & 49.1 & (slv. during collect.) & SAT & 278.8 \\
NAP & cls0.id53 & SAT & 1469.6 & SAT & 14.1 & 13.9 & (slv. during collect.) & SAT & 5.7 \\
NAP & cls0.id7 & SAT & 0.9 & SAT & 8.4 & 8.2 & (slv. during collect.) & SAT & 0.8 \\
NAP & cls4.id173 & SAT & 3197.4 & SAT & 1.2 & 1.2 & (slv. during collect.) & SAT & 0.8 \\
NAP & cls3.id156 & SAT & 997.4 & SAT & 1.0 & 0.8 & (slv. during collect.) & SAT & 1.5 \\
NAP & cls4.id197 & TO & 3601.5 & SAT & 47.0 & 46.9 & (slv. during collect.) & SAT & 1091.4 \\
NAP & cls4.id176 & SAT & 3034.6 & SAT & 1.6 & 1.5 & (slv. during collect.) & SAT & 49.0 \\
NAP & cls4.id186 & SAT & 724.1 & SAT & 901.6 & 179.2 &  & SAT & 259.0 \\
NAP & cls0.id67 & SAT & 848.6 & SAT & 24.9 & 24.7 & (slv. during collect.) & SAT & 0.8 \\
NAP & cls4.id199 & TO & 3601.5 & SAT & 1.4 & 1.2 & (slv. during collect.) & SAT & 51.9 \\
NAP & cls0.id1 & SAT & 20.6 & SAT & 0.7 & 0.6 & (slv. during collect.) & SAT & 0.6 \\
NAP & cls3.id166 & SAT & 721.4 & SAT & 76.5 & 76.4 & (slv. during collect.) & SAT & 1.1 \\
NAP & cls0.id51 & SAT & 108.1 & SAT & 17.8 & 17.6 & (slv. during collect.) & SAT & 0.8 \\
NAP & cls3.id134 & SAT & 324.3 & SAT & 30.4 & 30.2 & (slv. during collect.) & SAT & 108.3 \\
NAP & cls4.id182 & SAT & 854.4 & SAT & 2.3 & 2.1 & (slv. during collect.) & SAT & 2.4 \\
NAP & cls0.id43 & SAT & 50.6 & SAT & 13.8 & 13.7 & (slv. during collect.) & SAT & 598.3 \\
NAP & cls1.id86 & SAT & 1.2 & SAT & 92.5 & 92.4 & (slv. during collect.) & SAT & 7.0 \\
NAP & cls0.id52 & SAT & 888.0 & SAT & 3.7 & 3.6 & (slv. during collect.) & SAT & 0.6 \\
NAP & cls0.id66 & SAT & 72.8 & SAT & 243.7 & 171.8 &  & SAT & 86.4 \\
NAP & cls3.id120 & TO & 3601.5 & SAT & 104.2 & 104.1 & (slv. during collect.) & TO & 3601.6 \\
NAP & cls4.id198 & SAT & 193.8 & SAT & 389.2 & 195.8 &  & SAT & 2.7 \\
NAP & cls4.id189 & SAT & 1215.3 & SAT & 10.7 & 10.6 & (slv. during collect.) & SAT & 2826.7 \\
NAP & cls3.id145 & SAT & 2198.3 & SAT & 30.8 & 30.7 & (slv. during collect.) & SAT & 195.9 \\
NAP & cls3.id139 & SAT & 422.1 & SAT & 656.5 & 234.7 &  & TO & 3601.6 \\
NAP & cls0.id12 & SAT & 198.1 & SAT & 3.5 & 3.4 & (slv. during collect.) & SAT & 1.6 \\
NAP & cls0.id70 & SAT & 351.8 & SAT & 9.1 & 9.0 & (slv. during collect.) & SAT & 51.4 \\
NAP & cls3.id112 & SAT & 99.7 & SAT & 71.0 & 71.0 & (slv. during collect.) & SAT & 73.5 \\
NAP & cls4.id195 & SAT & 71.0 & SAT & 5.5 & 5.4 & (slv. during collect.) & SAT & 1.2 \\
NAP & cls2.id94 & TO & 3601.5 & TO & 3601.5 & 187.0 &  & SAT & 31.2 \\
NAP & cls4.id221 & TO & 3601.5 & TO & 3601.6 & 192.4 &  & SAT & 1.5 \\
NAP & cls0.id24 & TO & 3601.5 & SAT & 1.4 & 1.3 & (slv. during collect.) & SAT & 0.8 \\
NAP & cls4.id215 & TO & 3601.5 & TO & 3601.5 & 182.4 &  & SAT & 1033.7 \\
NAP & cls4.id227 & TO & 3601.4 & TO & 3601.6 & 168.5 &  & SAT & 2673.2 \\
NAP & cls0.id63 & SAT & 2685.1 & SAT & 18.0 & 17.9 & (slv. during collect.) & SAT & 41.5 \\
NAP & cls1.id88 & SAT & 1784.5 & SAT & 99.7 & 99.6 & (slv. during collect.) & SAT & 102.8 \\
NAP & cls3.id153 & SAT & 690.1 & SAT & 30.7 & 30.6 & (slv. during collect.) & SAT & 2731.6 \\
NAP & cls4.id235 & TO & 3601.5 & SAT & 6.5 & 6.4 & (slv. during collect.) & SAT & 6.1 \\
NAP & cls1.id79 & SAT & 1.1 & SAT & 176.5 & 175.3 &  & SAT & 105.8 \\
NAP & cls3.id132 & SAT & 483.2 & SAT & 11.0 & 10.8 & (slv. during collect.) & SAT & 1.0 \\
NAP & cls3.id168 & SAT & 683.3 & SAT & 29.4 & 29.2 & (slv. during collect.) & TO & 3601.4 \\
NAP & cls4.id192 & TO & 3601.5 & SAT & 1.0 & 0.8 & (slv. during collect.) & SAT & 1.1 \\
NAP & cls3.id128 & TO & 3601.6 & SAT & 30.5 & 30.3 & (slv. during collect.) & SAT & 244.8 \\
NAP & cls3.id137 & SAT & 995.0 & SAT & 21.9 & 21.8 & (slv. during collect.) & SAT & 77.0 \\
NAP & cls2.id98 & SAT & 382.2 & SAT & 109.3 & 109.3 & (slv. during collect.) & SAT & 11.6 \\
NAP & cls3.id106 & SAT & 258.5 & SAT & 2.4 & 2.3 & (slv. during collect.) & SAT & 1.0 \\
NAP & cls3.id167 & SAT & 1843.7 & SAT & 2030.8 & 176.1 &  & TO & 3601.4 \\
NAP & cls2.id95 & SAT & 2372.3 & SAT & 31.9 & 31.8 & (slv. during collect.) & SAT & 82.3 \\
NAP & cls4.id200 & SAT & 2467.7 & SAT & 32.1 & 32.0 & (slv. during collect.) & SAT & 1.0 \\
NAP & cls4.id213 & TO & 3601.5 & SAT & 2.6 & 2.5 & (slv. during collect.) & SAT & 367.1 \\
NAP & cls0.id68 & SAT & 2107.2 & SAT & 1.3 & 1.1 & (slv. during collect.) & SAT & 36.1 \\
NAP & cls0.id55 & SAT & 705.3 & SAT & 1.3 & 1.2 & (slv. during collect.) & SAT & 0.7 \\
NAP & cls4.id206 & TO & 3601.5 & TO & 3601.5 & 149.9 &  & TO & 3601.4 \\
NAP & cls0.id20 & SAT & 472.3 & SAT & 14.3 & 14.2 & (slv. during collect.) & SAT & 6.1 \\
NAP & cls4.id211 & TO & 3601.5 & SAT & 1.1 & 1.0 & (slv. during collect.) & SAT & 0.8 \\
NAP & cls3.id138 & SAT & 1704.0 & SAT & 1867.8 & 174.4 &  & TO & 3601.5 \\
NAP & cls3.id160 & SAT & 1678.0 & SAT & 1842.8 & 172.9 &  & SAT & 3.5 \\
NAP & cls4.id226 & TO & 3601.5 & SAT & 32.0 & 31.9 & (slv. during collect.) & SAT & 1.8 \\
NAP & cls3.id115 & SAT & 149.1 & SAT & 10.8 & 10.8 & (slv. during collect.) & SAT & 0.9 \\
NAP & cls4.id225 & TO & 3601.5 & SAT & 14.8 & 14.6 & (slv. during collect.) & SAT & 398.8 \\
NAP & cls3.id136 & SAT & 938.0 & SAT & 27.6 & 27.5 & (slv. during collect.) & SAT & 105.3 \\
NAP & cls3.id146 & SAT & 159.4 & SAT & 330.1 & 176.7 &  & SAT & 314.4 \\
NAP & cls3.id113 & SAT & 654.1 & SAT & 47.5 & 47.4 & (slv. during collect.) & SAT & 33.8 \\
NAP & cls0.id76 & TO & 3601.5 & SAT & 21.6 & 21.4 & (slv. during collect.) & SAT & 1.3 \\
NAP & cls3.id123 & SAT & 2.2 & SAT & 11.5 & 11.4 & (slv. during collect.) & SAT & 931.7 \\
NAP & cls3.id122 & TO & 3601.5 & SAT & 101.5 & 101.4 & (slv. during collect.) & SAT & 49.4 \\
NAP & cls0.id74 & TO & 3601.5 & SAT & 2.4 & 2.2 & (slv. during collect.) & SAT & 46.3 \\
NAP & cls0.id15 & SAT & 15.4 & SAT & 12.3 & 12.2 & (slv. during collect.) & SAT & 0.8 \\
NAP & cls0.id23 & SAT & 839.2 & SAT & 1.1 & 1.0 & (slv. during collect.) & SAT & 0.8 \\
NAP & cls0.id31 & SAT & 30.7 & SAT & 0.8 & 0.7 & (slv. during collect.) & SAT & 35.6 \\
NAP & cls0.id50 & SAT & 496.9 & SAT & 57.5 & 57.4 & (slv. during collect.) & SAT & 85.9 \\
NAP & cls3.id150 & TO & 3601.5 & TO & 3601.5 & 180.5 &  & TO & 3601.5 \\
NAP & cls4.id231 & TO & 3601.5 & SAT & 36.9 & 36.8 & (slv. during collect.) & SAT & 1086.4 \\
NAP & cls2.id99 & TO & 3601.5 & SAT & 41.0 & 40.9 & (slv. during collect.) & SAT & 1954.0 \\
NAP & cls4.id190 & SAT & 1197.7 & SAT & 35.3 & 35.1 & (slv. during collect.) & SAT & 1.1 \\
NAP & cls0.id60 & TO & 3601.5 & SAT & 21.5 & 21.3 & (slv. during collect.) & SAT & 0.8 \\
NAP & cls2.id101 & SAT & 2518.8 & TO & 3601.5 & 152.2 &  & SAT & 42.1 \\
NAP & cls3.id107 & SAT & 68.0 & SAT & 27.4 & 27.3 & (slv. during collect.) & SAT & 62.5 \\
NAP & cls0.id47 & SAT & 2671.9 & SAT & 21.4 & 21.4 & (slv. during collect.) & SAT & 10.9 \\
NAP & cls0.id62 & SAT & 22.9 & SAT & 17.2 & 17.1 & (slv. during collect.) & SAT & 51.5 \\
NAP & cls0.id38 & SAT & 1009.0 & SAT & 1158.5 & 148.0 &  & SAT & 57.6 \\
NAP & cls4.id181 & SAT & 1906.7 & SAT & 18.9 & 18.8 & (slv. during collect.) & SAT & 0.8 \\
NAP & cls0.id49 & TO & 3601.5 & SAT & 20.6 & 20.4 & (slv. during collect.) & SAT & 41.2 \\
NAP & cls0.id65 & TO & 3601.5 & SAT & 60.0 & 59.9 & (slv. during collect.) & SAT & 1.5 \\
NAP & cls4.id185 & SAT & 35.6 & SAT & 38.7 & 38.6 & (slv. during collect.) & SAT & 170.1 \\
NAP & cls0.id71 & SAT & 134.5 & SAT & 4.8 & 4.7 & (slv. during collect.) & SAT & 5.7 \\
NAP & cls4.id180 & SAT & 341.0 & SAT & 26.4 & 26.2 & (slv. during collect.) & SAT & 3.2 \\
NAP & cls4.id191 & TO & 3601.5 & SAT & 14.1 & 13.9 & (slv. during collect.) & SAT & 1.9 \\
NAP & cls0.id45 & SAT & 2255.4 & SAT & 18.2 & 18.1 & (slv. during collect.) & SAT & 26.6 \\
NAP & cls3.id148 & SAT & 2132.7 & SAT & 2292.6 & 160.2 &  & TO & 3601.5 \\
NAP & cls4.id183 & TO & 3601.5 & SAT & 2.5 & 2.3 & (slv. during collect.) & SAT & 1.3 \\
NAP & cls3.id133 & SAT & 336.2 & SAT & 534.8 & 203.2 &  & SAT & 913.4 \\
NAP & cls4.id208 & TO & 3601.5 & SAT & 2.0 & 1.9 & (slv. during collect.) & SAT & 360.3 \\
NAP & cls0.id28 & SAT & 52.1 & SAT & 56.3 & 56.1 & (slv. during collect.) & SAT & 31.5 \\
NAP & cls4.id220 & TO & 3601.5 & SAT & 35.5 & 35.4 & (slv. during collect.) & SAT & 12.2 \\
NAP & cls4.id222 & TO & 3601.5 & TO & 3601.5 & 154.1 &  & SAT & 1538.0 \\
NAP & cls4.id234 & TO & 3601.5 & TO & 3602.1 & 242.6 &  & SAT & 1.5 \\
NAP & cls3.id164 & SAT & 117.0 & SAT & 37.0 & 36.9 & (slv. during collect.) & SAT & 1.4 \\
NAP & cls2.id103 & TO & 3601.5 & TO & 3601.5 & 170.4 &  & SAT & 52.5 \\
NAP & cls4.id178 & TO & 3601.5 & SAT & 38.5 & 38.4 & (slv. during collect.) & SAT & 1144.6 \\
NAP & cls3.id130 & SAT & 1070.8 & SAT & 36.4 & 36.2 & (slv. during collect.) & SAT & 318.9 \\
NAP & cls3.id108 & SAT & 922.8 & SAT & 1.4 & 1.3 & (slv. during collect.) & SAT & 22.8 \\
NAP & cls0.id30 & SAT & 9.4 & SAT & 7.3 & 7.2 & (slv. during collect.) & SAT & 1.1 \\
NAP & cls0.id69 & SAT & 912.4 & SAT & 1096.3 & 172.9 &  & SAT & 1448.7 \\
NAP & cls4.id196 & SAT & 1233.7 & SAT & 9.9 & 9.9 & (slv. during collect.) & TO & 3601.4 \\
NAP & cls0.id39 & SAT & 0.8 & SAT & 14.6 & 14.4 & (slv. during collect.) & SAT & 5.9 \\
NAP & cls3.id152 & TO & 3601.5 & TO & 3601.5 & 245.8 &  & TO & 3601.5 \\
NAP & cls2.id100 & TO & 3601.5 & SAT & 1.7 & 1.6 & (slv. during collect.) & SAT & 0.6 \\
NAP & cls0.id41 & SAT & 360.0 & SAT & 16.3 & 16.2 & (slv. during collect.) & SAT & 41.6 \\
NAP & cls4.id209 & SAT & 2123.7 & SAT & 2.0 & 1.9 & (slv. during collect.) & SAT & 1097.8 \\
NAP & cls0.id26 & SAT & 58.4 & SAT & 4.2 & 4.1 & (slv. during collect.) & SAT & 1.1 \\
NAP & cls2.id105 & TO & 3601.5 & SAT & 6.0 & 5.9 & (slv. during collect.) & SAT & 71.7 \\
NAP & cls3.id162 & SAT & 1482.7 & SAT & 1696.1 & 228.4 &  & SAT & 295.0 \\
NAP & cls0.id48 & SAT & 2844.3 & SAT & 52.7 & 52.5 & (slv. during collect.) & SAT & 254.9 \\
NAP & cls3.id126 & SAT & 3.0 & SAT & 190.6 & 187.5 &  & SAT & 149.7 \\
NAP & cls3.id141 & TO & 3601.5 & TO & 3601.5 & 222.3 &  & TO & 3601.5 \\
NAP & cls0.id14 & SAT & 4.5 & SAT & 1.9 & 1.8 & (slv. during collect.) & SAT & 0.9 \\
NAP & cls0.id32 & SAT & 3231.1 & SAT & 10.9 & 10.8 & (slv. during collect.) & SAT & 129.3 \\
NAP & cls3.id142 & SAT & 33.6 & SAT & 25.2 & 25.1 & (slv. during collect.) & SAT & 1508.0 \\
NAP & cls3.id121 & SAT & 1583.8 & SAT & 10.2 & 10.1 & (slv. during collect.) & SAT & 0.6 \\
NAP & cls4.id216 & TO & 3601.5 & SAT & 6.9 & 6.7 & (slv. during collect.) & SAT & 7.0 \\
NAP & cls4.id201 & TO & 3601.4 & SAT & 77.8 & 77.7 & (slv. during collect.) & SAT & 9.3 \\
NAP & cls1.id81 & SAT & 1623.4 & SAT & 57.3 & 57.2 & (slv. during collect.) & SAT & 0.8 \\
NAP & cls4.id188 & TO & 3601.5 & TO & 3601.5 & 167.1 &  & SAT & 3338.7 \\
NAP & cls4.id232 & TO & 3601.5 & TO & 3602.1 & 173.8 &  & SAT & 3449.0 \\
NAP & cls2.id102 & TO & 3601.4 & SAT & 1.3 & 1.2 & (slv. during collect.) & SAT & 1.3 \\
NAP & cls3.id147 & SAT & 675.5 & SAT & 1.0 & 0.9 & (slv. during collect.) & SAT & 1.1 \\
NAP & cls0.id61 & SAT & 998.1 & SAT & 25.6 & 25.5 & (slv. during collect.) & SAT & 0.7 \\
NAP & cls0.id44 & SAT & 493.7 & SAT & 0.8 & 0.7 & (slv. during collect.) & SAT & 7.0 \\
NAP & cls2.id97 & SAT & 601.8 & SAT & 17.8 & 17.7 & (slv. during collect.) & SAT & 170.5 \\
NAP & cls2.id89 & SAT & 2897.6 & SAT & 18.2 & 18.1 & (slv. during collect.) & SAT & 31.5 \\
NAP & cls0.id42 & SAT & 24.4 & SAT & 16.4 & 16.2 & (slv. during collect.) & SAT & 1.4 \\
NAP & cls0.id4 & SAT & 255.9 & SAT & 1.7 & 1.6 & (slv. during collect.) & SAT & 0.5 \\
NAP & cls0.id40 & SAT & 166.8 & SAT & 0.9 & 0.8 & (slv. during collect.) & SAT & 0.9 \\
NAP & cls2.id96 & SAT & 863.1 & SAT & 1.1 & 0.9 & (slv. during collect.) & SAT & 193.5 \\
NAP & cls3.id149 & SAT & 962.8 & SAT & 9.8 & 9.6 & (slv. during collect.) & SAT & 0.7 \\
NAP & cls3.id159 & SAT & 723.6 & SAT & 41.6 & 41.5 & (slv. during collect.) & SAT & 140.1 \\
NAP & cls4.id174 & SAT & 99.4 & SAT & 1.6 & 1.5 & (slv. during collect.) & SAT & 9.2 \\
NAP & cls3.id116 & SAT & 561.6 & SAT & 2.3 & 2.2 & (slv. during collect.) & SAT & 0.9 \\
NAP & cls3.id118 & SAT & 2.3 & SAT & 66.4 & 66.3 & (slv. during collect.) & SAT & 8.7 \\
NAP & cls3.id161 & TO & 3601.5 & SAT & 44.4 & 44.2 & (slv. during collect.) & SAT & 201.8 \\
NAP & cls1.id84 & SAT & 1.5 & SAT & 12.1 & 12.0 & (slv. during collect.) & SAT & 1.1 \\
NAP & cls4.id177 & TO & 3601.5 & SAT & 1.5 & 1.4 & (slv. during collect.) & SAT & 2.2 \\
NAP & cls4.id193 & SAT & 2419.1 & SAT & 33.8 & 33.7 & (slv. during collect.) & SAT & 61.8 \\
NAP & cls3.id157 & TO & 3601.5 & TO & 3601.5 & 223.6 &  & SAT & 3051.9 \\
NAP & cls3.id163 & SAT & 3265.2 & SAT & 104.5 & 104.3 & (slv. during collect.) & TO & 3601.5 \\
NAP & cls4.id212 & SAT & 1524.5 & SAT & 2.5 & 2.3 & (slv. during collect.) & SAT & 158.9 \\
NAP & cls0.id5 & SAT & 0.8 & SAT & 0.8 & 0.6 & (slv. during collect.) & SAT & 0.6 \\
NAP & cls0.id17 & SAT & 266.5 & SAT & 4.8 & 4.6 & (slv. during collect.) & SAT & 0.7 \\
NAP & cls4.id223 & SAT & 105.2 & SAT & 263.4 & 156.9 &  & SAT & 4.9 \\
NAP & cls4.id203 & TO & 3601.5 & TO & 3601.5 & 175.5 &  & SAT & 687.9 \\
NAP & cls0.id11 & SAT & 41.3 & SAT & 11.5 & 11.4 & (slv. during collect.) & SAT & 0.9 \\
NAP & cls1.id77 & TO & 3601.5 & TO & 3601.6 & 162.3 &  & TO & 3601.5 \\
NAP & cls3.id111 & SAT & 280.2 & SAT & 2.4 & 2.2 & (slv. during collect.) & SAT & 0.8 \\
NAP & cls2.id104 & TO & 3601.4 & SAT & 1.5 & 1.4 & (slv. during collect.) & SAT & 1.0 \\
NAP & cls0.id21 & SAT & 0.9 & SAT & 1.9 & 1.7 & (slv. during collect.) & SAT & 0.7 \\
NAP & cls3.id172 & SAT & 1294.5 & SAT & 10.4 & 10.3 & (slv. during collect.) & SAT & 11.2 \\
NAP & cls2.id92 & SAT & 19.4 & SAT & 18.4 & 18.4 & (slv. during collect.) & TO & 3601.5 \\
NAP & cls4.id204 & TO & 3601.5 & TO & 3601.5 & 182.2 &  & SAT & 200.9 \\
NAP & cls4.id207 & SAT & 2150.0 & SAT & 35.0 & 34.9 & (slv. during collect.) & SAT & 12.8 \\
NAP & cls0.id6 & SAT & 1524.5 & SAT & 0.7 & 0.6 & (slv. during collect.) & SAT & 0.9 \\
NAP & cls3.id144 & SAT & 194.6 & SAT & 2.2 & 2.2 & (slv. during collect.) & SAT & 6.2 \\
NAP & cls4.id184 & SAT & 559.1 & SAT & 2.6 & 2.5 & (slv. during collect.) & SAT & 64.7 \\
NAP & cls4.id233 & SAT & 498.5 & SAT & 1.5 & 1.3 & (slv. during collect.) & SAT & 1.8 \\
NAP & cls0.id29 & SAT & 22.4 & SAT & 13.7 & 13.6 & (slv. during collect.) & SAT & 656.9 \\
NAP & cls0.id64 & TO & 3601.5 & SAT & 16.2 & 16.2 & (slv. during collect.) & SAT & 0.6 \\
NAP & cls0.id19 & SAT & 175.1 & SAT & 15.1 & 15.0 & (slv. during collect.) & SAT & 41.4 \\
NAP & cls3.id165 & SAT & 256.5 & SAT & 10.6 & 10.5 & (slv. during collect.) & SAT & 0.9 \\
NAP & cls0.id10 & SAT & 869.5 & SAT & 10.2 & 10.0 & (slv. during collect.) & SAT & 0.7 \\
NAP & cls0.id33 & SAT & 1475.3 & SAT & 1.3 & 1.1 & (slv. during collect.) & SAT & 70.5 \\
NAP & cls0.id13 & SAT & 35.9 & SAT & 3.5 & 3.3 & (slv. during collect.) & SAT & 16.9 \\
NAP & cls3.id154 & SAT & 1057.9 & SAT & 4.5 & 4.4 & (slv. during collect.) & SAT & 5.9 \\
NAP & cls1.id83 & SAT & 60.7 & SAT & 246.8 & 184.7 &  & SAT & 1405.8 \\
NAP & cls1.id85 & TO & 3601.5 & SAT & 26.3 & 26.2 & (slv. during collect.) & SAT & 372.3 \\
NAP & cls0.id75 & TO & 3601.5 & SAT & 4.0 & 3.9 & (slv. during collect.) & SAT & 114.9 \\
NAP & cls1.id78 & TO & 3601.5 & TO & 3601.5 & 142.5 &  & SAT & 2596.2 \\
NAP & cls3.id124 & SAT & 358.2 & SAT & 49.1 & 49.0 & (slv. during collect.) & SAT & 525.0 \\
NAP & cls3.id171 & SAT & 1041.2 & SAT & 26.6 & 26.5 & (slv. during collect.) & TO & 3601.5 \\
NAP & cls0.id3 & SAT & 0.8 & SAT & 0.8 & 0.8 & (slv. during collect.) & SAT & 0.6 \\
NAP & cls4.id187 & TO & 3601.5 & SAT & 44.6 & 44.5 & (slv. during collect.) & SAT & 42.5 \\
NAP & cls0.id37 & TO & 3601.5 & SAT & 1.6 & 1.4 & (slv. during collect.) & SAT & 11.1 \\
NAP & cls3.id158 & SAT & 2941.8 & SAT & 3221.1 & 215.5 &  & SAT & 120.0 \\
NAP & cls2.id90 & SAT & 2677.2 & SAT & 19.4 & 19.3 & (slv. during collect.) & SAT & 41.8 \\
NAP & cls3.id119 & SAT & 3.1 & SAT & 68.5 & 68.3 & (slv. during collect.) & SAT & 755.2 \\
NAP & cls0.id56 & SAT & 122.0 & SAT & 13.8 & 13.7 & (slv. during collect.) & SAT & 0.6 \\
NAP & cls3.id110 & SAT & 275.6 & SAT & 10.9 & 10.8 & (slv. during collect.) & SAT & 66.9 \\
NAP & cls0.id72 & SAT & 329.2 & SAT & 62.4 & 62.3 & (slv. during collect.) & SAT & 86.8 \\
NAP & cls3.id117 & SAT & 1.9 & SAT & 31.4 & 31.3 & (slv. during collect.) & SAT & 3.4 \\
NAP & cls0.id2 & SAT & 3006.0 & SAT & 0.8 & 0.6 & (slv. during collect.) & SAT & 0.7 \\
NAP & cls3.id131 & SAT & 893.9 & SAT & 74.8 & 74.6 & (slv. during collect.) & SAT & 86.0 \\
NAP & cls4.id175 & TO & 3601.5 & SAT & 2.9 & 2.8 & (slv. during collect.) & SAT & 42.1 \\
NAP & cls0.id46 & TO & 3601.5 & SAT & 8.1 & 7.9 & (slv. during collect.) & SAT & 62.7 \\
NAP & cls4.id229 & SAT & 3570.4 & SAT & 1.2 & 1.2 & (slv. during collect.) & SAT & 217.1 \\
NAP & cls0.id34 & SAT & 85.5 & SAT & 21.9 & 21.7 & (slv. during collect.) & SAT & 263.1 \\
NAP & cls0.id54 & TO & 3601.5 & SAT & 35.4 & 35.3 & (slv. during collect.) & SAT & 0.9 \\
NAP & cls4.id210 & TO & 3601.4 & SAT & 50.7 & 50.6 & (slv. during collect.) & SAT & 499.4 \\
NAP & cls4.id179 & TO & 3601.5 & SAT & 2.5 & 2.3 & (slv. during collect.) & SAT & 3.2 \\
NAP & cls3.id170 & SAT & 1973.8 & SAT & 31.7 & 31.6 & (slv. during collect.) & SAT & 140.0 \\
NAP & cls3.id135 & SAT & 626.9 & SAT & 845.6 & 216.2 &  & SAT & 986.4 \\
NAP & cls0.id35 & SAT & 138.9 & SAT & 11.4 & 11.3 & (slv. during collect.) & SAT & 34.2 \\
NAP & cls0.id27 & SAT & 1016.1 & SAT & 13.2 & 13.1 & (slv. during collect.) & SAT & 11.3 \\
NAP & cls0.id59 & SAT & 288.9 & SAT & 75.2 & 75.1 & (slv. during collect.) & SAT & 87.1 \\
NAP & cls3.id140 & SAT & 316.8 & SAT & 91.4 & 91.3 & (slv. during collect.) & SAT & 141.9 \\
NAP & cls4.id218 & TO & 3601.5 & SAT & 34.8 & 34.7 & (slv. during collect.) & SAT & 215.9 \\
NAP & cls0.id73 & TO & 3601.5 & SAT & 7.7 & 7.5 & (slv. during collect.) & SAT & 51.3 \\
NAP & cls3.id127 & SAT & 2945.3 & SAT & 3148.0 & 187.6 &  & TO & 3601.5 \\
NAP & cls3.id125 & SAT & 1250.8 & SAT & 65.8 & 65.7 & (slv. during collect.) & SAT & 1.2 \\
NAP & cls3.id129 & SAT & 152.0 & SAT & 353.1 & 201.7 &  & SAT & 206.9 \\
NAP & cls4.id224 & TO & 3601.5 & SAT & 88.0 & 87.8 & (slv. during collect.) & TO & 3601.4 \\
\bottomrule
\end{xltabular}
\end{center}

\subsection{\marabou vs. \cncTuneAndValSeq}
\begin{center}
\scriptsize
\setlength{\tabcolsep}{10.3pt}
\begin{xltabular}{\textwidth}{@{}llrrrrr@{}}
\toprule
& & \multicolumn{2}{c}{\tabtitle{\marabou}} & \multicolumn{3}{c}{\tabtitle{\cncTuneAndValSeq}} \\
\cmidrule(lr){3-4} \cmidrule(lr){5-7}
\tabtitle{Family} & \tabtitle{Benchmark}  & \tabtitle{Res.} & \coltimetotal & \tabtitle{Res.} & \coltimetotal & \coltimelearning \\
\endhead
AltLoop & REI.id298.ep79143 & UNS & 747.5 & UNS & 440.3 & 52.2 \\
AltLoop & REI.id303.ep95137 & SAT & 0.0 & SAT & 69.3 & 69.2 \\
AltLoop & REI.id549.ep99098 & SAT & 0.0 & SAT & 118.0 & 117.9 \\
AltLoop & REI.id512.ep97385 & SAT & 511.3 & SAT & 136.9 & 136.8 \\
AltLoop & REI.id234.ep70556 & SAT & 64.2 & SAT & 7.8 & 7.7 \\
AltLoop & REI.id68.ep95854 & SAT & 0.0 & SAT & 155.0 & 154.9 \\
AltLoop & REI.id174.ep97933 & SAT & 0.0 & SAT & 0.7 & 0.6 \\
AltLoop & REI.id418.ep92671 & SAT & 0.0 & SAT & 133.2 & 133.1 \\
AltLoop & REI.id418.ep91524 & SAT & 318.1 & SAT & 102.0 & 101.9 \\
AltLoop & REI.id518.ep95900 & UNS & 338.5 & UNS & 437.4 & 56.8 \\
AltLoop & REI.id549.ep93366 & SAT & 0.0 & SAT & 4.8 & 4.7 \\
AltLoop & REI.id35.ep85842 & SAT & 0.0 & SAT & 7.8 & 7.7 \\
AltLoop & REI.id76.ep91480 & UNS & 633.1 & UNS & 398.2 & 55.0 \\
AltLoop & REI.id176.ep95516 & UNS & 759.2 & UNS & 713.1 & 62.3 \\
AltLoop & REI.id239.ep71386 & UNS & 3661.5 & UNS & 609.8 & 56.5 \\
AltLoop & REI.id343.ep73171 & UNS & 1149.4 & UNS & 935.2 & 69.7 \\
AltLoop & REI.id444.ep87148 & UNS & 888.7 & UNS & 567.0 & 34.0 \\
AltLoop & REI.id381.ep85558 & SAT & 820.0 & SAT & 128.4 & 128.3 \\
AltLoop & REI.id418.ep92561 & SAT & 0.0 & SAT & 118.8 & 118.7 \\
AltLoop & REI.id319.ep84434 & SAT & 0.0 & SAT & 99.1 & 99.0 \\
AltLoop & REI.id457.ep46821 & SAT & 0.0 & SAT & 33.0 & 32.9 \\
AltLoop & REI.id240.ep98115 & UNS & 460.3 & UNS & 210.2 & 56.1 \\
AltLoop & REI.id491.ep67504 & SAT & 111.5 & SAT & 153.7 & 153.5 \\
AltLoop & REI.id319.ep86630 & SAT & 0.0 & SAT & 85.4 & 85.3 \\
AltLoop & REI.id457.ep69348 & SAT & 0.0 & SAT & 14.6 & 14.4 \\
AltLoop & REI.id381.ep93562 & UNS & 1755.0 & UNS & 1178.1 & 65.6 \\
AltLoop & REI.id180.ep77624 & SAT & 0.0 & SAT & 4.6 & 4.5 \\
AltLoop & REI.id530.ep82185 & SAT & 0.0 & SAT & 11.5 & 11.4 \\
AltLoop & REI.id81.ep91571 & UNS & 1148.1 & UNS & 2141.6 & 71.3 \\
AltLoop & REI.id491.ep96326 & SAT & 1361.5 & SAT & 28.7 & 28.5 \\
AltLoop & REI.id158.ep53113 & SAT & 135.7 & SAT & 14.5 & 14.3 \\
AltLoop & REI.id381.ep92541 & UNS & 1374.2 & UNS & 862.8 & 66.5 \\
AltLoop & REI.id165.ep82822 & SAT & 0.0 & SAT & 0.7 & 0.5 \\
AltLoop & REI.id165.ep94522 & SAT & 124.4 & SAT & 2.0 & 1.8 \\
AltLoop & REI.id176.ep96138 & UNS & 493.8 & UNS & 348.7 & 50.2 \\
AltLoop & REI.id298.ep97154 & SAT & 0.0 & SAT & 3.2 & 3.0 \\
AltLoop & REI.id91.ep75275 & UNS & 533.6 & UNS & 229.8 & 50.2 \\
AltLoop & REI.id318.ep85341 & SAT & 347.7 & SAT & 323.3 & 72.7 \\
AltLoop & REI.id530.ep91636 & SAT & 0.0 & SAT & 166.1 & 166.0 \\
AltLoop & REI.id114.ep97993 & SAT & 182.5 & SAT & 0.9 & 0.8 \\
AltLoop & REI.id76.ep90543 & UNS & 579.2 & UNS & 354.8 & 60.1 \\
AltLoop & REI.id549.ep87887 & SAT & 0.0 & SAT & 370.9 & 370.7 \\
AltLoop & REI.id128.ep85878 & UNS & 4491.0 & UNS & 452.8 & 59.1 \\
AltLoop & REI.id77.ep91446 & UNS & 427.6 & UNS & 247.3 & 44.8 \\
AltLoop & REI.id379.ep87708 & UNS & 2085.0 & UNS & 688.2 & 67.2 \\
AltLoop & REI.id512.ep97330 & SAT & 533.0 & SAT & 136.9 & 136.8 \\
AltLoop & REI.id165.ep96974 & SAT & 188.6 & SAT & 5.7 & 5.5 \\
AltLoop & REI.id491.ep67924 & SAT & 0.0 & SAT & 6.8 & 6.7 \\
AltLoop & REI.id318.ep96308 & SAT & 767.5 & SAT & 167.4 & 167.2 \\
AltLoop & REI.id286.ep89999 & UNS & 110.0 & UNS & 161.4 & 60.0 \\
AltLoop & REI.id35.ep61486 & SAT & 0.0 & SAT & 50.1 & 49.9 \\
AltLoop & REI.id234.ep83780 & SAT & 0.0 & SAT & 2.1 & 1.9 \\
AltLoop & REI.id73.ep74414 & SAT & 345.4 & SAT & 154.3 & 154.2 \\
AltLoop & REI.id286.ep94925 & UNS & 183.3 & UNS & 534.3 & 54.7 \\
AltLoop & REI.id91.ep58150 & UNS & 208.7 & UNS & 165.2 & 61.8 \\
AltLoop & REI.id68.ep95942 & SAT & 0.0 & SAT & 56.1 & 56.0 \\
AltLoop & REI.id196.ep73462 & SAT & 0.1 & SAT & 4.8 & 4.7 \\
AltLoop & REI.id128.ep91042 & UNS & 336.3 & UNS & 308.5 & 56.8 \\
AltLoop & REI.id540.ep57595 & UNS & 4232.3 & UNS & 503.2 & 89.9 \\
AltLoop & REI.id239.ep73504 & SAT & 22.7 & SAT & 99.1 & 68.5 \\
AltLoop & REI.id549.ep99331 & SAT & 0.0 & SAT & 157.4 & 157.3 \\
AltLoop & REI.id318.ep75036 & SAT & 1361.0 & SAT & 43.3 & 43.2 \\
AltLoop & REI.id530.ep82130 & SAT & 0.0 & SAT & 436.7 & 436.5 \\
AltLoop & REI.id512.ep70217 & SAT & 225.6 & SAT & 468.0 & 49.9 \\
AltLoop & REI.id180.ep58461 & SAT & 33.9 & SAT & 2.4 & 2.2 \\
AltLoop & REI.id176.ep92206 & UNS & 1802.5 & UNS & 559.4 & 84.1 \\
AltLoop & REI.id239.ep90118 & SAT & 0.0 & SAT & 4.5 & 4.4 \\
AltLoop & REI.id319.ep62661 & SAT & 0.0 & SAT & 297.9 & 297.8 \\
AltLoop & REI.id215.ep97421 & SAT & 0.0 & SAT & 13.9 & 13.8 \\
AltLoop & REI.id308.ep98363 & SAT & 309.0 & SAT & 9.8 & 9.7 \\
AltLoop & REI.id457.ep78403 & SAT & 0.0 & SAT & 29.9 & 29.7 \\
AltLoop & REI.id267.ep96333 & UNS & 0.0 & UNS & 117.1 & 117.0 \\
AltLoop & REI.id240.ep89317 & UNS & 159.7 & UNS & 448.5 & 69.0 \\
AltLoop & REI.id114.ep94159 & SAT & 215.9 & SAT & 0.8 & 0.6 \\
AltLoop & REI.id502.ep71445 & SAT & 0.0 & SAT & 142.3 & 142.2 \\
AltLoop & REI.id502.ep90122 & SAT & 0.0 & SAT & 12.5 & 12.4 \\
AltLoop & REI.id425.ep93745 & UNS & 4007.1 & UNS & 1748.3 & 66.7 \\
AltLoop & REI.id174.ep81057 & SAT & 0.1 & SAT & 0.8 & 0.7 \\
AltLoop & REI.id540.ep57610 & UNS & 3156.6 & UNS & 471.9 & 69.3 \\
AltLoop & REI.id165.ep95976 & SAT & 92.8 & SAT & 1.8 & 1.6 \\
AltLoop & REI.id296.ep79125 & UNS & 259.9 & UNS & 225.7 & 36.6 \\
AltLoop & REI.id549.ep85789 & SAT & 0.0 & SAT & 3.8 & 3.7 \\
AltLoop & REI.id379.ep96909 & SAT & 1277.8 & SAT & 574.6 & 76.1 \\
AltLoop & REI.id180.ep69795 & SAT & 0.0 & SAT & 5.3 & 5.1 \\
AltLoop & REI.id308.ep99303 & SAT & 377.7 & SAT & 168.3 & 168.2 \\
AltLoop & REI.id68.ep95821 & SAT & 0.0 & SAT & 274.4 & 274.2 \\
AltLoop & REI.id299.ep69400 & UNS & 462.2 & UNS & 412.4 & 60.3 \\
AltLoop & REI.id528.ep89177 & SAT & 0.0 & SAT & 113.2 & 113.1 \\
AltLoop & REI.id240.ep85050 & UNS & 433.8 & UNS & 270.6 & 55.8 \\
AltLoop & REI.id528.ep95627 & SAT & 71.7 & SAT & 177.8 & 177.7 \\
AltLoop & REI.id318.ep89536 & SAT & 894.8 & SAT & 46.0 & 45.9 \\
AltLoop & REI.id318.ep94480 & SAT & 1635.2 & SAT & 167.4 & 167.3 \\
AltLoop & REI.id418.ep99126 & SAT & 2703.4 & SAT & 7.5 & 7.3 \\
AltLoop & REI.id308.ep94787 & UNS & 409.0 & UNS & 345.2 & 38.7 \\
AltLoop & REI.id399.ep69315 & UNS & 187.8 & UNS & 240.8 & 69.0 \\
AltLoop & REI.id90.ep84292 & UNS & 3093.0 & UNS & 1279.3 & 47.7 \\
AltLoop & REI.id165.ep94134 & SAT & 59.9 & SAT & 2.1 & 2.0 \\
AltLoop & REI.id343.ep63040 & UNS & 60.0 & UNS & 178.0 & 76.9 \\
AltLoop & REI.id47.ep79665 & SAT & 43.7 & SAT & 22.5 & 22.4 \\
AltLoop & REI.id77.ep99229 & UNS & 534.6 & UNS & 208.2 & 59.1 \\
AltLoop & REI.id201.ep48186 & SAT & 62.4 & SAT & 3.4 & 3.3 \\
AltLoop & REI.id502.ep84337 & SAT & 0.0 & SAT & 142.3 & 142.2 \\
AltLoop & REI.id81.ep97775 & UNS & 1289.2 & UNS & 920.0 & 59.9 \\
AltLoop & REI.id540.ep57775 & UNS & 4935.6 & UNS & 562.3 & 73.8 \\
AltLoop & REI.id77.ep99119 & UNS & 564.9 & UNS & 211.8 & 65.2 \\
AltLoop & REI.id393.ep97111 & SAT & 0.0 & SAT & 103.6 & 42.7 \\
AltLoop & REI.id298.ep85589 & UNS & 589.3 & UNS & 466.5 & 48.6 \\
AltLoop & REI.id303.ep95038 & SAT & 0.0 & SAT & 5.5 & 5.3 \\
AltLoop & REI.id352.ep49376 & UNS & 24.7 & UNS & 177.9 & 177.7 \\
AltLoop & REI.id528.ep55711 & SAT & 279.1 & SAT & 438.4 & 438.3 \\
AltLoop & REI.id67.ep99214 & SAT & 81.1 & SAT & 197.0 & 196.9 \\
AltLoop & REI.id267.ep95713 & UNS & 0.1 & UNS & 70.6 & 70.5 \\
AltLoop & REI.id321.ep81083 & UNS & 111.8 & UNS & 369.9 & 58.5 \\
AltLoop & REI.id299.ep92705 & UNS & 1241.0 & UNS & 378.1 & 66.1 \\
AltLoop & REI.id512.ep75804 & SAT & 581.4 & SAT & 142.1 & 141.9 \\
AltLoop & REI.id530.ep82207 & SAT & 0.0 & SAT & 412.4 & 412.3 \\
AltLoop & REI.id176.ep98430 & UNS & 1510.7 & UNS & 1078.4 & 54.4 \\
AltLoop & REI.id399.ep62611 & UNS & 173.7 & UNS & 152.0 & 60.3 \\
AltLoop & REI.id201.ep39282 & SAT & 90.2 & SAT & 2.7 & 2.5 \\
AltLoop & REI.id67.ep95115 & SAT & 47.0 & SAT & 156.8 & 156.6 \\
AltLoop & REI.id128.ep87384 & SAT & 1147.2 & SAT & 8.7 & 8.6 \\
AltLoop & REI.id77.ep49509 & UNS & 150.6 & UNS & 181.9 & 59.8 \\
AltLoop & REI.id239.ep79112 & UNS & 3356.2 & UNS & 474.6 & 58.6 \\
AltLoop & REI.id444.ep78487 & UNS & 1754.0 & UNS & 1036.5 & 44.3 \\
AltLoop & REI.id296.ep77867 & UNS & 95.4 & UNS & 162.2 & 31.3 \\
AltLoop & REI.id286.ep89988 & UNS & 109.8 & UNS & 98.9 & 43.5 \\
AltLoop & REI.id81.ep91408 & UNS & 1029.7 & UNS & 666.4 & 54.1 \\
AltLoop & REI.id308.ep97397 & SAT & 228.5 & SAT & 139.9 & 36.7 \\
AltLoop & REI.id535.ep76894 & UNS & 229.8 & UNS & 410.3 & 58.2 \\
AltLoop & REI.id73.ep74469 & SAT & 273.8 & SAT & 218.1 & 218.0 \\
AltLoop & REI.id303.ep94994 & SAT & 0.0 & SAT & 5.4 & 5.2 \\
AltLoop & REI.id321.ep93439 & UNS & 71.2 & UNS & 250.0 & 94.9 \\
AltLoop & REI.id381.ep92994 & SAT & 1479.5 & SAT & 29.6 & 29.4 \\
AltLoop & REI.id502.ep71027 & SAT & 0.0 & SAT & 201.0 & 200.9 \\
AltLoop & REI.id502.ep86449 & SAT & 0.0 & SAT & 142.2 & 142.0 \\
AltLoop & REI.id174.ep91107 & SAT & 0.0 & SAT & 0.8 & 0.6 \\
AltLoop & REI.id76.ep92143 & UNS & 206.6 & UNS & 242.0 & 56.3 \\
AltLoop & REI.id528.ep89331 & SAT & 64.1 & SAT & 108.5 & 108.4 \\
AltLoop & REI.id176.ep99427 & UNS & 1021.9 & UNS & 569.5 & 45.6 \\
AltLoop & REI.id352.ep68171 & UNS & 150.3 & UNS & 195.6 & 52.9 \\
AltLoop & REI.id47.ep96122 & SAT & 45.5 & SAT & 70.5 & 70.3 \\
AltLoop & REI.id286.ep90032 & UNS & 96.8 & UNS & 108.5 & 44.5 \\
AltLoop & REI.id35.ep94624 & SAT & 0.0 & SAT & 135.1 & 135.0 \\
AltLoop & REI.id425.ep83340 & SAT & 0.0 & SAT & 536.4 & 536.3 \\
AltLoop & REI.id399.ep92145 & UNS & 436.5 & UNS & 217.1 & 71.2 \\
AltLoop & REI.id399.ep96449 & SAT & 2964.5 & SAT & 812.7 & 50.6 \\
AltLoop & REI.id352.ep70312 & UNS & 43.2 & UNS & 125.5 & 125.4 \\
AltLoop & REI.id267.ep91426 & UNS & 0.1 & UNS & 49.6 & 49.5 \\
AltLoop & REI.id77.ep75183 & UNS & 220.0 & UNS & 210.0 & 53.4 \\
AltLoop & REI.id114.ep67447 & SAT & 59.1 & SAT & 0.7 & 0.6 \\
AltLoop & REI.id114.ep84297 & UNS & 742.2 & UNS & 562.7 & 49.9 \\
AltLoop & REI.id425.ep96981 & SAT & 0.0 & SAT & 54.2 & 54.0 \\
AltLoop & REI.id68.ep95777 & SAT & 0.0 & SAT & 197.0 & 196.9 \\
AltLoop & REI.id352.ep51510 & UNS & 73.5 & UNS & 276.8 & 79.9 \\
AltLoop & REI.id379.ep80359 & SAT & 28.5 & SAT & 205.2 & 205.1 \\
AltLoop & REI.id444.ep92635 & UNS & 1172.2 & UNS & 518.5 & 56.0 \\
AltLoop & REI.id158.ep63427 & UNS & 115.4 & UNS & 121.8 & 54.3 \\
AltLoop & REI.id128.ep68334 & UNS & 1027.1 & UNS & 352.8 & 56.5 \\
AltLoop & REI.id196.ep43251 & SAT & 0.0 & SAT & 5.2 & 5.1 \\
AltLoop & REI.id73.ep82370 & UNS & 2597.3 & UNS & 364.0 & 48.2 \\
AltLoop & REI.id286.ep95719 & UNS & 145.9 & UNS & 522.6 & 56.0 \\
AltLoop & REI.id303.ep69708 & SAT & 47.2 & SAT & 116.6 & 81.6 \\
AltLoop & REI.id90.ep97680 & UNS & 1533.5 & UNS & 281.8 & 65.7 \\
AltLoop & REI.id91.ep95105 & UNS & 1272.3 & UNS & 565.9 & 56.7 \\
AltLoop & REI.id158.ep53554 & SAT & 285.3 & SAT & 212.6 & 54.9 \\
AltLoop & REI.id308.ep98830 & UNS & 534.0 & UNS & 334.5 & 41.3 \\
AltLoop & REI.id418.ep99947 & SAT & 1007.0 & SAT & 105.2 & 105.0 \\
AltLoop & REI.id379.ep71969 & SAT & 871.0 & SAT & 617.2 & 50.6 \\
AltLoop & REI.id512.ep70283 & SAT & 245.3 & SAT & 509.2 & 42.6 \\
AltLoop & REI.id90.ep99033 & TO & 14503.1 & UNS & 1248.5 & 65.7 \\
AltLoop & REI.id425.ep98177 & SAT & 0.0 & SAT & 97.9 & 97.8 \\
AltLoop & REI.id530.ep94463 & SAT & 0.0 & SAT & 58.8 & 58.6 \\
AltLoop & REI.id215.ep97333 & SAT & 0.0 & SAT & 12.9 & 12.8 \\
AltLoop & REI.id267.ep96300 & UNS & 0.1 & UNS & 70.1 & 70.0 \\
AltLoop & REI.id234.ep71965 & SAT & 30.3 & SAT & 20.5 & 20.4 \\
AltLoop & REI.id457.ep79411 & SAT & 0.0 & SAT & 1.0 & 0.8 \\
AltLoop & REI.id518.ep68431 & UNS & 269.5 & UNS & 417.7 & 55.5 \\
AltLoop & REI.id201.ep48219 & SAT & 59.0 & SAT & 72.9 & 72.8 \\
AltLoop & REI.id425.ep69300 & SAT & 356.1 & SAT & 17.3 & 17.1 \\
AltLoop & REI.id128.ep97983 & UNS & 2546.4 & UNS & 451.3 & 87.0 \\
AltLoop & REI.id234.ep83703 & SAT & 0.0 & SAT & 1.7 & 1.5 \\
AltLoop & REI.id518.ep95639 & UNS & 303.1 & UNS & 405.1 & 59.8 \\
AltLoop & REI.id444.ep77347 & UNS & 1318.2 & UNS & 953.4 & 35.2 \\
AltLoop & REI.id201.ep88825 & SAT & 0.0 & SAT & 2.3 & 2.2 \\
AltLoop & REI.id215.ep96277 & SAT & 0.0 & SAT & 1.7 & 1.6 \\
AltLoop & REI.id35.ep81067 & SAT & 0.0 & SAT & 220.1 & 220.0 \\
AltLoop & REI.id239.ep73100 & UNS & 3207.2 & UNS & 609.8 & 57.6 \\
AltLoop & REI.id174.ep98455 & SAT & 0.0 & SAT & 0.9 & 0.8 \\
AltLoop & REI.id535.ep87226 & SAT & 0.0 & SAT & 226.7 & 226.6 \\
AltLoop & REI.id528.ep90341 & SAT & 45.5 & SAT & 25.7 & 25.5 \\
AltLoop & REI.id298.ep75768 & UNS & 728.3 & UNS & 515.5 & 39.4 \\
AltLoop & REI.id91.ep58172 & UNS & 211.9 & UNS & 147.2 & 53.6 \\
AltLoop & REI.id180.ep61293 & SAT & 63.5 & SAT & 8.8 & 8.7 \\
AltLoop & REI.id299.ep99426 & UNS & 538.1 & UNS & 332.6 & 71.6 \\
AltLoop & REI.id47.ep92041 & SAT & 0.0 & SAT & 21.7 & 21.6 \\
AltLoop & REI.id535.ep87193 & SAT & 32.6 & SAT & 165.5 & 165.4 \\
AltLoop & REI.id296.ep74854 & UNS & 233.2 & UNS & 486.1 & 54.2 \\
AltLoop & REI.id90.ep94747 & UNS & 2791.3 & UNS & 1337.5 & 49.3 \\
AltLoop & REI.id158.ep53499 & SAT & 301.9 & SAT & 145.6 & 57.7 \\
AltLoop & REI.id381.ep85470 & SAT & 1112.5 & SAT & 158.9 & 158.8 \\
AltLoop & REI.id399.ep85537 & UNS & 410.4 & UNS & 605.2 & 42.0 \\
AltLoop & REI.id319.ep89649 & SAT & 0.0 & SAT & 8.3 & 8.2 \\
AltLoop & REI.id321.ep92994 & UNS & 85.5 & UNS & 286.1 & 87.1 \\
AltLoop & REI.id240.ep98891 & UNS & 307.3 & UNS & 271.0 & 56.4 \\
AltLoop & REI.id81.ep89983 & UNS & 765.1 & UNS & 440.6 & 47.6 \\
AltLoop & REI.id491.ep70134 & SAT & 110.7 & SAT & 2.5 & 2.4 \\
AltLoop & REI.id296.ep76414 & UNS & 278.6 & UNS & 980.9 & 35.9 \\
AltLoop & REI.id518.ep70223 & UNS & 335.0 & UNS & 271.6 & 51.7 \\
AltLoop & REI.id91.ep54663 & UNS & 338.6 & UNS & 179.6 & 53.5 \\
AltLoop & REI.id240.ep91714 & UNS & 322.4 & UNS & 420.1 & 67.9 \\
AltLoop & REI.id352.ep69872 & UNS & 63.6 & UNS & 329.7 & 83.0 \\
AltLoop & REI.id343.ep63051 & UNS & 58.0 & UNS & 342.0 & 86.0 \\
AltLoop & REI.id196.ep41633 & SAT & 0.0 & SAT & 5.0 & 4.9 \\
AltLoop & REI.id319.ep84478 & SAT & 0.0 & SAT & 423.3 & 423.2 \\
AltLoop & REI.id196.ep99092 & SAT & 0.0 & SAT & 12.1 & 12.0 \\
AltLoop & REI.id68.ep98804 & SAT & 0.0 & SAT & 275.1 & 275.0 \\
AltLoop & REI.id196.ep86142 & SAT & 0.0 & SAT & 2.9 & 2.8 \\
AltLoop & REI.id81.ep97833 & UNS & 2134.4 & UNS & 473.1 & 68.7 \\
AltLoop & REI.id343.ep55239 & UNS & 29.8 & UNS & 167.5 & 167.3 \\
AltLoop & REI.id180.ep58329 & SAT & 32.3 & SAT & 2.3 & 2.1 \\
AltLoop & REI.id201.ep56910 & SAT & 89.3 & SAT & 15.0 & 14.9 \\
AltLoop & REI.id321.ep86834 & UNS & 79.8 & UNS & 335.2 & 84.6 \\
AltLoop & REI.id158.ep59378 & UNS & 122.7 & UNS & 120.1 & 46.4 \\
AltLoop & REI.id298.ep77187 & SAT & 35.1 & SAT & 119.2 & 69.8 \\
AltLoop & REI.id234.ep71035 & SAT & 42.0 & SAT & 4.0 & 3.9 \\
AltLoop & REI.id343.ep96195 & UNS & 402.6 & UNS & 277.6 & 68.2 \\
AltLoop & REI.id518.ep96115 & UNS & 336.6 & UNS & 636.9 & 66.6 \\
AltLoop & REI.id535.ep76916 & UNS & 220.6 & UNS & 410.4 & 44.2 \\
AltLoop & REI.id540.ep94690 & SAT & 0.0 & SAT & 4.7 & 4.6 \\
AltLoop & REI.id67.ep87611 & UNS & 205.1 & UNS & 352.7 & 71.3 \\
AltLoop & REI.id267.ep96426 & UNS & 0.1 & UNS & 57.2 & 57.0 \\
AltLoop & REI.id114.ep94581 & SAT & 175.0 & SAT & 0.9 & 0.7 \\
AltLoop & REI.id393.ep97221 & SAT & 0.0 & SAT & 636.0 & 73.7 \\
AltLoop & REI.id73.ep84023 & UNS & 730.4 & UNS & 258.8 & 47.0 \\
AltLoop & REI.id215.ep97355 & SAT & 0.0 & SAT & 14.6 & 14.5 \\
AltLoop & REI.id299.ep81175 & UNS & 268.8 & UNS & 401.4 & 63.4 \\
AltLoop & REI.id321.ep55647 & UNS & 31.2 & UNS & 397.4 & 397.3 \\
AltLoop & REI.id444.ep96160 & UNS & 510.6 & UNS & 527.3 & 34.4 \\
AltLoop & REI.id379.ep97133 & SAT & 173.6 & SAT & 150.0 & 49.7 \\
AltLoop & REI.id393.ep88789 & SAT & 0.0 & SAT & 349.7 & 45.5 \\
AltLoop & REI.id215.ep98767 & SAT & 0.0 & SAT & 13.4 & 13.2 \\
AltLoop & REI.id393.ep97199 & SAT & 0.0 & SAT & 277.7 & 74.1 \\
AltLoop & REI.id457.ep81404 & SAT & 0.0 & SAT & 5.2 & 5.0 \\
AltLoop & REI.id47.ep99600 & SAT & 62.3 & SAT & 408.2 & 408.0 \\
AltLoop & REI.id73.ep91504 & SAT & 0.0 & SAT & 53.5 & 53.3 \\
AltLoop & REI.id47.ep99688 & SAT & 50.2 & SAT & 13.8 & 13.7 \\
AltLoop & REI.id296.ep74969 & UNS & 227.1 & UNS & 342.6 & 36.5 \\
AltLoop & REI.id491.ep54571 & SAT & 249.5 & SAT & 153.6 & 153.5 \\
AltLoop & REI.id393.ep90550 & SAT & 182.0 & SAT & 65.0 & 64.9 \\
AltLoop & REI.id535.ep91323 & UNS & 284.6 & UNS & 333.1 & 61.4 \\
AltLoop & REI.id35.ep61788 & SAT & 7.4 & SAT & 79.7 & 79.5 \\
AltLoop & REI.id303.ep68520 & UNS & 1640.8 & UNS & 913.0 & 59.3 \\
AltLoop & REI.id67.ep82981 & UNS & 122.5 & UNS & 365.9 & 52.9 \\
AltLoop & REI.id76.ep98692 & UNS & 335.2 & UNS & 230.4 & 58.4 \\
AltLoop & REI.id299.ep69103 & UNS & 422.9 & UNS & 373.9 & 53.6 \\
AltLoop & REI.id174.ep66258 & SAT & 34.7 & SAT & 0.8 & 0.6 \\
AltLoop & REI.id76.ep99198 & UNS & 769.9 & UNS & 410.9 & 55.6 \\
AltLoop & REI.id540.ep85904 & SAT & 0.0 & SAT & 18.1 & 17.9 \\
AltLoop & REI.id67.ep84585 & UNS & 183.5 & UNS & 352.6 & 45.2 \\
NAP & cls0.id11 & SAT & 11898.8 & SAT & 76.0 & 75.8 \\
NAP & cls0.id5 & SAT & 0.4 & SAT & 3.9 & 3.7 \\
NAP & cls3.id169 & SAT & 5020.3 & SAT & 1395.3 & 966.8 \\
NAP & cls4.id173 & TO & 14400.0 & SAT & 6.3 & 6.2 \\
NAP & cls1.id85 & TO & 14400.0 & SAT & 178.5 & 178.3 \\
NAP & cls1.id81 & TO & 14400.0 & SAT & 377.1 & 376.9 \\
NAP & cls3.id127 & TO & 14400.0 & TO & 14400.0 & 1151.9 \\
NAP & cls1.id88 & TO & 14400.0 & SAT & 651.3 & 651.0 \\
NAP & cls3.id156 & SAT & 0.8 & SAT & 4.7 & 4.6 \\
NAP & cls3.id154 & TO & 14400.0 & SAT & 35.5 & 35.4 \\
NAP & cls3.id144 & TO & 14400.0 & SAT & 1.9 & 1.8 \\
NAP & cls4.id209 & TO & 14400.0 & SAT & 9.6 & 9.5 \\
NAP & cls0.id38 & TO & 14400.0 & SAT & 8240.0 & 882.3 \\
NAP & cls0.id71 & SAT & 5437.8 & SAT & 32.5 & 32.3 \\
NAP & cls1.id80 & SAT & 699.2 & SAT & 424.1 & 424.0 \\
NAP & cls0.id66 & SAT & 284.5 & SAT & 1182.7 & 997.0 \\
NAP & cls2.id95 & TO & 14400.0 & SAT & 214.7 & 214.5 \\
NAP & cls4.id187 & TO & 14400.0 & SAT & 281.3 & 281.2 \\
NAP & cls4.id207 & TO & 14400.0 & SAT & 233.7 & 233.6 \\
NAP & cls3.id166 & TO & 14400.0 & SAT & 512.6 & 512.5 \\
NAP & cls2.id96 & TO & 14400.0 & SAT & 7.5 & 7.4 \\
NAP & cls4.id204 & TO & 14400.0 & TO & 14400.0 & 1065.5 \\
NAP & cls2.id97 & TO & 14400.0 & SAT & 122.9 & 122.7 \\
NAP & cls2.id93 & TO & 14400.0 & SAT & 13214.1 & 1110.1 \\
NAP & cls4.id232 & TO & 14400.0 & TO & 14400.0 & 1076.5 \\
NAP & cls0.id57 & SAT & 8335.5 & SAT & 66.6 & 66.4 \\
NAP & cls0.id4 & SAT & 0.8 & SAT & 16.6 & 16.3 \\
NAP & cls4.id181 & SAT & 47.8 & SAT & 128.5 & 128.4 \\
NAP & cls3.id153 & TO & 14400.0 & SAT & 199.8 & 199.6 \\
NAP & cls0.id34 & SAT & 13902.2 & SAT & 138.6 & 138.5 \\
NAP & cls3.id148 & TO & 14400.0 & TO & 14400.0 & 988.2 \\
NAP & cls4.id225 & TO & 14400.0 & SAT & 97.2 & 97.0 \\
NAP & cls0.id31 & SAT & 4148.0 & SAT & 0.8 & 0.6 \\
NAP & cls0.id33 & TO & 14400.0 & SAT & 11.2 & 11.1 \\
NAP & cls1.id77 & SAT & 12369.8 & TO & 14400.0 & 967.2 \\
NAP & cls0.id72 & SAT & 9185.6 & SAT & 410.9 & 410.7 \\
NAP & cls3.id118 & SAT & 5.2 & SAT & 453.1 & 452.9 \\
NAP & cls0.id52 & TO & 14400.0 & SAT & 26.8 & 26.5 \\
NAP & cls4.id189 & TO & 14400.0 & SAT & 67.7 & 67.5 \\
NAP & cls3.id123 & SAT & 271.5 & SAT & 75.7 & 75.5 \\
NAP & cls4.id186 & TO & 14400.0 & SAT & 13340.0 & 1075.2 \\
NAP & cls4.id182 & SAT & 3624.3 & SAT & 7.9 & 7.8 \\
NAP & cls4.id219 & TO & 14400.0 & TO & 14400.0 & 999.0 \\
NAP & cls3.id116 & SAT & 11796.3 & SAT & 7.9 & 7.8 \\
NAP & cls3.id115 & SAT & 0.6 & SAT & 71.9 & 71.8 \\
NAP & cls2.id92 & TO & 14400.0 & SAT & 123.8 & 123.6 \\
NAP & cls0.id21 & SAT & 1.3 & SAT & 4.4 & 4.3 \\
NAP & cls3.id159 & TO & 14400.0 & SAT & 288.6 & 288.4 \\
NAP & cls0.id15 & SAT & 4021.3 & SAT & 74.2 & 74.0 \\
NAP & cls3.id161 & TO & 14400.0 & SAT & 293.8 & 293.7 \\
NAP & cls3.id146 & TO & 14400.0 & SAT & 10029.6 & 1031.5 \\
NAP & cls2.id99 & TO & 14400.0 & SAT & 262.7 & 262.5 \\
NAP & cls0.id20 & SAT & 5123.2 & SAT & 97.0 & 96.8 \\
NAP & cls2.id104 & TO & 14400.0 & SAT & 17.0 & 16.8 \\
NAP & cls3.id106 & SAT & 7596.1 & SAT & 6.5 & 6.4 \\
NAP & cls3.id109 & SAT & 2876.1 & SAT & 43.5 & 43.3 \\
NAP & cls2.id98 & SAT & 5349.1 & SAT & 736.1 & 735.9 \\
NAP & cls0.id24 & TO & 14400.0 & SAT & 15.3 & 15.2 \\
NAP & cls4.id195 & SAT & 0.5 & SAT & 29.6 & 29.5 \\
NAP & cls3.id119 & SAT & 510.2 & SAT & 447.3 & 447.2 \\
NAP & cls0.id1 & SAT & 0.4 & SAT & 3.1 & 3.0 \\
NAP & cls0.id55 & TO & 14400.0 & SAT & 1.2 & 1.0 \\
NAP & cls0.id63 & TO & 14400.0 & SAT & 120.7 & 120.6 \\
NAP & cls3.id155 & TO & 14400.0 & SAT & 1.8 & 1.7 \\
NAP & cls2.id103 & TO & 14400.0 & TO & 14400.0 & 1033.7 \\
NAP & cls0.id67 & TO & 14400.0 & SAT & 166.7 & 166.6 \\
NAP & cls0.id75 & TO & 14400.0 & SAT & 29.4 & 29.3 \\
NAP & cls4.id175 & TO & 14400.0 & SAT & 12.5 & 12.4 \\
NAP & cls0.id14 & SAT & 0.7 & SAT & 20.0 & 19.9 \\
NAP & cls3.id112 & SAT & 1552.2 & SAT & 454.5 & 454.4 \\
NAP & cls0.id30 & SAT & 71.3 & SAT & 46.0 & 45.9 \\
NAP & cls0.id39 & SAT & 539.1 & SAT & 91.0 & 90.9 \\
NAP & cls3.id133 & TO & 14400.0 & TO & 14400.0 & 1200.6 \\
NAP & cls4.id224 & TO & 14400.0 & SAT & 582.9 & 582.8 \\
NAP & cls4.id192 & TO & 14400.0 & SAT & 11.4 & 11.3 \\
NAP & cls0.id76 & TO & 14400.0 & SAT & 138.0 & 137.8 \\
NAP & cls2.id101 & SAT & 939.8 & TO & 14400.0 & 937.4 \\
NAP & cls3.id121 & TO & 14400.0 & SAT & 72.0 & 71.9 \\
NAP & cls1.id78 & SAT & 192.9 & TO & 14400.0 & 821.1 \\
NAP & cls3.id151 & TO & 14400.0 & SAT & 508.5 & 508.3 \\
NAP & cls0.id60 & TO & 14400.0 & SAT & 135.7 & 135.6 \\
NAP & cls3.id168 & SAT & 3735.4 & SAT & 187.4 & 187.3 \\
NAP & cls4.id191 & TO & 14400.0 & SAT & 90.1 & 89.9 \\
NAP & cls3.id160 & TO & 14400.0 & TO & 14400.0 & 1057.3 \\
NAP & cls3.id139 & TO & 14400.0 & SAT & 11114.9 & 1373.7 \\
NAP & cls0.id16 & SAT & 9178.0 & SAT & 35.9 & 35.8 \\
NAP & cls2.id91 & TO & 14400.0 & SAT & 340.1 & 339.9 \\
NAP & cls4.id229 & TO & 14400.0 & SAT & 11.8 & 11.6 \\
NAP & cls0.id74 & TO & 14400.0 & SAT & 16.8 & 16.6 \\
NAP & cls4.id185 & SAT & 580.5 & SAT & 258.2 & 258.1 \\
NAP & cls2.id94 & TO & 14400.0 & TO & 14400.0 & 1123.8 \\
NAP & cls3.id152 & TO & 14400.0 & SAT & 11854.0 & 1463.2 \\
NAP & cls2.id102 & TO & 14400.0 & SAT & 13.1 & 13.0 \\
NAP & cls3.id149 & SAT & 0.6 & SAT & 65.9 & 65.8 \\
NAP & cls4.id180 & SAT & 517.2 & SAT & 166.7 & 166.5 \\
NAP & cls4.id220 & TO & 14400.0 & SAT & 229.6 & 229.4 \\
NAP & cls3.id162 & TO & 14400.0 & TO & 14400.0 & 1381.4 \\
NAP & cls0.id61 & TO & 14400.0 & SAT & 166.2 & 166.1 \\
NAP & cls1.id79 & SAT & 6.3 & SAT & 1074.1 & 1073.0 \\
NAP & cls3.id142 & SAT & 1976.4 & SAT & 166.9 & 166.8 \\
NAP & cls3.id163 & TO & 14400.0 & SAT & 678.3 & 678.2 \\
NAP & cls3.id136 & TO & 14400.0 & SAT & 181.9 & 181.7 \\
NAP & cls3.id120 & TO & 14400.0 & SAT & 664.2 & 664.0 \\
NAP & cls4.id208 & SAT & 7703.1 & SAT & 10.8 & 10.7 \\
NAP & cls0.id19 & SAT & 11025.8 & SAT & 96.4 & 96.2 \\
NAP & cls3.id124 & SAT & 11145.3 & SAT & 311.3 & 311.2 \\
NAP & cls4.id203 & TO & 14400.0 & TO & 14400.0 & 1077.1 \\
NAP & cls0.id42 & TO & 14400.0 & SAT & 105.6 & 105.5 \\
NAP & cls0.id36 & SAT & 1213.8 & SAT & 117.1 & 117.0 \\
NAP & cls0.id68 & TO & 14400.0 & SAT & 1.1 & 1.0 \\
NAP & cls4.id174 & SAT & 10077.6 & SAT & 9.9 & 9.8 \\
NAP & cls3.id129 & SAT & 5631.8 & SAT & 5462.4 & 1257.5 \\
NAP & cls0.id62 & TO & 14400.0 & SAT & 109.8 & 109.7 \\
NAP & cls0.id56 & SAT & 5127.9 & SAT & 96.4 & 96.2 \\
NAP & cls0.id45 & TO & 14400.0 & SAT & 120.3 & 120.2 \\
NAP & cls0.id37 & TO & 14400.0 & SAT & 16.6 & 16.5 \\
NAP & cls0.id12 & TO & 14400.0 & SAT & 26.4 & 26.2 \\
NAP & cls4.id205 & TO & 14400.0 & SAT & 319.6 & 319.5 \\
NAP & cls4.id223 & SAT & 915.8 & SAT & 1576.1 & 981.9 \\
NAP & cls3.id137 & SAT & 8360.3 & SAT & 150.3 & 150.1 \\
NAP & cls4.id233 & TO & 14400.0 & SAT & 12.2 & 12.0 \\
NAP & cls3.id130 & TO & 14400.0 & SAT & 220.9 & 220.8 \\
NAP & cls4.id190 & TO & 14400.0 & SAT & 223.1 & 222.9 \\
NAP & cls0.id59 & SAT & 2830.0 & SAT & 471.8 & 471.7 \\
NAP & cls3.id122 & SAT & 13830.6 & SAT & 649.9 & 649.7 \\
NAP & cls4.id218 & TO & 14400.0 & SAT & 230.9 & 230.7 \\
NAP & cls3.id158 & TO & 14400.0 & TO & 14400.0 & 1298.3 \\
NAP & cls3.id111 & TO & 14400.0 & SAT & 2.2 & 2.1 \\
NAP & cls0.id2 & TO & 14400.0 & SAT & 10.4 & 10.3 \\
NAP & cls0.id40 & SAT & 4928.6 & SAT & 1.0 & 0.9 \\
NAP & cls0.id48 & SAT & 7429.5 & SAT & 324.5 & 324.3 \\
NAP & cls2.id100 & TO & 14400.0 & SAT & 15.0 & 14.9 \\
NAP & cls4.id206 & TO & 14400.0 & TO & 14400.0 & 880.6 \\
NAP & cls3.id145 & TO & 14400.0 & SAT & 200.3 & 200.2 \\
NAP & cls4.id227 & TO & 14400.0 & TO & 14400.0 & 991.3 \\
NAP & cls3.id157 & TO & 14400.0 & TO & 14400.0 & 1361.0 \\
NAP & cls1.id82 & TO & 14400.0 & SAT & 192.6 & 192.4 \\
NAP & cls3.id117 & SAT & 438.0 & SAT & 206.0 & 205.8 \\
NAP & cls4.id200 & TO & 14400.0 & SAT & 214.6 & 214.4 \\
NAP & cls3.id147 & SAT & 7052.2 & SAT & 0.9 & 0.8 \\
NAP & cls0.id64 & TO & 14400.0 & SAT & 100.6 & 100.5 \\
NAP & cls4.id188 & TO & 14400.0 & TO & 14400.0 & 1034.8 \\
NAP & cls3.id125 & SAT & 0.6 & SAT & 434.5 & 434.4 \\
NAP & cls0.id65 & TO & 14400.0 & SAT & 380.5 & 380.3 \\
NAP & cls4.id202 & TO & 14400.0 & SAT & 6.6 & 6.5 \\
NAP & cls2.id105 & TO & 14400.0 & SAT & 41.0 & 40.9 \\
NAP & cls4.id176 & TO & 14400.0 & SAT & 5.8 & 5.7 \\
NAP & cls4.id199 & TO & 14400.0 & SAT & 11.9 & 11.8 \\
NAP & cls4.id201 & TO & 14400.0 & SAT & 490.2 & 490.1 \\
NAP & cls4.id234 & TO & 14400.0 & TO & 14400.0 & 1417.9 \\
NAP & cls1.id87 & SAT & 4459.2 & SAT & 682.4 & 682.2 \\
NAP & cls0.id53 & TO & 14400.0 & SAT & 88.7 & 88.5 \\
NAP & cls4.id198 & SAT & 8242.8 & SAT & 3799.6 & 1176.4 \\
NAP & cls0.id32 & TO & 14400.0 & SAT & 75.3 & 75.1 \\
NAP & cls4.id184 & TO & 14400.0 & SAT & 9.3 & 9.1 \\
NAP & cls1.id86 & SAT & 3.6 & SAT & 608.6 & 608.4 \\
NAP & cls0.id54 & TO & 14400.0 & SAT & 231.7 & 231.5 \\
NAP & cls3.id135 & TO & 14400.0 & SAT & 10164.9 & 1302.3 \\
NAP & cls4.id178 & TO & 14400.0 & SAT & 256.0 & 255.8 \\
NAP & cls3.id131 & SAT & 11623.1 & SAT & 486.5 & 486.3 \\
NAP & cls3.id113 & SAT & 3436.4 & SAT & 300.8 & 300.7 \\
NAP & cls4.id183 & TO & 14400.0 & SAT & 6.9 & 6.8 \\
NAP & cls4.id177 & TO & 14400.0 & SAT & 6.0 & 5.9 \\
NAP & cls2.id90 & TO & 14400.0 & SAT & 126.7 & 126.6 \\
NAP & cls3.id141 & TO & 14400.0 & TO & 14400.0 & 1378.2 \\
NAP & cls3.id134 & SAT & 11226.8 & SAT & 187.9 & 187.8 \\
NAP & cls4.id226 & TO & 14400.0 & SAT & 218.0 & 217.9 \\
NAP & cls4.id221 & TO & 14400.0 & TO & 14400.0 & 1198.2 \\
NAP & cls3.id114 & TO & 14400.0 & SAT & 0.9 & 0.7 \\
NAP & cls0.id9 & SAT & 690.2 & SAT & 3.6 & 3.5 \\
NAP & cls0.id25 & SAT & 5140.2 & SAT & 27.1 & 26.9 \\
NAP & cls3.id110 & SAT & 4557.6 & SAT & 72.0 & 71.9 \\
NAP & cls0.id13 & TO & 14400.0 & SAT & 24.0 & 23.9 \\
NAP & cls0.id22 & TO & 14400.0 & SAT & 16.5 & 16.4 \\
NAP & cls0.id6 & TO & 14400.0 & SAT & 9.9 & 9.7 \\
NAP & cls3.id164 & SAT & 6979.9 & SAT & 246.9 & 246.7 \\
NAP & cls3.id143 & TO & 14400.0 & SAT & 705.5 & 705.4 \\
NAP & cls4.id179 & TO & 14400.0 & SAT & 6.8 & 6.6 \\
NAP & cls0.id51 & TO & 14400.0 & SAT & 119.7 & 119.6 \\
NAP & cls3.id165 & SAT & 10156.0 & SAT & 68.5 & 68.3 \\
NAP & cls3.id150 & TO & 14400.0 & TO & 14400.0 & 1114.4 \\
NAP & cls3.id171 & TO & 14400.0 & SAT & 179.3 & 179.2 \\
NAP & cls4.id194 & TO & 14400.0 & SAT & 10.7 & 10.6 \\
NAP & cls0.id3 & SAT & 0.3 & SAT & 11.7 & 11.5 \\
NAP & cls4.id193 & TO & 14400.0 & SAT & 226.8 & 226.6 \\
NAP & cls0.id26 & SAT & 794.9 & SAT & 25.8 & 25.6 \\
NAP & cls0.id8 & TO & 14400.0 & SAT & 14.0 & 13.9 \\
NAP & cls4.id197 & TO & 14400.0 & SAT & 302.6 & 302.4 \\
NAP & cls4.id196 & SAT & 2676.1 & SAT & 60.4 & 60.3 \\
NAP & cls0.id35 & SAT & 2305.6 & SAT & 74.1 & 74.0 \\
NAP & cls0.id69 & TO & 14400.0 & SAT & 8919.0 & 1011.7 \\
NAP & cls4.id216 & TO & 14400.0 & SAT & 44.4 & 44.2 \\
NAP & cls0.id44 & SAT & 18.6 & SAT & 1.0 & 0.9 \\
NAP & cls0.id27 & TO & 14400.0 & SAT & 82.0 & 81.9 \\
NAP & cls4.id215 & TO & 14400.0 & TO & 14400.0 & 1090.8 \\
NAP & cls3.id172 & TO & 14400.0 & SAT & 73.1 & 72.9 \\
NAP & cls0.id17 & SAT & 6964.8 & SAT & 33.4 & 33.3 \\
NAP & cls2.id89 & TO & 14400.0 & SAT & 122.2 & 122.1 \\
NAP & cls0.id10 & SAT & 2181.3 & SAT & 72.1 & 72.0 \\
NAP & cls3.id107 & SAT & 335.3 & SAT & 170.3 & 170.1 \\
NAP & cls0.id49 & TO & 14400.0 & SAT & 136.0 & 135.8 \\
NAP & cls4.id217 & TO & 14400.0 & SAT & 37.0 & 36.9 \\
NAP & cls3.id108 & SAT & 1432.7 & SAT & 12.7 & 12.5 \\
NAP & cls4.id222 & TO & 14400.0 & TO & 14400.0 & 923.4 \\
NAP & cls3.id126 & SAT & 478.9 & SAT & 1403.8 & 1097.3 \\
NAP & cls0.id18 & TO & 14400.0 & SAT & 15.7 & 15.6 \\
NAP & cls3.id140 & SAT & 3109.5 & SAT & 597.7 & 597.6 \\
NAP & cls4.id210 & TO & 14400.0 & SAT & 323.2 & 323.0 \\
NAP & cls1.id83 & TO & 14400.0 & SAT & 5523.6 & 1114.2 \\
NAP & cls3.id138 & TO & 14400.0 & SAT & 12538.7 & 1052.4 \\
NAP & cls4.id214 & TO & 14400.0 & SAT & 5.9 & 5.7 \\
NAP & cls0.id73 & TO & 14400.0 & SAT & 45.2 & 45.1 \\
NAP & cls4.id213 & SAT & 2359.6 & SAT & 11.2 & 11.0 \\
NAP & cls0.id47 & TO & 14400.0 & SAT & 147.6 & 147.4 \\
NAP & cls0.id7 & SAT & 0.4 & SAT & 38.7 & 38.6 \\
NAP & cls0.id28 & TO & 14400.0 & SAT & 357.5 & 357.3 \\
NAP & cls0.id58 & SAT & 3853.7 & SAT & 29.6 & 29.4 \\
NAP & cls4.id212 & TO & 14400.0 & SAT & 10.7 & 10.6 \\
NAP & cls0.id70 & TO & 14400.0 & SAT & 62.6 & 62.4 \\
NAP & cls4.id231 & TO & 14400.0 & SAT & 243.4 & 243.3 \\
NAP & cls0.id23 & TO & 14400.0 & SAT & 1.0 & 0.9 \\
NAP & cls3.id170 & TO & 14400.0 & SAT & 216.5 & 216.3 \\
NAP & cls4.id211 & TO & 14400.0 & SAT & 0.9 & 0.8 \\
NAP & cls3.id132 & SAT & 0.5 & SAT & 69.7 & 69.6 \\
NAP & cls3.id128 & SAT & 11560.2 & SAT & 192.3 & 192.2 \\
NAP & cls0.id29 & TO & 14400.0 & SAT & 87.2 & 87.1 \\
NAP & cls1.id84 & SAT & 497.8 & SAT & 76.0 & 75.8 \\
NAP & cls0.id46 & SAT & 12.1 & SAT & 42.8 & 42.7 \\
NAP & cls0.id41 & SAT & 3723.8 & SAT & 109.8 & 109.7 \\
NAP & cls4.id228 & TO & 14400.0 & SAT & 875.4 & 875.3 \\
NAP & cls0.id43 & TO & 14400.0 & SAT & 83.6 & 83.5 \\
NAP & cls3.id167 & TO & 14400.0 & TO & 14400.0 & 1094.2 \\
NAP & cls0.id50 & SAT & 3227.2 & SAT & 378.5 & 378.4 \\
NAP & cls4.id230 & TO & 14400.0 & SAT & 321.8 & 321.7 \\
NAP & cls4.id235 & TO & 14400.0 & SAT & 43.8 & 43.6 \\
\bottomrule
\end{xltabular}
\end{center}

\end{document}